\documentclass[12pt]{article}
\usepackage[utf8]{inputenc}
\usepackage[sectionbib]{natbib}

\title{\vspace{-2cm}A Hybrid Framework Combining Autoregression and Common Factors for Matrix Time Series}
\author{Zhiyun Fan$^\dagger$, Xiaoyu Zhang$^\ddagger$, and Di Wang$^\dagger$\\$\dagger$ Shanghai Jiao Tong University and $\ddagger$ Tongji University}

\usepackage[margin=1in]{geometry}
\usepackage{graphicx}
\usepackage{subcaption}
\usepackage{mathrsfs}
\usepackage{multirow}
\usepackage{amsfonts}
\usepackage{amsmath}
\usepackage{comment}
\usepackage{amsthm}
\usepackage{color}
\usepackage{mathtools}
\usepackage{algorithm}
\usepackage{sectsty}
\usepackage[colorlinks]{hyperref}
\usepackage{amsmath}
\usepackage{amssymb}
\usepackage[toc,page]{appendix}
\usepackage[mathscr]{euscript} 
\usepackage[toc]{appendix}

\let\counterwithin\relax
\usepackage{chngcntr}
\usepackage{apptools}
\usepackage{booktabs}
\usepackage{lscape}
\usepackage{arydshln}
\usepackage{soul}
\usepackage{epstopdf}
\usepackage{algpseudocode}
\usepackage{tikz}
\usepackage{blkarray}
\usetikzlibrary{shapes.geometric, arrows, positioning}

\tikzstyle{startstop} = [ellipse,minimum width=3cm, minimum height=1.5cm, align=left, draw=black, fill=red!20]
\tikzstyle{process} = [rectangle, rounded corners,minimum width=4cm, minimum height=1.5cm, align=left, draw=black, fill=blue!20]
\tikzstyle{decision} = [diamond, minimum width=4cm, minimum height=1.5cm, align=left, draw=black, fill=green!30]
\tikzstyle{circle} = [circle, minimum size=2cm, align=left, draw=black, fill=yellow!20]
\tikzstyle{arrow} = [thick,->,>=stealth]
\tikzstyle{io} = [rectangle, rounded corners,minimum width=4cm, minimum height=1.5cm, align=left, draw=black, fill=orange!20]

\newtheorem{assumption}{Assumption}
\newtheorem{definition}{Definition}
\newtheorem{lemma}{Lemma}
\newtheorem{proposition}{Proposition}
\newtheorem{theorem}{Theorem}

\newtheorem{remark}{Remark}

\theoremstyle{definition}

\hypersetup{
    colorlinks=true,
    linkcolor=blue,
    citecolor=blue,
    filecolor=magenta,
    urlcolor=cyan,
}

\DeclareMathOperator*{\argmin}{arg\,min}

\newcommand{\bm}{\mathbf}
\newcommand{\bbm}{\boldsymbol}

\newcommand{\V}[1]{\mathrm{vec}(#1)}
\newcommand{\norm}[1]{\left\| #1 \right\|}
\newcommand{\fnorm}[1]{\norm{#1}_{\mathrm{F}}}
\newcommand{\opnorm}[1]{\norm{#1}_{\mathrm{op}}}
\newcommand{\sqmat}[1]{\left[ #1 \right]}
\newcommand{\inner}[2]{\left\langle #1, #2\right\rangle}

\mathtoolsset{showonlyrefs}

\begin{document}

\setlength{\parindent}{16pt}

\maketitle

\begin{abstract}
  Matrix-valued time series are ubiquitous in modern economics and finance, yet modeling them requires navigating a trade-off between flexibility and parsimony. We propose the Matrix Autoregressive model with Common Factors (MARCF), a unified framework for high-dimensional matrix time series that bridges the structural gap between the Matrix Autoregression (MAR) and Matrix Factor Model (MFM). While MAR typically assumes distinct predictor and response subspaces and MFM enforces identical ones, MARCF explicitly characterizes the intersection of these subspaces. By decomposing the coefficient matrices into common, predictor-specific, and response-specific components, the framework accommodates distinct input and output structures while exploiting their overlap for dimension reduction. We develop a regularized gradient descent estimator that is scalable for high-dimensional data and can efficiently handle the non-convex parameter space. Theoretical analysis establishes local linear convergence of the algorithm and statistical consistency of the estimator under high-dimensional scaling. The estimation efficiency and interpretability of the proposed methods are demonstrated through simulations and an application to global macroeconomic forecasting.
\end{abstract}

\textit{Keywords}: Matrix-valued time series, factor model, autoregression, dimension reduction.

\newpage

\setlength\abovedisplayskip{4pt}
\setlength\belowdisplayskip{4pt}

\section{Introduction}\label{sec: intro}

Matrix-valued time series, where each observation is a matrix rather than a vector, have become ubiquitous in modern economics and finance. Prominent examples include multinational macroeconomic indicators (e.g., GDP, inflation, and interest rates observed across multiple countries) and financial panel data characterized by inherent two-way dependencies. Modeling such data presents a fundamental statistical challenge: one must capture complex temporal dynamics and cross-sectional correlations while respecting the intrinsic matrix structure to maintain estimation efficiency and interpretability.

Currently, the literature is dominated by two distinct modeling paradigms, each dealing with the high-dimensionality of matrix data through different structural assumptions. The first paradigm is the Matrix Autoregressive (MAR) framework, which extends vector autoregression to the matrix domain. For instance, \citet{RongChen2021MatrixAR} introduced a bilinear MAR model, while \citet{XiaoRRMAR2024} proposed the Reduced-Rank MAR (RRMAR) model. A series of extensions have been developed to capture more complex data patterns, such as spatio-temporal MAR model \citep{Hsu2Spatio-TempMAR}, envelope-based MAR model \citep{Samadi25EMAR}, sparse MAR model \citep{jiang24regu_highdim_MAR}, and additive MAR models \citep{Zhang24additiveMAR}, among many others. The defining feature of the RRMAR framework is its flexibility: it allows the subspace containing predictive information (the predictor subspace) to be distinct from the subspace in which the future response evolves (the response subspace) \citep{huang2025supervised}. While this framework allows for complex dynamics, it ignores potential structural overlaps, leading to parameter inefficiency when strong co-movements exist.

The second paradigm is the Matrix Factor Model (MFM) \citep{wang2019factor}. These models assume that the high-dimensional observation is driven by a low-dimensional latent factor process. In the literature, two primary classes of assumptions on factors are commonly adopted. The first class assumes that factors are pervasive and influence most observed series, while allowing weak serial dynamic dependence in the error processes \citep[e.g.][]{stock2002forecasting,bai2002determining,bai2008large,chen2023statistical}. The second assumes that the latent factors capture all dynamic dependencies of the observed process, rendering the error process devoid of serial dependence \citep[e.g.][]{LamAndYao2012,wang2019factor,gao2022modeling,gao2023twotrans}. The literature has expanded to include constrained factors \citep{Chen24ConstaintMFM} and tensor-based decompositions \citep{Chang23CPMFM}. In contrast to MAR models, MFM frameworks typically enforce a strict symmetry: the loadings that map factors to the observed matrix are invariant over the lag structure. This implies that the predictor and response subspaces are identical. While highly parsimonious, this assumption is restrictive; it limits the ability of the model to capture asymmetric transmission mechanisms where the driving forces of the dynamics transform structurally over time (e.g., a shock in a specific financial sector propagating to the broad real economy).

Consequently, practitioners face a dilemma between the flexibility of MAR and the parsimony of MFM. Ideally, one would seek a hybrid approach that integrates the strengths of both paradigms. In the vector time series literature, significant efforts have been made to combine autoregressive dynamics with factor structures. A seminal example is the Factor-Augmented VAR (FAVAR) \citep{bernanke2005measuring,bai2016estimation}, which augments standard VAR systems with factors extracted from large datasets to capture broad economic movements. More recently, \citet{miao2023high} investigated high-dimensional VARs with common factors, employing low-rank constraints on the transition matrix to consistently estimate dynamics in large systems. While these approaches successfully blend factors and autoregressions, they typically treat factors as auxiliary regressors or impose low-rank structures on the vectorized process. Additionally, they do not explicitly model the relationship between the response and predictor spaces. A step in this direction was taken by \cite{wang2023commonFactor}, who proposed a common basis VAR model that identifies the intersection between predictor and response subspaces. However, all these methods are fundamentally designed for vector-valued data. Direct application to matrix-valued series requires vectorization, which demolishes the inherent two-way structural information, such as the distinct correlations among countries (rows) and economic indicators (columns), and fails to address the specific identifiability challenges imposed by the bilinear matrix structure.

In this paper, we propose the Matrix Autoregressive model with Common Factors (MARCF), a unified framework that generalizes the RRMAR specification and bridges the gap between MAR and MFM. Instead of treating the row and column spaces of the coefficient matrices as arbitrary low-rank subspaces, we explicitly model the intersection between the predictor and response subspaces. We decompose the relevant factor spaces into \textbf{common components} shared between the past and present, \textbf{predictor-specific components} unique to the past, and \textbf{response-specific components} unique to the present; see Section \ref{sec: model} for a formal introduction and interpretations of these components. Crucially, this decomposition applies simultaneously to both the row and column dimensions, allowing the model to capture not only symmetric dependencies but also complex interaction dynamics. For instance, the framework can model scenarios where dynamics are driven by a common factor in the cross-section (rows) interacting with a specific leading indicator in the variable dimension (columns). This capability allows the model to reveal asymmetric transmission channels, such as global country-level trends being predicted by specific financial variables.

This decomposition offers three key advantages. First, it creates a continuum between the existing paradigms: the model reduces to RRMAR when the intersection is empty ($d_i=0$) and approximates a structured MFM when the subspaces fully align ($d_i=r_i$). Second, by identifying the shared structure, the MARCF model achieves a dimension reduction of approximately $p_1d_1+p_2d_2$ relative to the baseline RRMAR, significantly improving forecasting performance in high-dimensional regimes. Third, it provides a transparent interpretation of the sources of variation, allowing practioners to distinguish between persistent trends, driving forces, and reactive components in a complex time series dataset.

To estimate the model, we develop a computationally efficient regularized gradient descent algorithm. Unlike existing methods that rely on alternating minimization \citep{RongChen2021MatrixAR, XiaoRRMAR2024}, our approach avoids computationally expensive matrix inversions and scales efficiently with the dimension of the data. We tackle the complex identification issues inherent in bilinear models by introducing a balanced regularization strategy, for which we establish theoretical guarantees regarding statistical consistency, algorithmic convergence, and structural dimension selection consistency.

The remainder of the paper is organized as follows. Section \ref{sec: model} introduces the MARCF framework and its relationship to RRMAR and MFM. Section \ref{sec: estimation} details the estimation strategy, including the regularized loss function, gradient descent algorithm, initialization, and rank selection. Section \ref{sec: stat theory} presents the theoretical properties of the estimator, including computational convergence, statistical error rates, and rank selection consistency. Section \ref{sec: simulation} reports simulation results, and Section \ref{sec: real data} illustrates the model performance and interpretability using a multinational macroeconomic dataset. Technical proofs and supplementary details are provided in the supplementary materials.

\section{Model}\label{sec: model}

\subsection{Matrix Factor Model and Reduced-Rank MAR}\label{subsec: RRMAR}

Matrix-valued time series models extend classical multivariate time series methods to data with inherent two-way structure. A fundamental approach in this domain is the Matrix Factor Model (MFM). The MFM assumes that the high-dimensional observed matrix $\bm{Y}_t$ is driven by a low-dimensional latent process, given by
\begin{equation}\label{eq:DFM}
    \bm{Y}_t = \bm{\Lambda}_1\bm{F}_t\bm{\Lambda}_2^\top + \bbm{\varepsilon}_t,
\end{equation}
where $\bbm{\varepsilon}_t\in\mathbb{R}^{p_1\times p_2}$ is a matrix-valued error term. As discussed in Section \ref{sec: intro}, the literature varies in its assumptions regarding the factors and error terms. In this paper, we assume that the latent factors $\bm{F}_t$ drive all temporal dynamics and $\bbm{\varepsilon}_t$ is a white noise process.

Another fundamental approach is the Matrix Autoregressive (MAR) model, which captures bilinear dependencies in an observed time series $\bm{Y}_t\in\mathbb{R}^{p_1\times p_2}$ through the recursion
\begin{equation}\label{eq:MAR1}
    \bm{Y}_t = \bm{A}_1\bm{Y}_{t-1}\bm{A}_2^\top + \bm{E}_t,\quad t=1,2,\dots,T,
\end{equation}
where each $\bm{A}_i \in \mathbb{R}^{p_i \times p_i}$ is a coefficient matrix and $\bm{E}_t$ is a matrix-valued white noise process. To address the high-dimensional challenge of the parameter space, the Reduced-Rank MAR (RRMAR) model \citep{XiaoRRMAR2024} imposes low-rank constraints on the coefficient matrices, requiring $\text{rank}(\bm{A}_i)=r_i\ll p_i$ for $i=1,2$. We parameterize this constraint using the singular value decomposition (SVD) $\bm{A}_i=\bm{U}_i\bm{S}_i\bm{V}_i^\top$, where $\bm{U}_i\in\mathbb{R}^{p_i\times r_i}$ and $\bm{V}_i\in\mathbb{R}^{p_i\times r_i}$ have orthonormal columns, and $\bm{S}_i$ is a diagonal matrix of singular values. Although the SVD components are not unique, the column spaces of $\bm{U}_i$ and $\bm{V}_i$, as well as the corresponding subspace projection matrices $\bm{U}_i\bm{U}_i^\top$ and $\bm{V}_i\bm{V}_i^\top$, are uniquely defined.

The RRMAR and MFM frameworks are closely linked through their subspace representations. First, the low-rank structure of the RRMAR allows for a supervised factor analysis interpretation \citep{wang2022high,wang2023commonFactor,huang2025supervised}. Specifically, by pre-multiplying \eqref{eq:MAR1} by $\bm{U}_1^\top$ and post-multiplying by $\bm{U}_2$, we obtain the following low-dimensional representation of the dynamics:
\begin{equation}\label{eq:low-dimensional-AR}
    \bm{U}_1^\top\bm{Y}_t\bm{U}_2 = \bm{S}_1(\bm{V}_1^\top\bm{Y}_{t-1}\bm{V}_2)\bm{S}_2 + \bm{U}_1^\top\bm{E}_t\bm{U}_2.
\end{equation}
Following \citet{huang2025supervised}, we view $\bm{U}_1^\top\bm{Y}_t\bm{U}_2$ and $\bm{V}_1^\top\bm{Y}_t\bm{V}_2$ as the \textit{response factor} and $\textit{predictor factor}$. The term ``response factor'' is justified because the RRMAR model implies
\begin{equation}
    \bm{Y}_t = \bm{U}_1(\bm{U}_1^\top\bm{Y}_t\bm{U}_2)\bm{U}_2^\top + (\bm{E}_t - \bm{U}_1\bm{U}_1^\top\bm{E}_t\bm{U}_2\bm{U}_2^\top) =: \bm{U}_1(\bm{U}_1^\top\bm{Y}_t\bm{U}_2)\bm{U}_2^\top + \widetilde{\bm{E}}_t.
\end{equation}
This takes the form of an MFM as in \eqref{eq:DFM} where the low-dimensional transformation $\bm{U}_1^\top\bm{Y}_t\bm{U}_2$ summarizes all predictable information in $\bm{Y}_t$. On the other hand, the relationship in \eqref{eq:low-dimensional-AR} implies that the predictor factor $\bm{V}_1^\top\bm{Y}_{t-1}\bm{V}_2$ contains all driving forces in the historical data necessary to predict the future response. 

Consider now the case where the latent factor $\bm{F}_t$ in MFM follows a MAR process, as proposed by \citet{yu2024dynamic}:
\begin{equation}\label{eq:factor_AR}
    \bm{F}_t = \bm{B}_1\bm{F}_{t-1}\bm{B}_2^\top + \bbm{\xi}_t,
\end{equation}
where $\bm{B}_i\in\mathbb{R}^{r_i\times r_i}$ are coefficient matrices and $\bbm{\xi}_t$ is low-dimensional white noise. Substituting \eqref{eq:factor_AR} into \eqref{eq:DFM}, the structured dynamics become
\begin{equation}\label{eq:MAR_in_DFM}
    \bm{Y}_t = (\bm{\Lambda}_1\bm{B}_1\bm{\Lambda}_1)\bm{Y}_{t-1}(\bm{\Lambda}_2\bm{B}_2^\top\bm{\Lambda}_2^\top) + \bm{N}_t,
\end{equation}
where $\bm{N}_t$ is a composite noise term involving $\bbm{\xi}_t$ and a moving average of $\bm{E}_t$. Equation \eqref{eq:MAR_in_DFM} reveals the implicit MAR structure of this dynamic MFM, with the effective coefficient matrices $\bm{\Lambda}_i\bm{B}_i\bm{\Lambda}_i^\top$. Comparing this to the RRMAR decomposition in \eqref{eq:low-dimensional-AR}, we observe a fundamental restriction in the MFM framework: the predictor subspace and the response subspace are identical (both spanned by $\bm{\Lambda}_i$). While this symmetry is parsimonious, it prevents the model from capturing rotational dynamics where the relevant subspaces shift between the predictor and the response. In contrast, the RRMAR model allows $\bm{V}_i$ and $\bm{U}_i$ to be distinct, offering greater flexibility but potentially ignoring structural commonalities shared by the response and predictor.

\subsection{Matrix Autoregression with Common Factors}

To reconcile the distinct subspace assumption of the RRMAR model with the shared subspace assumption of the MFM, we propose a hybrid framework. This framework explicitly parameterizes the intersection of the predictor and response subspaces. By controlling the degree of this overlap, we can capture shared dynamics while accommodating distinct structural features in the response and predictor spaces.

Let $\mathcal{M}(\bm{U}_i)$ and $\mathcal{M}(\bm{V}_i)$ denote the response and predictor subspaces, respectively. We assume their intersection has dimension $d_i$, where $0\leq d_i\leq r_i$. The parameters $d_1$ and $d_2$ quantify the extent of structural overlap. Given these dimensions, there exist orthogonal matrices $\bm{O}_{i1}$ and $\bm{O}_{i2}$ allowing the decomposition:
\begin{equation}
    \bm{U}_i\bm{O}_{i1} = [\bm{C}_i~\bm{R}_i]\quad\text{and} \quad \bm{V}_i\bm{O}_{i2} = [\bm{C}_i~\bm{P}_i].
\end{equation}
Here, $\bm{C}_i\in\mathbb{R}^{p_i\times d_i}$ represents the \textit{common subspaces} shared by both responses and predictors. The matrix $\bm{R}_i\in\mathbb{R}^{p_i\times (r_i-d_i)}$ spans the \textit{response-specific subspace}, while $\bm{P}_i\in\mathbb{R}^{p_i\times (r_i-d_i)}$ spans the \textit{predictor-specific subspace}.

Substituting this decomposition into the general bilinear framework yields the \textit{Matrix Autoregressive model with Common Factors} (MARCF):
\begin{equation}\label{eq:MARCF}
    \bm{Y}_t = [\bm{C}_1~\bm{R}_1]\bm{D}_1[\bm{C}_1~\bm{P}_1]^\top\bm{Y}_{t-1}[\bm{C{}}_2~\bm{P}_2]\bm{D}_2^\top[\bm{C}_2~\bm{R}_2]^\top + \bm{E}_t.
\end{equation}
In this formulation, $\bm{D}_i=\bm{O}_{i1}^\top\bm{S}_i\bm{O}_{i2}$ serves as a dense mixing matrix within the low-dimensional latent space. 

This specification provides a unified continuum between the two existing paradigms. When $d_1=d_2=0$, the common component vanishes, and the model reduces to the RRMAR framework with disjoint subspaces. Conversely, when $d_i=r_i$, the specific components vanish, and the predictor and response subspaces align perfectly, approximating the structure of a dynamic MFM. By estimating $d_1$ and $d_2$ from the data, the MARCF model adaptively balances the efficiency of shared factors with the flexibility of rotational dynamics.

For interpretation, we let $\bm{C}_i^\top\bm{C}_i=\bm{I}_{d_i}$, $\bm{R}_i^\top\bm{R}_i=\bm{P}_i^\top\bm{P}_i = \bm{I}_{r_i - d_i}$, and $\bm{C}_i^\top\bm{R}_i = \bm{C}_i^\top\bm{P}_i = \bm{0}_{d_i\times(r_i-d_i)}$; while the detailed identification conditions are formally presented in Section \ref{sec:2.4}. By pre-multiplying \eqref{eq:MARCF} by $[\bm{C}_1~\bm{R}_1]^\top$ and post-multiplying by $[\bm{C}_2~\bm{R}_2]$, we have
\begin{equation}\label{eq:MARCF-2}
\begin{bmatrix}
    \bm C_1^\top\bm Y_t\bm C_2 & \bm C_1^\top\bm Y_t\bm R_2\\
    \bm R_1^\top\bm Y_t\bm C_2 & \bm R_1^\top\bm Y_t\bm R_2
  \end{bmatrix}
  =\bm D_1\begin{bmatrix}
    \bm C_1^\top\bm Y_{t-1}\bm C_2 & \bm C_1^\top\bm Y_{t-1}\bm P_2\\
    \bm P_1^\top\bm Y_{t-1}\bm C_2 & \bm P_1^\top\bm Y_{t-1}\bm P_2
  \end{bmatrix}\bm D_2^\top
  +\text{Noise}.
\end{equation}
This representation extends the factor interpretation in \eqref{eq:low-dimensional-AR} by explicitly characterizing the \textit{common factors} shared by responses and predictors, alongside their specific counterparts.

Furthermore, when $d_1$ and $d_2$ are nonzero, the proposed framework achieves significant dimension reduction. The total number of free parameters in the MARCF model is given by
\begin{equation}\label{eq:df_MARCF}
    \mathrm{df}_{\mathrm{MARCF}} = p_1(2r_1 - d_1) + p_2(2r_2 - d_2) + r_1^2 + r_2^2.
\end{equation}
This contrasts with $\mathrm{df}_{\mathrm{RRMAR}} = 2p_1 r_1 + 2p_2 r_2 + r_1^2 + r_2^2$ for the standard reduced-rank model. In high-dimensional settings, the MARCF model reduces the effective dimension of parameter space by approximately $p_1 d_1 + p_2 d_2$, enhancing both estimation efficiency and forecasting accuracy.

\subsection{Model Interpretation}\label{sec:model_interpretation}

The MARCF framework offers a transparent decomposition of the sources of variation, enabling an innovative analysis of dynamic mechanisms. To illustrate, consider the multinational macroeconomic application discussed in Section \ref{sec: intro}, where rows represent countries and columns represent economic indicators.

\paragraph{(1) Row Dynamics (Country-Level Heterogeneity):}
The row decomposition identifies groups of countries with distinct roles in the global system.
\begin{itemize}
    \item \textbf{Common Countries ($\bm{C}_1$):} This identifies systemic economies, such as the G7 nations. These countries exhibit \textit{persistent influence}: their economic state at time $t-1$ strongly predicts the global state at time $t$, and they remain the primary drivers of variance in the response at time $t$.
    \item \textbf{Predictor-Specific Countries ($\bm{P}_1$):} This subspace captures ``leading'' economies. For instance, it may isolate emerging markets or commodity exporters. While these nations might not dominate total global variance in response, shocks to their economies, such as a sudden shift in manufacturing output, often serve as early warning signals, transmitting information that predicts future movements in the core economies.
    \item \textbf{Response-Specific Countries ($\bm{R}_1$):} This captures ``reactive'' economies, possibly small open economies. These countries absorb global shocks and exhibit volatility in the response period, but their idiosyncratic fluctuations have little predictive power for the future global state.
\end{itemize}

\paragraph{(2) Column Dynamics (Variable-Level Lead-Lag Relationships):}
The column decomposition separates indicators based on their temporal properties.
\begin{itemize}
    \item \textbf{Common Variables ($\bm{C}_2$):} These represent fundamental macroeconomic aggregates like industrial production or GDP. These variables are auto-correlated and central to describing the economic cycle; they predict themselves and are the primary targets of prediction.
    \item \textbf{Predictor-Specific Variables ($\bm{P}_2$):} These correspond to financial variables or sentiment indices (e.g., stock returns, PMI). These are classic leading indicators; they contain rich information about the future direction of the economy but may be noisy or less relevant for defining the real economic state at the current moment.
    \item \textbf{Response-Specific Variables ($\bm{R}_2$):} These capture lagging indicators, such as unemployment rates. These variables react to past economic conditions but are slow to adjust and do not necessarily drive future cycles.
\end{itemize}

\paragraph{(3) Interaction Dynamics (Cross-Sectional Spillovers):}
Beyond the diagonal blocks, the MARCF model captures crucial cross-terms (e.g., corresponding to $\bm{C}_1^\top \bm{Y}_{t-1} \bm{P}_2$). These interactions allow for asymmetric dependencies where the source of predictive information differs structurally from the response it generates. For instance, the model can describe how \textit{predictor-specific variables} (e.g., leading indicators, $\bm{P}_2$) in \textit{common countries} ($\bm{C}_1$) drive future \textit{common real output} ($\bm{C}_2$). This capability allows the model to capture dynamic transformations where the factors driving prediction differ from those characterizing the response, contrasting with standard factor models that enforce a restrictive symmetry.

\subsection{Stationarity and Identification}\label{sec:2.4}

The MARCF model defined in \eqref{eq:MARCF} belongs to the general class of bilinear MAR processes. Following \citet{RongChen2021MatrixAR}, weak stationarity is guaranteed provided that the spectral radii of the coefficient matrices satisfy $\rho(\bm{A}_1)\rho(\bm{A}_2) < 1$, where $\bm{A}_i = [\bm{C}_i~\bm{R}_i]\bm{D}_i[\bm{C}_i~\bm{P}_i]^\top$ for $i=1,2$.

The model structure introduces identification challenges at two distinct levels. First, the coefficient matrices $\bm{A}_1$ and $\bm{A}_2$ are subject to global scaling indeterminacy: for any nonzero scalar $c$, the pairs $(\bm{A}_1, \bm{A}_2)$ and $(c \bm{A}_1, c^{-1} \bm{A}_2)$ yield identical $\bm{A}=\bm{A}_2\otimes\bm{A}_1$. In other words, the Kronecker product $\bm{A}$ is identifiable, but the row and column coefficient matrices are not. To resolve this, we adopt the constraint $\|\bm{A}_1\|_\mathrm{F} = \|\bm{A}_2\|_\mathrm{F}$. We define the identifiable true value of $\bm{A}$ as $\bm{A}^*$, and define the signal strength of the true model as $\phi=\|\bm{A}^*\|_\text{F}^{1/2}$ such that $\|\bm{A}_1^*\|_\text{F}=\|\bm{A}_2^*\|_\text{F} = \phi$.

Second, given a fixed $\bm{A}_i$, the decomposition into components $(\bm{C}_i,\bm{R}_i,\bm{P}_i,\bm{D}_i)$ is not unique. Specifically, for any invertible matrices $\bm{Q}_1\in\mathbb{R}^{d_i\times d_i}$ and $\bm{Q}_2,\bm{Q}_3\in\mathbb{R}^{(r_i-d_i)\times(r_i-d_i)}$, the parameter sets
\begin{equation}
    (\bm{C}_i,\bm{R}_i,\bm{P}_i,\bm{D}_i)\quad\text{and}\quad(\bm{C}_i\bm{Q}_1,\bm{R}_i\bm{Q}_2,\bm{P}_i\bm{Q}_3,\text{diag}(\bm{Q}_1^{-1},\bm{Q}_2^{-1})\bm{D}_i\text{diag}(\bm{Q}_1^{-\top},\bm{Q}_3^{-\top}))
\end{equation}
are equivalent. To resolve this internal indeterminacy and balance these components, we impose the following constraints:
\begin{equation}\label{eq:identification}
    \bm{C}_i^\top \bm{C}_i = b^2 \bm{I}_{d_i}, \quad \bm{R}_i^\top \bm{R}_i = \bm{P}_i^\top \bm{P}_i = b^2 \bm{I}_{r_i - d_i}, \quad\text{and}\quad \bm{C}_i^\top \bm{R}_i = \bm{C}_i^\top \bm{P}_i = \bm{0}_{d_i \times (r_i - d_i)},
\end{equation}
where $b>0$ is a balancing scalar. To ensure that $\bm{D}_i$ is of the same order of magnitude as $(\bm{C}_i,\bm{R}_i,\bm{P}_i)$, we set $b\asymp\phi^{1/3}$. This choice implies that
$\|\bm{C}_i\|_\text{F} \asymp \|\bm{R}_i\|_\text{F} \asymp \|\bm{P}_i\|_\text{F} \asymp \|\bm{D}_i\|_\text{F}$ when $r_i$ and $d_i$ remain constant. 

Under the balancing constraint in \eqref{eq:identification}, the components are determined uniquely up to orthogonal rotations. We treat parameterizations that differ only by such rotations as equivalent. Let $\bm{\Theta}$ denote the full set of component matrices. For any two sets $\bm{\Theta}$ and $\bm{\Theta}'$ satisfying the identification conditions in \eqref{eq:identification}, we define their distance as
\begin{equation}\label{def: error}
    \begin{split}    
        \mathrm{dist}(\bm{\Theta}, \bm{\Theta}')^2 = \min_{\substack{\bm{Q}_1 \in \mathbb{O}^{d_i}, \\ \bm{Q}_2, \bm{Q}_3 \in \mathbb{O}^{r_i - d_i}}} & 
        \Big\{
        \|\bm{C}_i - \bm{C}_i' \bm{Q}_1\|_\mathrm{F}^2 + \|\bm{R}_i - \bm{R}_i' \bm{Q}_2\|_\mathrm{F}^2 + \|\bm{P}_i - \bm{P}_i' \bm{Q}_3\|_\mathrm{F}^2  \\ 
        & \quad + \|\bm{D}_i - \text{diag}(\bm{Q}_1, \bm{Q}_2)^\top \bm{D}_i' \text{diag}(\bm{Q}_1, \bm{Q}_3)\|_\mathrm{F}^2
        \Big\},
    \end{split}
\end{equation}
where $\mathbb{O}^k$ denotes the set of $k\times k$ orthogonal matrices. This metric remains invariant under rotations and forms the basis of our theoretical analysis in Section \ref{sec: stat theory}. 

\begin{remark}[Interpretation vs. Estimation]
    The parameterization with $b \asymp \phi^{1/3}$ is designed to facilitate theoretical analysis and numerical optimization. However, for model interpretation as discussed in Section \ref{sec:model_interpretation}, it is often more intuitive to work with orthonormal bases. One can always transform the estimated parameters to an equivalent representation where $b=1$ (i.e., $\bm{C}_i^\top \bm{C}_i = \bm{I}_{d_i}$) by absorbing the scaling factors into $\bm{D}_i$. The interpretation of the subspaces and dynamic factors in Section \ref{sec:model_interpretation} remains invariant under this rescaling.
\end{remark}

\section{Estimation and Modeling Procedure}\label{sec: estimation}

This section presents the estimation methodology for the MARCF model \eqref{eq:MARCF}. Given the high-dimensional and non-convex nature of the parameter space, we employ a computationally efficient strategy that combines spectral initialization with a regularized gradient descent algorithm. We first define the optimization objective, which incorporates specific regularization terms to enforce the balanced identification constraints discussed in Section \ref{sec:2.4}. Subsequently, we describe an initialization procedure that exploits the subspace structures of the auxiliary RRMAR model. Finally, we propose a data-driven procedure for sequentially selecting the model ranks $(r_1, r_2)$ and the common subspace dimensions $(d_1, d_2)$.

\subsection{Regularized Gradient Descent Estimation}\label{subsec: estimation}

Consider the observed matrix-valued time series $\{\bm{Y}_t\}_{t=0}^T$. For a fixed set of structural dimensions $(r_1, r_2, d_1, d_2)$, the model parameters are collected in $\bm{\Theta}=(\{\bm{C}_i,\bm{R}_i,\bm{P}_i,\bm{D}_i\}_{i=1}^2)$. Recall that the implied coefficient matrices are given by
\begin{equation}\label{eq:decomposition}
    \bm{A}_i=[\bm{C}_i~\bm{R}_i]\bm{D}_i[\bm{C}_i~\bm{P}_i]^\top,\quad \text{for }i=1,2.
\end{equation}
Our estimation procedure minimizes the least-squares loss function:
\begin{equation}\label{eq: LS loss}
    \mathcal{L}(\bm{\Theta})=\frac{1}{2T}\sum_{t=1}^T\left\|\bm{Y}_t-[\bm{C}_1~\bm{R}_1]\bm{D}_1[\bm{C}_1~\bm{P}_1]^\top\bm{Y}_{t-1}[\bm{C}_2~\bm{P}_2]\bm{D}_2^\top[\bm{C}_2~\bm{R}_2]^\top\right\|_\text{F}^2.
\end{equation}

To ensure estimation stability and satisfy the identification conditions outlined in Section \ref{sec:2.4}, we incorporate two types of regularization. First, to resolve the global scaling indeterminacy between $\bm{A}_1$ and $\bm{A}_2$, we introduce a regularization term that encourages their Frobenius norms to be equal:
\begin{equation}
    \mathcal{R}_1(\bm{\Theta}) = (\|[\bm{C}_1~\bm{R}_1]\bm{D}_1[\bm{C}_1~\bm{P}_1]^\top\|_\text{F}^2 - \|[\bm{C}_2~\bm{P}_2]\bm{D}_2^\top[\bm{C}_2~\bm{R}_2]^\top\|_\text{F}^2)^2 = (\|\bm{A}_1\|_\text{F}^2 - \|\bm{A}_2\|_\text{F}^2)^2.
\end{equation}
This regularization term, consistent with the related work on reduced-rank models \citep{Tu2016,Guquanquan2017}, helps stabilize the estimation of these two coefficient matrices. Second, for a given balance scalar $b$, inspired by \citet{Han2022} and \citet{wang2023commonFactor}, we introduce a second regularization term to enforce the identification constraints on the components in \eqref{eq:identification}
\begin{equation}
    \mathcal{R}_2(\bbm \Theta;b):=\sum_{i=1}^2 \bigg(\|[\bm{C}_i ~ \bm{R}_i]^\top[\bm{C}_i ~ \bm{R}_i] - b^2\bm{I}_{r_i}\|_\text{F}^2 + \|[\bm{C}_i ~ \bm{P}_i]^\top[\bm{C}_i ~ \bm{P}_i] - b^2\bm{I}_{r_i}\|_\text{F}^2 \bigg).
\end{equation}
The full regularized objective function is
\begin{equation}\label{eq: loss}
\overline{\mathcal{L}}(\bm{\Theta};\lambda_1,\lambda_2,b)=\mathcal{L}(\bm{\Theta})+\frac{\lambda_1}{4}\mathcal{R}_1(\bbm\Theta)+\frac{\lambda_2}{2}\mathcal{R}_2(\bbm \Theta;b),
\end{equation}
where $\lambda_1$ and $\lambda_2$ are tuning parameters.

The model parameters are estimated by minimizing $\overline{\mathcal{L}}$ using a standard gradient descent algorithm, as outlined in Algorithm \ref{algo: GD}. Here, $\nabla_{\bm{M}}\overline{\mathcal{L}}^{(j)}$ denotes the partial gradient of $\overline{\mathcal{L}}$ with respect to any component $\bm{M}\in\{\bm{C}_i,\bm{R}_i,\bm{P}_i,\bm{D}_i\}_{i=1}^2$, evaluated at the parameters $\bbm\Theta^{(j)}$ after $j$ iterations. The partial gradients with respect to $\bm{C}_i$, $\bm{R}_i$, $\bm{P}_i$, and $\bm{D}_i$ are derived from the chain rule and are provided in detail in Appendix S2. Crucially, all update steps involve only matrix multiplications and additions, avoiding expensive matrix inversions or SVDs at each iteration, making the method highly scalable to high-dimensional data.
\begin{algorithm}
    \caption{Gradient Decent Algorithm for MARCF}
    \label{algo: GD}
    \begin{algorithmic}[1]
        \State \textbf{input:} data $\{\bm Y_t\}_{t=0}^T$, initial values $\bbm\Theta^{(0)}$, $r_1,r_2,d_1,d_2$, hyperparameters parameters $\lambda_1,\lambda_2,b$, step size $\eta$, and max iteration $J$.
        \For{$j = 0$ to $J-1$}
            \For{$i = 1$ to 2}
                
                \State $\bm C_i^{(j+1)} = \bm C_i^{(j)}-\eta\nabla_{\bm C_i}\overline{\mathcal{L}}^{(j)} $,\quad$\bm R_i^{(j+1)} = \bm R_i^{(j)}-\eta\nabla_{\bm R_i}\overline{\mathcal{L}}^{(j)},$
                \State $\bm P_i^{(j+1)} = \bm P_i^{(j)}-\eta\nabla_{\bm P_i}\overline{\mathcal{L}}^{(j)} $,\quad$\bm D_i^{(j+1)} = \bm D_i^{(j)}-\eta\nabla_{\bm D_i}\overline{\mathcal{L}}^{(j)}.$
            \EndFor
        \EndFor
        \State \textbf{output:} $\widehat{\bbm\Theta}=(\{\bm C_i^{(J)},\bm R_i^{(J)},\bm P_i^{(J)},\bm D_i^{(J)}\}_{i=1}^2)$.
    \end{algorithmic}
\end{algorithm}

\subsection{Algorithm Initialization}\label{subsec: initial}

Since the optimization landscape of \eqref{eq: loss} is nonconvex, the global convergence of Algorithm \ref{algo: GD} depends critically on high-quality initialization. We propose a spectral initialization strategy that leverages the structure of the simpler RRMAR model. The core idea is to first estimate the aggregate predictor and response subspaces, and then extract their intersection to initialize the common components.

Given the dimensions $(r_1, r_2, d_1, d_2)$ and scalar $b$, the initialization proceeds in two steps:

\begin{enumerate}
    \item First obtain a consistent estimate of the total subspaces by fitting the RRMAR model:
    \begin{equation}\label{eq: RR.LS}
        \widehat{\bm A}_1^{\text{RR}}, \widehat{\bm A}_2^{\text{RR}}:=\argmin_{\mathrm{rank}(\bm A_i)\leq r_i, i=1,2}\frac{1}{2T}\sum_{i=1}^{T}\fnorm{\bm Y_t-\bm A_1\bm Y_{t-1}\bm A_2^\top}^2.
    \end{equation}
    This can be solved using the ALS method \citep{XiaoRRMAR2024}. Let $\widehat{\bm A}_i^\mathrm{RR}$ be rescaled such that $\|\widehat{\bm A}_1^\mathrm{RR}\|_\text{F} = \|\widehat{\bm A}_2^\mathrm{RR}\|_\text{F}$. We compute the rank-$r_i$ SVD: $\widehat{\bm A}_i^\mathrm{RR} = \widehat{\bm U}_i \widehat{\bm S}_i \widehat{\bm V}_i^\top$. Here, $\widehat{\bm U}_i$ spans the sum of the common and response-specific subspaces, while $\widehat{\bm V}_i$ spans the sum of the common and predictor-specific subspaces.

    \item Use spectral methods to decompose the subspaces of $\widehat{\bm A}_1^\mathrm{RR}$ and $\widehat{\bm A}_2^\mathrm{RR}$ to obtain $\bbm\Theta^{(0)}$. We first discuss the rationale of the spectral method. Denote $\bm U_i:=[\bm C_i~\bm R_i]$ and $\bm V_i:=[\bm C_i~\bm P_i]$ with orthonormal columns. For any matrix $\bm{M}$ with orthonormal columns, denote $\mathcal{P}_{\bm{M}}:=\bm{M}\bm{M}^\top$ as the subspace projection matrix and $\mathcal{P}_{\bm{M}}^\perp:=\bm{I}_p-\bm{M}\bm{M}^\top$ as the subspace orthogonal complement projector. Straightforward algebra shows that $\mathcal{P}_{\bm U_i}\mathcal{P}_{\bm V_i}^\perp=\mathcal{P}_{\bm R_i}\mathcal{P}_{\bm P_i}^\perp$. This implies that $\mathcal{M}(\bm R_i)$ is spanned by the first $r_i-d_i$ left singular vectors of $\mathcal{P}_{\bm U_i}\mathcal{P}_{\bm V_i}^\perp$. A similar argument holds for $\mathcal{M}(\bm P_i)$. Additionally,
    $\mathcal{P}_{\bm R_i}^\perp \,\mathcal{P}_{\bm P_i}^\perp
    \bigl(\mathcal{P}_{\bm U_i} + \mathcal{P}_{\bm V_i}\bigr)
    \mathcal{P}_{\bm R_i}^\perp \,\mathcal{P}_{\bm P_i}^\perp
    \;=\;2\,\mathcal{P}_{\bm C_i},$
    implying that $\mathcal{M}(\bm C_i)$ is spanned by the leading $d_i$ eigenvectors of the left-hand side. Based on these, we obtain the initial values as follows:
    \begin{enumerate}
        \item For each $i=1,2$, using the SVD components in the first step, calculate the top $r_i-d_i$ left singular vectors of $\mathcal{P}_{\widehat{\bm U}_i}\mathcal{P}_{\widehat{\bm V}_i}^\perp$ and $\mathcal{P}_{\widehat{\bm V}_i}\mathcal{P}_{\widehat{\bm U}_i}^\perp$, and denote them by $\widetilde{\mathbf{R}}_i$ and $\widetilde{\mathbf{P}}_i$, respectively. Calculate the top $d_i$ eigenvectors of 
        $\mathcal{P}_{\widetilde{\bm R}_i}^\perp\mathcal{P}_{\widetilde{\bm P}_i}^\perp\left(\mathcal{P}_{\widehat{\bm U}_i}+\mathcal{P}_{\widehat{\bm V}_i}\right)\mathcal{P}_{\widetilde{\bm R}_i}^\perp\mathcal{P}_{\widetilde{\bm P}_i}^\perp,$
        and denote them by $\widetilde{\mathbf{C}}_i$. Then calculate $\widetilde{\mathbf{D}}_1=[\widetilde{\mathbf{C}}_1 ~\widetilde{\mathbf{R}}_1]^{\top}~ \widehat{\mathbf{A}}_1^\text{RR}[\widetilde{\mathbf{C}}_1~\widetilde{\mathbf{P}}_1]$ and $\widetilde{\mathbf{D}}_2=[\widetilde{\mathbf{C}}_2~ \widetilde{\mathbf{R}}_2]^{\top}~ \widehat{\mathbf{A}}_2^\text{RR}[\widetilde{\mathbf{C}}_2~\widetilde{\mathbf{P}}_2]$.
        
        \item For each $i=1,2$, set $\mathbf{C}_i^{(0)}=b \widetilde{\mathbf{C}}_i, \mathbf{R}_i^{(0)}=b \widetilde{\mathbf{R}}_i, \mathbf{P}_i^{(0)}=b \widetilde{\mathbf{P}}_i$, and $\mathbf{D}_i^{(0)}=b^{-2} \widetilde{\mathbf{D}}_i$.
    \end{enumerate}
\end{enumerate}

\subsection{Selection of Ranks and Common Dimensions}\label{subsec: selection}

Although our estimation methodology is applicable to any pre-specified $(r_1,r_2,d_1,d_2)$, the appropriate selection of these parameters is critical in practice. These structural parameters define the model's complexity and determine the interpretability of the factors. Since the true ranks are unknown in real-world applications, we propose a two-step procedure to determine $(r_1,r_2)$ and $(d_1,d_2)$ sequentially.

For rank selection, since $r_1$ and $r_2$ are typically much smaller than $p_1$ and $p_2$, we choose two upper bounds for the selection: $\bar{r}_1\ll p_1$ and $\bar{r}_2\ll p_2$. Then, we estimate the RRMAR model \eqref{eq: RR.LS} with ranks $\bar{r}_1$ and $\bar{r}_2$. Denote the rescaled solutions by $\widehat{\bm A}_1^\mathrm{RR}(\bar{r}_1)$ and $\widehat{\bm A}_2^\mathrm{RR}(\bar{r}_2)$, respectively. Let $\widehat{\sigma}_{i,1}\geq \widehat{\sigma}_{i,2}\geq ...\geq \widehat{\sigma}_{i,\bar{r}_i}$ be the singular values of $\widehat{\bm A}_i^\mathrm{RR}(\bar{r}_i)$. Motivated by \cite{Xia2015SelectR}, we introduce a ridge-type ratio to determine $r_1$ and $r_2$ separately:
$$
\widehat{r}_i=\underset{1 \leqslant j < \bar{r}_i}{\arg \min } \frac{\widehat{\sigma}_{i,j+1}+s(p_1,p_2, T)}{\widehat{\sigma}_{i,j}+s(p_1,p_2, T)},
$$
where we suggest $s(p_1,p_2,T)=\sqrt{(p_1+p_2)\log(T)/(20T)}$. The theoretical guarantee of this estimator is given in Section \ref{sec:consistency_of_rank_selection}, and its empirical performance is justified by the numerical experiments in Section \ref{sec: simulation}. Once $s(p_1,p_2,T)$ is properly specified and $\bar{r}_i> r_i$, the method is not sensitive to the choice to $\bar{r}_i$. Therefore, in practice, if $p_i$ is not too large and the computational cost is affordable, $\bar{r}_i$ can be chosen largely or even $\bar{r}_i=p_i$.

After $r_1$ and $r_2$ are determined, the problem reduces to selecting a model in low dimensions. We use the BIC criterion to select $d_1$ and $d_2$. Let $\widehat{\bm A}(\widehat{r}_1,\widehat{r}_2,d_1,d_2)$ be the estimator of $\bm A_2\otimes \bm A_1$ in \eqref{eq:decomposition} under $\widehat{r}_i$ and $d_i$. Then, 
$$
    \mathrm{BIC}(d_1,d_2) = T p_1p_2 \log \left(\sum_{t=1}^T\|\text{vec}(\bm{Y}_t)-\widehat{\bm A}(\widehat{r}_1,\widehat{r}_2,d_1,d_2) \text{vec}(\bm{Y}_{t-1})\|_2^2\right)+\textup{df}_{\mathrm{MARCF}}\log (T),
$$
where $\textup{df}_{\mathrm{MARCF}}$ is defined in \eqref{eq:df_MARCF} with $\widehat{r}_1$ and $\widehat{r}_2$. Then, we choose $d_1$ and $d_2$ as:
$$
    \widehat{d}_1,\widehat{d}_2=\min_{1\leq d_i\leq \widehat{r}_i, i=1,2}\mathrm{BIC}(d_1,d_2).
$$
\begin{remark}
    While simultaneous selection of $r_i$ and $d_i$ via BIC is theoretically possible, the search space of size $O(\bar{r}_1^2 \bar{r}_2^2)$ is computationally prohibitive. The proposed sequential approach balances statistical accuracy with computational feasibility. Rolling window cross-validation may also be used for hyperparameter tuning but incurs significantly higher computational costs.
\end{remark}

\section{Theory}\label{sec: stat theory}

In this section, we investigate the theoretical properties of the proposed estimation methodology for the MARCF model. To facilitate a clear understanding of the results, we outline a roadmap that decouples the optimization landscape from the stochastic nature of the data.

First, in Section \ref{sec:4.1_compuConvergence}, we establish the computational convergence of the algorithm via a deterministic analysis. Theorem \ref{theorem: computational convergence} demonstrates that Algorithm \ref{algo: GD} achieves local linear convergence, conditioned on a valid initialization and the local curvature of the loss function (characterized by Restricted Strong Convexity and Smoothness). Subsequently, in Section \ref{sec:4.2_statRates}, we provide the stochastic analysis under mild assumptions on the data-generating process. We verify that the conditions required for computational convergence hold with high probability and derive the final statistical error rates for the estimator. Finally, Section \ref{sec:consistency_of_rank_selection} proves the consistency of the rank selection procedure.

\subsection{Computational Convergence Analysis}\label{sec:4.1_compuConvergence}

We first analyze the convergence properties of Algorithm \ref{algo: GD}. Due to the nonconvex nature of the objective function $\overline{\mathcal{L}}$, our analysis relies on the assumptions of restricted strong convexity (RSC) and restricted strong smoothness (RSS), which characterize the curvature of the loss function along the directions of relevant low-rank matrices. These properties are defined with respect to the Kronecker product parameter matrix $\bm A=\bm A_2\otimes \bm A_1$. Let $\widetilde{\mathcal{L}}(\bm A)$ denote the least-squares loss function with respect to $\bm A$:
\begin{equation}
    \label{eq: LS loss2}
    \widetilde{\mathcal{L}}(\bm{A})=\frac{1}{2T}\sum_{t=1}^T\left\|\text{vec}(\bm{Y}_t)-\bm{A}\text{vec}(\bm{Y}_{t-1})\right\|_2^2.
\end{equation}

\begin{definition}[RSC/RSS Condition] \label{def:RSCRSS}
    The loss function $\widetilde{\mathcal{L}}(\bm{A})$ is said to be restricted strongly convex (RSC) with parameter $\alpha$ and restricted strongly smooth (RSS) with parameter $\beta$ (where $0<\alpha\leq \beta$) if, for any two parameter matrices $\bm{A}=\bm{A}_2\otimes\bm{A}_1$ and $\bm{A}'=\bm{A}_2'\otimes\bm{A}_1'$ admitting the MARCF structural decomposition presented in \eqref{eq:decomposition}, the following inequality holds:
    \begin{equation}
        \frac{\alpha}{2}\left\|\mathbf{A}-\mathbf{A}'\right\|_{\mathrm{F}}^2 \leqslant \widetilde{\mathcal{L}}(\bm{A})-\widetilde{\mathcal{L}}\left(\mathbf{A}'\right)-\left\langle\nabla \widetilde{\mathcal{L}}\left(\mathbf{A}'\right), \mathbf{A}-\mathbf{A}'\right\rangle \leqslant \frac{\beta}{2}\left\|\mathbf{A}-\mathbf{A}'\right\|_{\mathrm{F}}^2.
    \end{equation}
\end{definition}

The statistical error is quantified by the following deviation bound.
\begin{definition}[Deviation Bound]\label{def: statistical error xi}
    For given ranks $(r_1,r_2)$, common subspace dimensions $(d_1,d_2)$, and the identifiable true value $\bm A^*$, the deviation bound is defined as
    \begin{equation}
        \xi(r_1,r_2,d_1,d_2):=\sup_{\substack{[\bm C_i~\bm R_i]\in\mathbb{O}^{p_i\times r_i},\\
        [\bm C_i~\bm P_i]\in\mathbb{O}^{p_i\times r_i},\\\bm D_i\in\mathbb{R}^{r_i},\fnorm{\bm D_i}=1}}\inner{\nabla\widetilde{\mathcal{L}}(\bm{A}^*)}{[\bm{C}_2~\bm{R}_2]\bm{D}_2[\bm{C}_2~\bm{P}_2]^\top\otimes[\bm{C}_1~\bm{R}_1]\bm{D}_1[\bm{C}_1~\bm{P}_1]^\top}.
    \end{equation}
\end{definition}

To quantify the estimation error, we consider $\fnorm{\bm A_2\otimes\bm A_1-\bm A^*}^2$, which is invariant under scaling and rotation. For the truth $\bm{A}_i^*$ and its decomposition components $\bm{C}_i^*$, $\bm{R}_i^*$, and $\bm{P}_i^*$, we follow the identification constraint in Section \ref{sec:2.4} to require $\|\bm A_1^*\|_\mathrm{F}=\|\bm A_2^*\|_\mathrm{F}=\phi$, $[\bm C_i^*~\bm R_i^*]^\top[\bm C_i^*~\bm R_i^*]=\phi^{2/3}\bm I_{r_i}$ and $[\bm C_i^*~\bm P_i^*]^\top[\bm C_i^*~\bm P_i^*]=\phi^{2/3}\bm I_{r_i}$, for $i=1,2$. As defined in \eqref{def: error}, $\mathrm{dist}(\bbm \Theta^{(j)},\bbm \Theta^*)^2$ measures the estimation error of the parameters after $j$ iterations, and naturally, $\mathrm{dist}(\bbm \Theta^{(0)},\bbm \Theta^*)^2$ represents the initialization error. Let $\underline{\sigma}:=\min(\sigma_{1,r_1},\sigma_{2,r_2})$ be the smallest sigular value among all the non-zero singular values of $\bm A_1^*$ and $\bm A_2^*$. Define $\kappa:=\phi/\underline{\sigma}$, which quantifies the ratio between the overall signal and the minimal signal. Then, we have the following sufficient conditions for local convergence of the algorithm.
\begin{theorem}[Local Linear Convergence]
    \label{theorem: computational convergence}
    Suppose that the RSC/RSS condition is satisfied with $\alpha$ and $\beta$. If $\lambda_1\asymp\alpha$, $\lambda_2\asymp\kappa^{-2}\phi^{8/3}\alpha$, $b\asymp\phi^{1 / 3}$, and $\eta=\eta_0\phi^{-10/3}\beta^{-1}$ with $\eta_0$ being a sufficiently small positive constant, $\xi^2\lesssim \kappa^{-6}\phi^4\alpha^3\beta^{-1}$, and the initialization error $\mathrm{dist}(\bbm \Theta^{(0)},\bbm \Theta^*)^2\leq c_0\kappa^{-2}\phi^{2/3}\alpha\beta^{-1}$, where $c_0$ is a small constant, then for all $j \geq 1$,
    \begin{equation}
        \begin{split}
            \mathrm{dist}(\bbm \Theta^{(j)},\bbm \Theta^*)^2 \leq & (1-C\eta_0\alpha\beta^{-1}\kappa^{-2})^j\mathrm{dist}(\bbm \Theta^{(0)},\bbm \Theta^*)^2
            +C\kappa^4\phi^{-10/3}\alpha^{-2}\xi^2(r_1,r_2,d_1,d_2),\\
        \end{split}
    \end{equation}
    \begin{equation}
        \begin{split}
            \fnorm{\bm A_1^{(j)}-\bm A_1^*}^2+\fnorm{\bm A_2^{(j)}-\bm A_2^*}^2 \lesssim & ~\kappa^2(1-C\eta_0\alpha\beta^{-1}\kappa^{-2})^j\left[\fnorm{\bm A_1^{(0)}-\bm A_1^*}^2+\fnorm{\bm A_2^{(0)}-\bm A_2^*}^2\right]\\
            &+\kappa^4\phi^{-2}\alpha^{-2}\xi^2(r_1,r_2,d_1,d_2),
        \end{split}
    \end{equation}
    and
    \begin{equation}
        \begin{split}
            \fnorm{\bm A_2^{(j)}\otimes\bm A_1^{(j)}-\bm A_2^*\otimes\bm A_1^*}^2\lesssim &~\kappa^2(1-C\eta_0\alpha\beta^{-1}\kappa^{-2})^j\fnorm{\bm A_2^{(0)}\otimes\bm A_1^{(0)}-\bm A_2^*\otimes\bm A_1^*}^2\\
            &+\kappa^4\alpha^{-2}\xi^2(r_1,r_2,d_1,d_2).
        \end{split}
    \end{equation}
\end{theorem}

Theorem \ref{theorem: computational convergence} confirms that, provided a good initialization, Algorithm \ref{algo: GD} converges linearly to a neighborhood of the truth defined by the statistical noise level $\xi$. Importantly, the conditions on the hyperparameters do not explicitly depend on the ambient dimensions $p_1, p_2$ or sample size $T$, but rather on the signal structure ($\alpha, \beta, \kappa$), highlighting the scalability of the method. The RSC/RSS and initialization conditions are justified in next subsection.

\subsection{Statistical Convergence Rates}\label{sec:4.2_statRates}

We now establish the statistical guarantees by verifying the deterministic conditions of Theorem \ref{theorem: computational convergence} under a probabilistic framework. We impose the following standard assumptions on the data generating process.

\begin{assumption}[Stationarity]
    \label{assumption: spectral redius}
    The spectral radius of $\mathbf{A}^*=\bm A_1^*\otimes \bm A_2^*$ is strictly less than one.
\end{assumption}

\begin{assumption}[Sub-Gaussian Noise]
    \label{assumption: sub-Gaussian noise}
    The white noise $\bm E_t$ can be represented as $\V{\bm E_t}=\bbm\Sigma_{\bm e}^{1/2}\bbm\zeta_t$, where $\bbm\Sigma_{\bm e}$ is a positive definite matrix, $\{\bbm\zeta_t\}$ are independent and identically distributed random variables with $\mathbb{E}[\bbm \zeta_t]=0$ and $\mathrm{Cov}(\bbm \zeta_t)=\bm I_{p_1p_2}$. The entries of $\bbm\zeta_t$, denoted by $\{\zeta_{it}\}_{i=1}^{p_1p_2}$, are independent $\tau^2$-sub-Gaussian, i.e., $\mathbb{E}\left[\exp \left(\mu \zeta_{i t}\right)\right] \leq \exp \left(\tau^2 \mu^2 / 2\right)$ for any $\mu \in \mathbb{R}$ and $1 \leq i \leq p_1p_2$.
\end{assumption}

Assumption \ref{assumption: spectral redius} ensures the existence of a unique strictly stationary solution to the proposed model. The sub-Gaussianity condition in Assumption \ref{assumption: sub-Gaussian noise} is standard in high-dimensional time series literature, such as \cite{Zheng2020LinearRisVAR} and \cite{wang2024LRtsAR}.

To ensure the identifiability of the common versus specific subspaces, we require a minimal separation between the response-specific and predictor-specific components. We quantify the separation between subspaces using $\sin\theta$ distance. Let $s_{i,1} \geq \cdots \geq s_{i,r_i-d_i} \geq 0$ be the singular values of $\mathbf{R}_i^{\top} \mathbf{P}_i$. Then, the canonical angles are defined as $\theta_k(\mathbf{R}_i, \mathbf{P}_i)=\arccos \left(s_{i,k}\right)$ for $1 \leq k \leq r_i-d_i$ and $i=1,2$.

\begin{assumption}[Minimal Gap between Specific Subspaces]
    \label{assumption: gmin}
    There exists a constant $g_{\min}>0$ such that $\min_{i=1,2}\sin \theta_1\left(\mathbf{R}_i^*, \mathbf{P}_i^*\right) \geq g_{\min }$. 
\end{assumption}

We quantify the dependency structure of the time series following the spectral approach as in \cite{BasuandGeorgeMichailidis2015}. The characterstic matrix polynomial of the MARCF model is given by $\mathcal{A}(z)=\mathbf{I}_p-\mathbf{A}^* z$ for $z\in\mathbb{C}$. Define $\mu_{\min }(\mathcal{A})=\min _{|z|=1} \lambda_{\min }\left(\mathcal{A}^{\dagger}(z) \mathcal{A}(z)\right)$ and $\mu_{\max }(\mathcal{A})=$ $\max _{|z|=1} \lambda_{\max }\left(\mathcal{A}^{\dagger}(z) \mathcal{A}(z)\right)$, where $\mathcal{A}^{\dagger}(z)$ is the conjugate transpose of $\mathcal{A}(z)$. We define the key theoretical quantities:
\begin{equation}
    M_1= \frac{\lambda_{\max }\left(\bbm{\Sigma}_{\bm e}\right)}{\mu_{\min }^{1/2}\left(\mathcal{A}\right)},~
    M_2= \frac{\lambda_{\min }\left(\bbm{\Sigma}_{\bm e}\right) \mu_{\min }\left(\mathcal{A}\right)}{\lambda_{\max}\left(\bbm{\Sigma}_{\bm e}\right) \mu_{\max }\left(\mathcal{A}\right)},~
    \alpha_{\textup{RSC}} = \frac{\lambda_{\min }\left(\bbm{\Sigma}_{\bm e}\right)}{2\mu_{\max }(\mathcal{A})},\text{ and } \beta_{\textup{RSS}}=\frac{3\lambda_{\max }\left(\bbm{\Sigma}_{\bm e}\right)}{2\mu_{\min }(\mathcal{A})}.
\end{equation}    
    
The following proposition verifies that the RSC and RSS conditions required for optimization hold with high probability given sufficient sample size.
\begin{proposition}[Condition Verification]\label{prop:conditionVerification}
    Under Assumptions \ref{assumption: spectral redius}--\ref{assumption: gmin}, if $T\gtrsim [M_2^{-2}\max\left(\tau,\tau^2\right) + g_{\min}^{-4}\kappa^2\underline{\sigma}^{-4}\alpha_{\textup{RSC}}^{-3}\beta_{\textup{RSS}} \tau^2M_1^2](p_1r_1+p_2r_2+4r_1^2r_2^2)$, we have that with probability at least $1-C\exp \left(-C\left(p_1r_1+p_2r_2\right)\right)$,
    the RSC/RSS condition in \ref{def:RSCRSS} holds with $\alpha=\alpha_{\textup{RSC}}$ and $\beta=\beta_{\textup{RSS}}$ and the initial value $\bm{\Theta}^{(0)}$ obtained as in Section \ref{subsec: initial} satisfies $\mathrm{dist}^2_{(0)}\lesssim \kappa^{-2}\phi^{2/3}\alpha_{\textup{RSC}}\beta_{\textup{RSS}}^{-1}$.
\end{proposition}

Combining the computational convergence with the bound on the statistical deviation $\xi$, we derive the final error rates for the MARCF estimator.

\begin{theorem}[Statistical Error Rates]
    \label{theorem: statistical error}
    Under Assumptions \ref{assumption: spectral redius}--\ref{assumption: gmin}, if the hyperparameters $\lambda_1$, $\lambda_2$, $b$, and $\eta$ are set as in Theorem \ref{theorem: computational convergence} with $\alpha=\alpha_\mathrm{RSC}$ and $\beta=\beta_\mathrm{RSS}$, $T\gtrsim [M_2^{-2}\max\left(\tau,\tau^2\right) + g_{\min}^{-4}\kappa^2\underline{\sigma}^{-4}\alpha_{\textup{RSC}}^{-3}\beta_{\textup{RSS}} \tau^2M_1^2](p_1r_1+p_2r_2+4r_1^2r_2^2)$,
    and
    \begin{equation}\label{LB of stat convergence}
        J\gtrsim\frac{\log\left(\kappa^{-2}g_{\min}^{4}\right)}{\log(1-C\eta_0\alpha_\mathrm{RSC}\beta_\mathrm{RSS}^{-1}\kappa^{-2})},
    \end{equation}
    then, with probability at least $1-C\exp\left(-C(p_1r_1+p_2r_2)\right)$,
$$
    \fnorm{\bm A_1^{(J)}-\bm A_1^*}^2+\fnorm{\bm A_2^{(J)}-\bm A_2^*}^2\lesssim \kappa^4\phi^{-2}\alpha_\mathrm{RSC}^{-2}\tau^2M_1^2\frac{\mathrm{df}_\mathrm{MARCF}}{T},
$$
and
$$
    \fnorm{\bm A_2^{(J)}\otimes \bm A_1^{(J)}-\bm A_2^*\otimes \bm A_1^*}^2\lesssim \kappa^4\alpha_\mathrm{RSC}^{-2}\tau^2M_1^2\frac{\mathrm{df}_\mathrm{MARCF}}{T}.
$$
where $\mathrm{df}_\mathrm{MARCF}$ is the effective number of free parameters defined in \eqref{eq:df_MARCF}.
\end{theorem}

Theorem \ref{theorem: statistical error} establishes that the estimator achieves the parametric rate $O_p(\mathrm{df}_\mathrm{MARCF}/T)$. Since $\mathrm{df}_\mathrm{MARCF} =\mathrm{df}_\mathrm{RRMAR} - (p_1d_1 + p_2d_2)$, the MARCF model offers a demonstrable variance reduction over the standard RRMAR model by exploiting the common subspace structure. Moreover, the required number of iterations in \eqref{LB of stat convergence} does not explicitly depend on $p_1$, $p_2$, or $T$, implying the scalability and computational efficiency of our method even in high-dimensional regimes. 

\subsection{Consistency of Rank Selection}\label{sec:consistency_of_rank_selection}

In this subsection, we establish the consistency of rank selection under the standard high-dimensional asymptotic regime \citep[see, e.g.,][]{lam2011estimation, LamAndYao2012}, where $T$, $p_1$, and $p_2\to\infty$ with $r_1$ and $r_2$ remaining fixed.
\begin{theorem}[Rank Selection Consistency]
    \label{theorem: selecting r}
    Under Assumptions \ref{assumption: spectral redius} and \ref{assumption: sub-Gaussian noise}, if $\bar{r}_1\geq r_1, \bar{r}_2\geq r_2$, $T\gtrsim (p_1r_1+p_2r_2+4r_1^2r_2^2)M_2^{-2}\max\left(\tau,\tau^2\right)$, $\phi^{-1}\alpha_{\textup{RSC}}^{-1}\tau M_1\sqrt{(p_1+p_2)/T}=o(s(p_1,p_2,T))$, and $s(p_1,p_2,T)=o(\sigma_{i,r_i}^{-1}\min_{1\leq j\leq r_i-1}\sigma_{i,j+1}/\sigma_{i,j})$ for $i=1,2$, then we have $\mathbb{P}(\widehat{r}_1=r_1,\widehat{r}_2=r_2)\to 1$, as $T\to \infty$ and $p_1,p_2\to\infty$.
\end{theorem}

This theorem provides sufficient conditions under which the selected ranks $\widehat{r}_1$ and $\widehat{r}_2$ converge in probability to the true ranks $r_1$ and $r_2$. In practice, if the parameters $M_1$, $M_2$, $\alpha_{\textup{RSC}}$, $\sigma_{1,r_1}$, and $\sigma_{2,r_2}$ are bounded, and $\phi$ is either bounded or diverges slowly as $p_1,p_2\to\infty$, the conditions simplify to requiring that $s(p_1,p_2,T)\to 0$ and $\sqrt{(p_1+p_2)/T}/s(p_1,p_2,T)\to 0$. These simplified conditions are more tractable in applications and help ensure the interpretability and reliability of the rank selection procedure. Particularly, the value $s(p_1,p_2,T)=\sqrt{(p_1+p_2)\log(T)/(20T)}$ suggested in Section \ref{subsec: selection} satisfies these requirements, ensuring consistency across a wide range of settings for $p_1$, $p_2$, and $T$.

\section{Simulation Studies}\label{sec: simulation}

In this section, we evaluate the finite-sample performance of the proposed MARCF model through a series of simulation experiments. The objectives are three-fold: (i) to assess the accuracy of parameter estimation and model selection, (ii) to validate the model's ability to identify structured factor-driven dynamics, and (iii) to compare its performance with existing approaches. 

\subsection{Experiment I: Estimation Accuracy and Model Selection}

In the first experiment, we generate data from the MARCF model defined in \eqref{eq:MARCF}. We first assess the reliability of the proposed procedures for selecting the model ranks and common dimensions. Then, we compare the MARCF model with two competing approaches: the full-rank MAR model \citep{RongChen2021MatrixAR} and the RRMAR model \citep{RongChen2021MatrixAR}, by evaluating their estimation accuracy of $\bm A_2 \otimes \bm A_1$.

We fix the model ranks at $r_1=2$ and $r_2=3$, and consider all combinations where $d_1\in\{0,1,2\}$ and $d_2\in\{0,1,2,3\}$. Two settings for dimensions and sample size are examined: (1) $p_1=20$, $p_2=10$, and $T=500$; and (2) $p_1=30$, $p_2=20$, and $T=1000$. For each combination of $\{p_1,p_2,d_1,d_2,T\}$, the matrices $\bm A_1^*$, $\bm A_2^*$, and the data are generated as follows:
\begin{enumerate}
    \item For $i=1, 2$, generate $\bm D_i$ as $\bm D_i=\bm O_{i,1}^\top\bm S_i\bm O_{i,2}$, where $\bm O_{i,1}$ and $\bm O_{i,2}$ are random orthogonal matrices, and $\bm S_i \in \mathbb{R}^{r_i\times r_i}$ is a diagonal matrix with entries drawn from $U(0.8,1)$ when $r_i=d_i$, and from $U(0.9,1.1)$ when $r_i\neq d_i$.
    \item Generate two random matrices $\bm C_1\in\mathbb{R}^{p_1\times r_1}$ and $\bm C_2\in\mathbb{R}^{p_2\times r_2}$ with orthonormal columns. Throughout this section, random matrices with orthonormal columns are constructed by applying QR decomposition to matrices with i.i.d. standard normal entries. Then, generate random matrices $\bm M_{i,1}$ and $\bm M_{i,2}$, and apply QR decomposition onto $(\bm I-\bm C_i\bm C_i^\top)\bm M_{i,1}$ and $(\bm I-\bm C_i\bm C_i^\top)\bm M_{i,2}$ to obtain $\bm R_i$ and $\bm P_i$. Results with $\sin\theta_1<0.8$ are rejected. 
    \item Combine the components to form $\bm A_1$ and $\bm A_2$. We check for stationarity and repeat the above steps if $\rho(\bm A_1)\cdot\rho(\bm A_2)\geq 1$. The white noise is generated with $\V{\bm E_t}\overset{i.i.d.}\sim N(\bm 0,\bm I)$, and $\{\bm Y_t\}_{t=0}^T$ is generated according to model \eqref{eq:MARCF}.
\end{enumerate}

For each pair of $(d_1,d_2)$, the MARCF model is trained using the full procedure proposed in Section \ref{sec: estimation}, including rank selection, common dimension selection, initialization, and final gradient descent estimation (see Appendix S1 of supplementary materials for details). For rank selection, we set $\bar{r}_1=\bar{r}_2=8$ with $s\left(p_1, p_2, T\right)=\sqrt{\left(p_1+p_2\right) \log (T) /(20 T)}$. Both the RRMAR and MARCF models are estimated using our gradient descent algorithm, as RRMAR is a special case of MARCF. Based on preliminary experiments, we observed that the model performance is insensitive to variations in $\lambda_1, \lambda_2$, and $b$. Accordingly, we set these parameters to 1. We set $\lambda_1=\lambda_2=b=1$ and $\eta=0.001$. Both initialization and final estimation are performed by running Algorithm \ref{algo: GD} until convergence or for a maximum of 1000 iterations. No divergence was observed under these hyperparameter settings.

Table \ref{table:sim1_accuracy} presents the accuracy of selecting $(r_1, d_1, r_2, d_2)$ across various settings. A trial is considered successful only if all four parameters are correctly identified. The results show that in both settings of $p_1$, $p_2$, and $T$, the proposed selection procedure achieves a success rate approaching 1. Even when $d_1$ and $d_2$ are close to $r_1$ and $r_2$, the procedure maintains high accuracy. 

\begin{table}[htp]
    \centering
    \caption{Selection accuracy under $r_1=3$ and $r_2=2$. The left panel corresponds to the case with $p_1=20$, $p_2=10$, and $T=500$, and the right panel corresponds to the case with $p_1=30$, $p_2=20$, and $T=1000$. For each combination of $p_i$, $r_i$, $d_i$, and $T$, the successful rates are computed from 500 replications.}
    \begin{tabular}{ccccc|cccc}
    \toprule
     & $d_1=0$ & $d_1=1$ & $d_1=2$ & $d_1=3$ & $d_1=0$ & $d_1=1$ & $d_1=2$ & $d_1=3$\\
    \midrule
    $d_2=0$ & 1.000 & 1.000 & 1.000 & 1.000 & 1.000 & 0.998 & 1.000 & 1.000\\
    $d_2=1$ & 1.000 & 0.988 & 0.946 & 0.944 & 1.000 & 0.984 & 0.926 & 0.942\\
    $d_2=2$ & 1.000 & 0.982 & 0.884 & 0.998 & 1.000 & 0.954 & 0.854 & 0.998\\
    \bottomrule
    \end{tabular}
    \label{table:sim1_accuracy}
\end{table}

Figure \ref{fig:sim1_estimation_error} displays the relative estimation errors of the three models under comparison. The error is defined as $\|\widehat{\bm A}_2\otimes \widehat{\bm A}_1-\bm A_2^*\otimes \bm A_1^*\|_\text{F}/\|\bm A_2^*\otimes \bm A_1^*\|_\text{F}$. For each model, the lower and upper lines of each error bar represent the first and third quartiles of the 500 errors, respectively and the cross markers indicate the medians. In all cases, the estimation errors of the MAR model are significantly larger than those of the other models, as expected, due to its failure to account for the low-rank structure in the data. When $d_1=d_2=0$, the MARCF model reduces to the RRMAR model, and the estimation errors of the two are identical. However, as $d_1$ and $d_2$ increase, the error bars of the MARCF model decrease and gradually diverge from those of the RRMAR model, highlighting the improved performance of MARCF model when common structures of the subspaces exist. This performance advantage becomes most pronounced when $d_1=r_1$ and $d_2=r_2$, where MARCF achieves the lowest estimation error. This trend empirically validates the efficiency gains obtained by explicitly modeling the common subspace structures.

\begin{figure}[htp]
    \centering
    \hspace{-0.25cm}
    \begin{subfigure}[b]{0.5\textwidth}
        \centering
        \includegraphics[width=\linewidth]{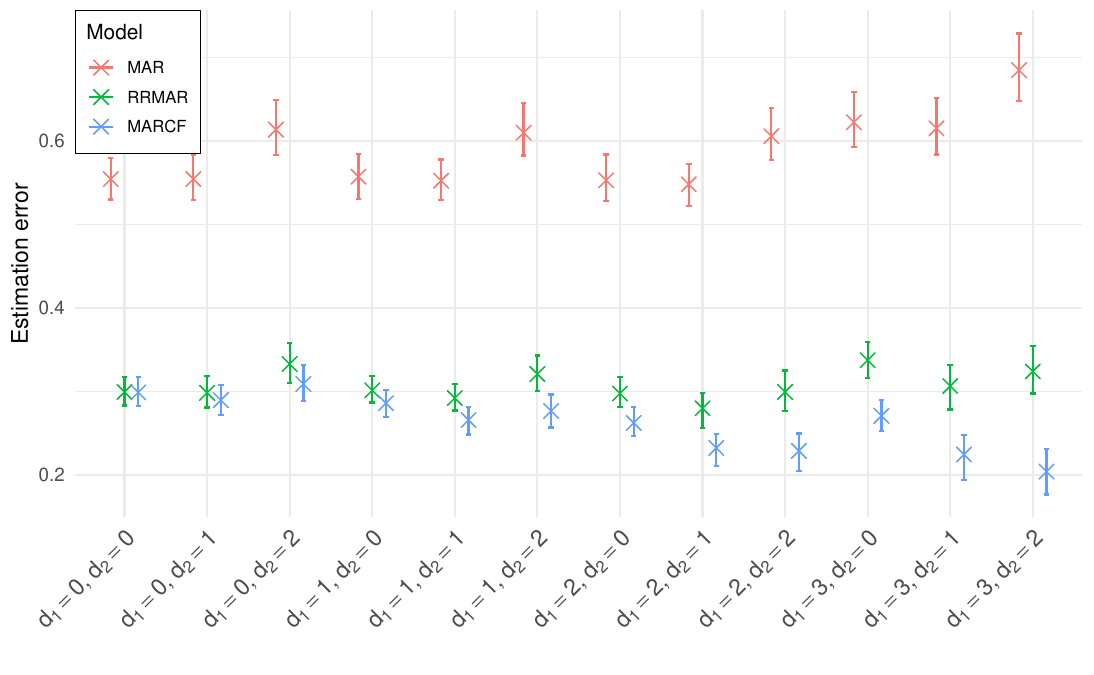}
        \subcaption{}
        \label{subfig:500}
    \end{subfigure}
    \hspace{-0.25cm}
    \begin{subfigure}[b]{0.5\textwidth}
        \centering
        \includegraphics[width=\linewidth]{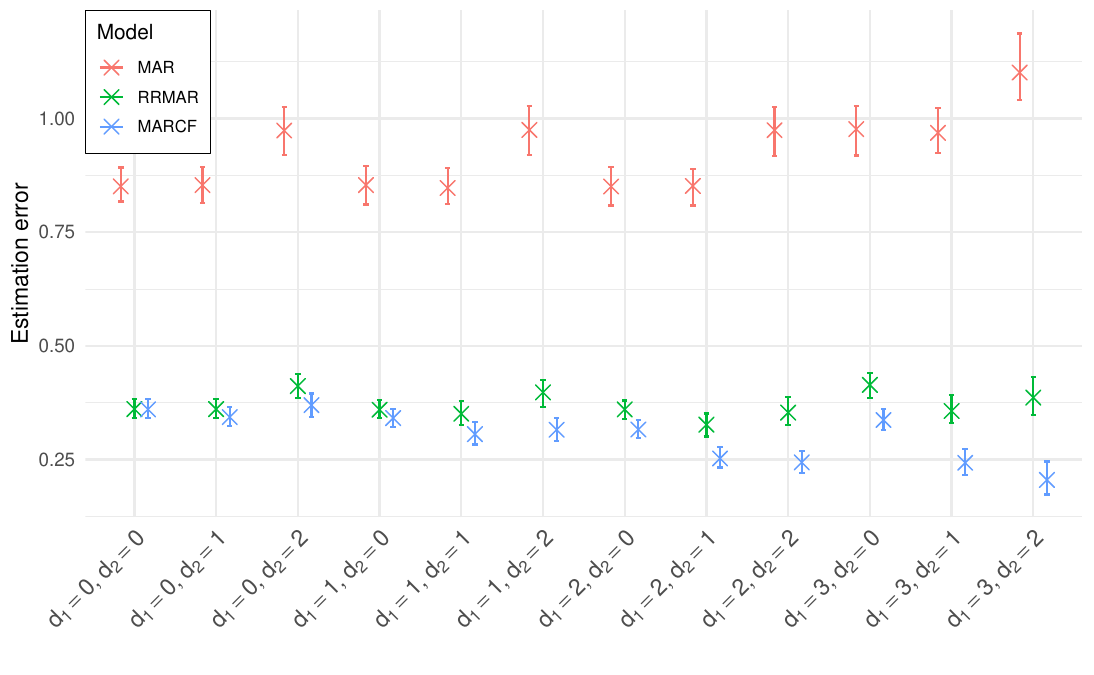}
        \subcaption{}
        \label{subfig:1000}
    \end{subfigure}
    \caption{Relative estimation error with various combinations of $(d_1,d_2)$ under two settings of $p_1$, $p_2$ and $T$: (a) $p_1=20$, $p_2=10$, and $T=500$; and (b) $p_1=30$, $p_2=20$, and $T=1000$. Results are based on 500 replications.}
    \label{fig:sim1_estimation_error}
\end{figure}

\subsection{Experiment II: Identification of Factor-Driven Dynamics}

In the second experiment, we evaluate whether the MARCF model can correctly identify a pure factor-driven structure when the data are generated from the dynamic MFM, as defined in \eqref{eq:DFM} and \eqref{eq:factor_AR}. The goal is to assess the proportion of trials in which the modeling precedure correctly identifies the dynamic MFM structure, and to compare the estimation errors of the factor loading projection matrices between MARCF and MFM in \citet{wang2019factor}.

We set $p_1=p_2=16$, $T=800$, and consider cases where $r_1=d_1\in\{1,2,3,4\}$ and $r_2=d_2\in\{1,2,3,4\}$. The data are generated as follows:
\begin{enumerate}
    \item For $i=1,2$, generate $\bm \Lambda_i\in \mathbb{R}^{p_i\times r_i}$ with random orthonormal columns as factor loading matrices. Generate $\bm B_i\in\mathbb{R}^{r_i\times r_i}$ using the same procedure as for $\bm D_i$ in Experiment I. Ensure stationarity by rejecting realizations where $\rho(\bm B_1)\cdot\rho(\bm B_2)\geq 1$.
    \item Generate two independent white noise processes $\bm E_t\in\mathbb{R}^{p_1\times p_2}$ and $\bbm \xi_t\in\mathbb{R}^{r_1 \times r_2}$ with $\V{\bm E_t}\stackrel{\text{i.i.d.}}{\sim} N_{p_1p_2}(\bm 0,\bm I)$ and $\V{\bbm \xi_t}\stackrel{\text{i.i.d.}}{\sim} N_{r_1r_2}(\bm 0,\bm I)$. Construct $\bm W_t=\bm E_t-\bm\Lambda_1\bm\Lambda_1^\top\bm E_t\bm \Lambda_2\bm \Lambda_2^\top$ to eliminate the moving average term in the composite noise of \eqref{eq:MAR_in_DFM}.
    \item Generate the factors as $\bm F_t=\bm B_1\bm F_{t-1}\bm B_2^\top+\bbm\xi_t$ and the observation as $\bm Y_t=\bm\Lambda_1\bm F_t\bm\Lambda_2^\top+\bm W_t$.
\end{enumerate}
The training procedure and tuning parameters are identical to those in Experiment I, except that the step size $\eta$ is initially set to 0.01 and reduced to 0.001 if Algorithm \ref{algo: GD} diverges. For the MFM, we use the TIPUP method from \cite{Chen2022TensorFM} with true ranks.

Table \ref{table: DMFM accuracy} presents the proportion of trials in which the MFM pattern is correctly identified, i.e., $\widehat{r}_1=\widehat{d}_1=r_1$ and $\widehat{r}_2=\widehat{d}_2=r_2$. The results demonstrate that the MARCF model successfully identifies the MFM with high probability, particularly when multiple factors are present. Even for $r_1=r_2=1$, the success rate is 84.8\%.

\begin{table}[htp]
    \centering
    \caption{Selection accuracy on MFM data. For each pair of $(r_1, r_2)$, the result is the successful rate over 500 replications.}
    \begin{tabular}{cccccc}
        \toprule
        & $r_1 = 1$ & $r_1 = 2$ & $r_1 = 3$ & $r_1 = 4$ \\
        \midrule
        $r_2 = 1$ & 0.848 & 0.996 & 0.984 & 0.992 \\
        $r_2 = 2$ & 0.992 & 0.996 & 1.000 & 0.996 \\
        $r_2 = 3$ & 0.990 & 0.996 & 1.000 & 1.000 \\
        $r_2 = 4$ & 0.992 & 0.998 & 1.000 & 0.998 \\
        \bottomrule
    \end{tabular}
    \label{table: DMFM accuracy}
\end{table}

Figure \ref{fig:sim2} persents the log estimation errors of the proposed model and MFM. Only trials with correctly identified parameters are included. The error is defined as 
$$
\mathrm{Log ~estimation ~error}=\log\left(\fnorm{\widehat{\bm \Lambda}_i(\widehat{\bm \Lambda}_i^\top\widehat{\bm \Lambda}_i)^{-1}\widehat{\bm \Lambda}_i^\top-\bm \Lambda_i^*\bm \Lambda_i^{*\top}}\right),\quad \mathrm{for~} i=1,2.
$$ 
Figures \ref{subfig:sim2_L1} and \ref{subfig:sim2_L2} show the estimation errors for the column spaces $\mathcal{M}(\bbm\Lambda_1^*)$ and $\mathcal{M}(\bbm\Lambda_2^*)$, respectively. In each subplot, the left (blue) boxplot represents the MFM, and the right (yellow) boxplot represents the MARCF model. For $\mathcal{M}(\bm{\Lambda}_1^*)$, the MARCF model shows slightly smaller errors when $r_2=1$, and significantly smaller errors when $r_2\geq2$, especially for $r_2\geq3$. Similar trends hold for $\mathcal{M}(\bm{\Lambda}_2^*)$. Overall, MARCF performs comparably to or better than MFM, demonstrating its effectiveness and flexibility.

\begin{figure}[htp]
    \centering
    \hspace{-0.2cm}
    \begin{subfigure}[b]{0.5\textwidth}
        \centering
        \includegraphics[width=\linewidth]{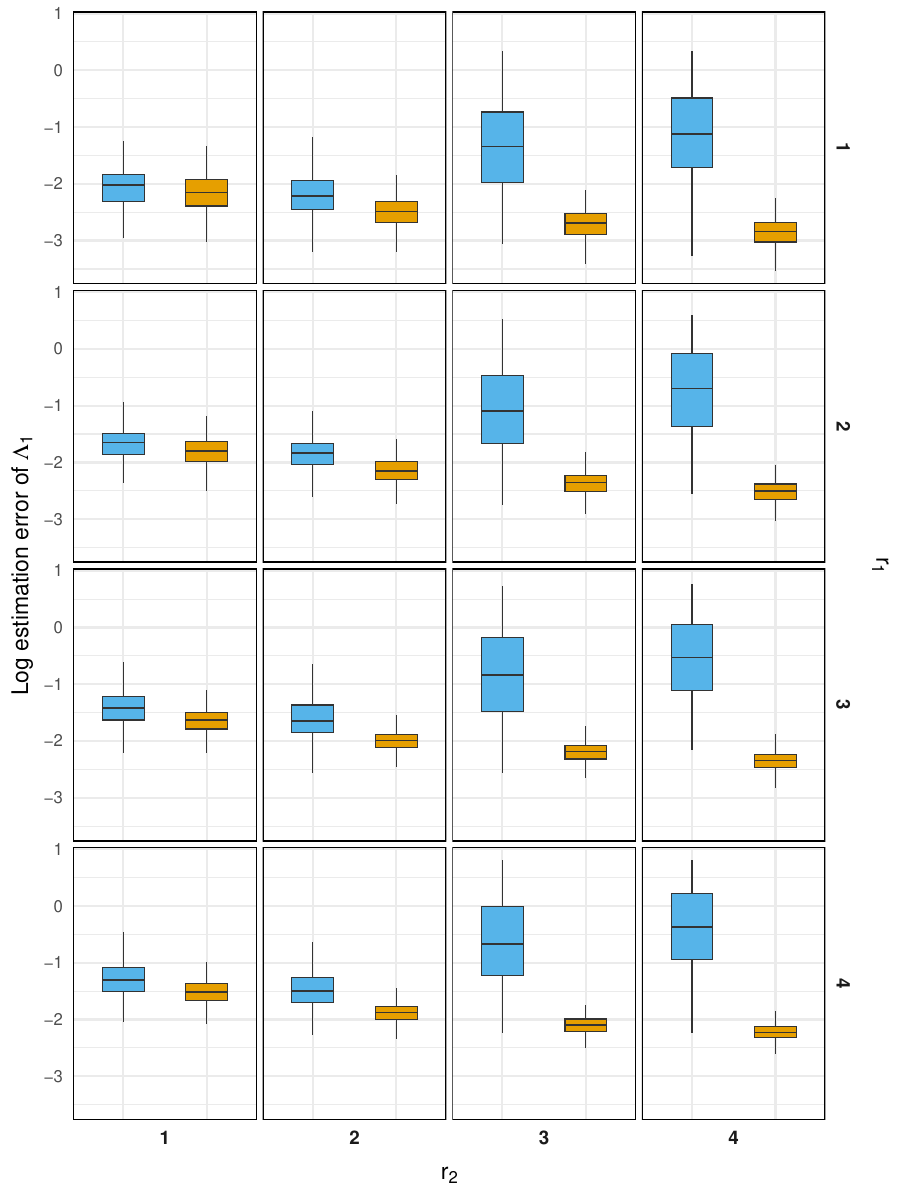}
        \subcaption{}
        \label{subfig:sim2_L1}
    \end{subfigure}
    \hspace{-0.2cm}
    \begin{subfigure}[b]{0.5\textwidth}
        \centering
        \includegraphics[width=\linewidth]{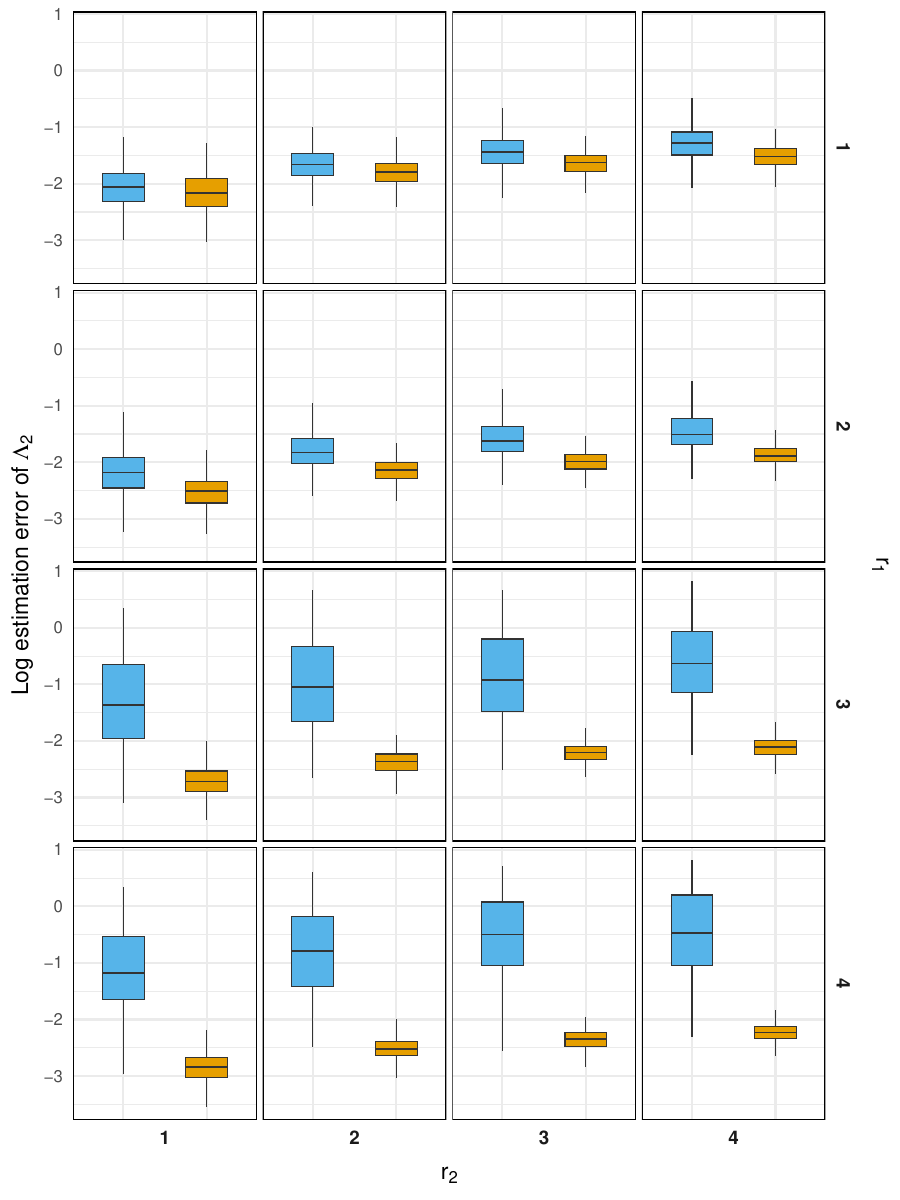}
        \subcaption{}
        \label{subfig:sim2_L2}
    \end{subfigure}
    \caption{Log estimation error of (a) $\mathcal{M}(\bm \Lambda_1^*)$; and (b) $\mathcal{M}(\bm \Lambda_2^*)$. In each subplot, the left (blue) and right (yellow) boxplot represent the estimation errors of the dynamic MFM and the MARCF model, respectively. Each pair of $(r_1,r_2)$ is evaluated over 500 replications.}
    \label{fig:sim2}
\end{figure}

\section{Real Example: Quarterly Macroeconomic Data}\label{sec: real data}

To demonstrate the empirical utility of the MARCF framework, we apply it to a multinational macroeconomic dataset. This dataset, previously analyzed by \citet{Chen24ConstaintMFM}, comprises 10 key economic indicators observed across 14 countries over 107 quarters, spanning from 1990-Q2 to 2016-Q4. The data is available in the supplementary materials for reproducibility.

The economic indicators include the Consumer Price Index (CPI) for food (CPIF), energy (CPIE), and total items (CPIT); long-term interest rates (government bond yields, IRLT) and 3-month interbank rates (IR3); total industrial production excluding construction (PTEC) and total manufacturing production (PTM); Gross Domestic Product (GDP); and the total value of exports (ITEX) and imports (ITEM). The 14 countries are Australia (AUS), Austria (AUT), Canada (CAN), Denmark (DNK), Finland (FIN), France (FRA), Germany (DEU), Ireland (IRL), the Netherlands (NLD), Norway (NOR), New Zealand (NZL), Sweden (SWE), the United Kingdom (GBR), and the United States (USA). Prior to analysis, all series were transformed to induce stationarity (via logarithmic differences), seasonally adjusted, and standardized to zero mean and unit variance.

\subsection{Model Selection and Estimation}

We determine the model structure based on rolling window validation. Based on the predictive performance over a 16-quarter validation set, we selected the total ranks $r_1 = 4$ and $r_2 = 5$. Subsequently, the common subspace dimensions were selected via BIC as described in Section \ref{subsec: selection}, yielding $d_1 = 3$ and $d_2 = 2$. These non-zero intersections ($d_i > 0$) and non-zero specific dimensions ($r_i > d_i$) provide strong empirical support that the data structure is hybrid: neither the purely disjoint subspaces of RRMAR nor the perfectly aligned subspaces of DMF adequately capture the complexity of global economic dynamics. For interpretation, the estimated basis matrices $[\widehat{\bm C}_i~\widehat{\bm R}_i]$ and $[\widehat{\bm C}_i~\widehat{\bm P}_i]$ are orthogonalized via QR decomposition to enforce strict orthonormality.

\subsection{Interpretation of Global Economic Dynamics}

We dissect the estimated dynamics by analyzing the projection matrices associated with the row (country) and column (indicator) subspaces.

\paragraph{Country-Level Dynamics (Row Space).} 
Figure \ref{fig: country 1} displays the estimated projection matrices for the country dimension. The diagonal elements of these matrices reveal four distinct clusters with specific roles in the global economy: (1) an Anglo-Pacific group (Australia and UK); (2) a Core Economic Driver group (USA, Germany, Netherlands); (3) a Periphery group (Ireland, Finland); and (4) a general group including Canada and other EU members.

\begin{figure}[htp]
    \centering
    \begin{subfigure}[b]{0.32\textwidth} 
        \centering
        \includegraphics[width=\textwidth]{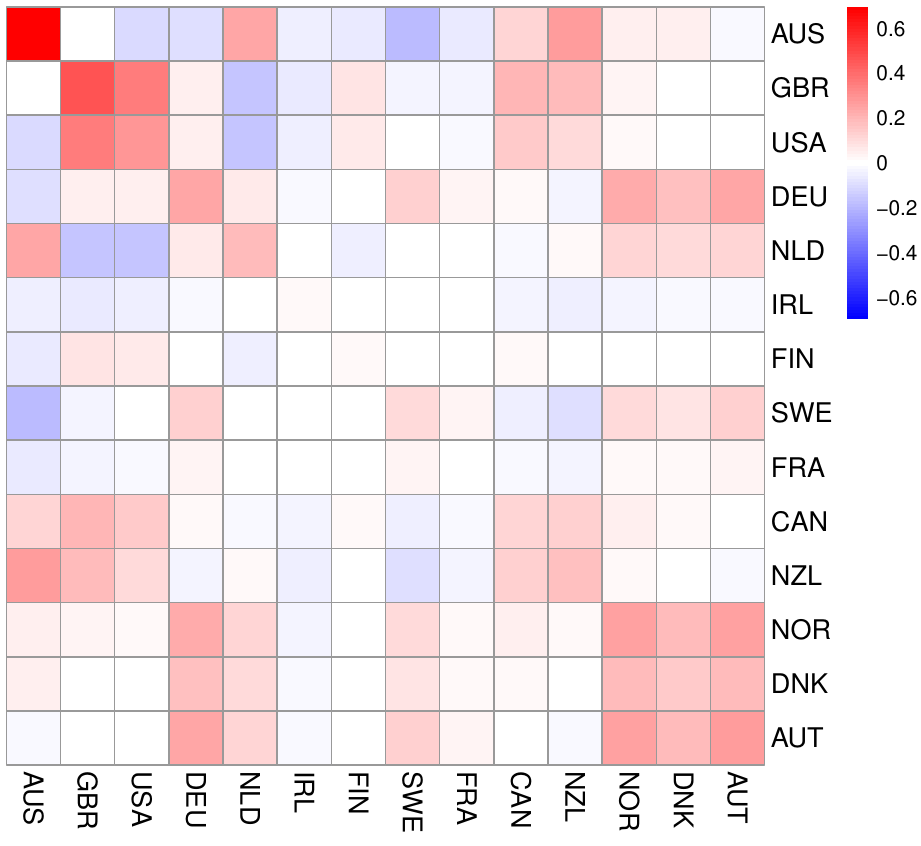}
        \caption{$\widehat{\bm C}_1\widehat{\bm C}_1^\top$}
        \label{subfig: C1}
    \end{subfigure}
    \begin{subfigure}[b]{0.32\textwidth} 
        \centering
        \includegraphics[width=\textwidth]{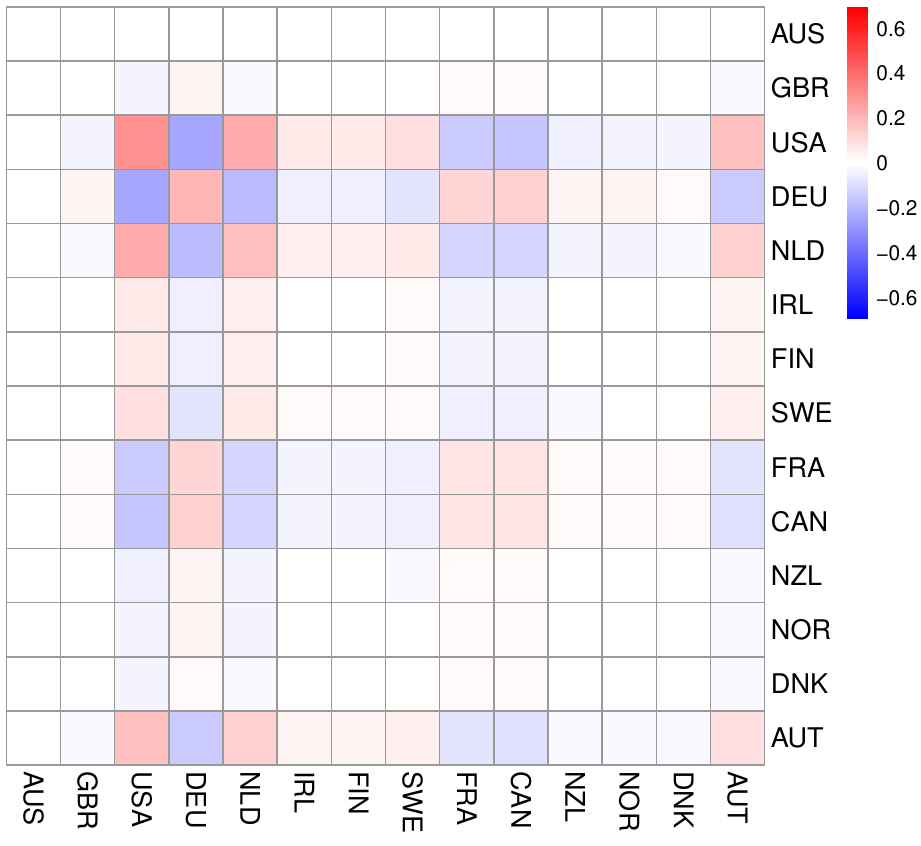}
        \caption{$\widehat{\bm P}_1\widehat{\bm P}_1^\top$}
        \label{subfig: P1}
    \end{subfigure}
    \begin{subfigure}[b]{0.32\textwidth} 
        \centering
        \includegraphics[width=\textwidth]{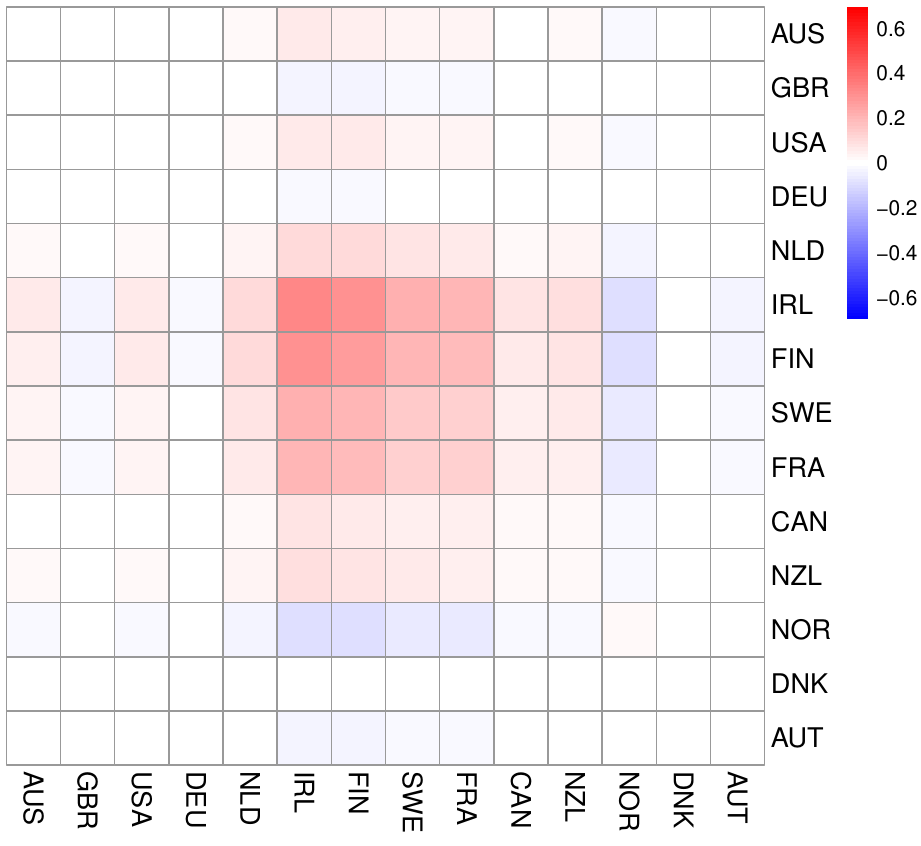}
        \caption{$\widehat{\bm R}_1\widehat{\bm R}_1^\top$}
        \label{subfig: R1}
    \end{subfigure}
    \caption{Estimates of common and specific projection matrices over countries.}
    \label{fig: country 1}
\end{figure}

The decomposition of these groups across the subspaces offers significant economic insight:
\begin{itemize}
    \item Common Influencers: Australia and the UK (Figure \ref{subfig: C1}) exhibit strong loadings in the common subspace. This implies a dual role: their economic conditions are highly predictive of global trends, and simultaneously, they are strongly predicted by the global state. This reflects their status as deeply integrated, open economies with significant financial sectors.
    \item Global Drivers: The USA, Germany, and the Netherlands feature prominently in the predictor-specific subspace (Figure \ref{subfig: P1}) but are less visible in the response-specific subspace. This asymmetry suggests a causal directionality: shocks originating in these major economies drive future global dynamics, but their own fluctuations are less determined by external factors. This aligns with the US's role as a global policy setter and Germany/Netherlands as the industrial and logistical engines of Europe.
    \item Global Responders: In contrast, Ireland and Finland load heavily on the response-specific subspace (Figure \ref{subfig: R1}). This characterizes them as small open economies that absorb global shocks—meaning their economic state is predicted by global factors—but their idiosyncratic movements have limited predictive power for the rest of the world.
\end{itemize}

Figure \ref{fig: country 2} combines these components to visualize the total predictive and response footprints. The predictor footprint (Figure \ref{subfig: C1_P1}) is dominated by the largest economies (USA, Germany, UK, Australia), confirming that global co-movements are driven by this core group. The correlation structure (off-diagonal elements) reveals positive associations within trade-dependent clusters (Australia/NZL/Netherlands) and cultural/economic blocks (USA/UK). Interestingly, negative correlations appear between the USA and continental European powers (Germany/France), potentially reflecting divergent economic cycles or structural differences between Anglo-Saxon and Continental economic models.

\begin{figure}[!htp]
    \centering
    \begin{subfigure}[b]{0.32\textwidth} 
        \centering
        \includegraphics[width=\textwidth]{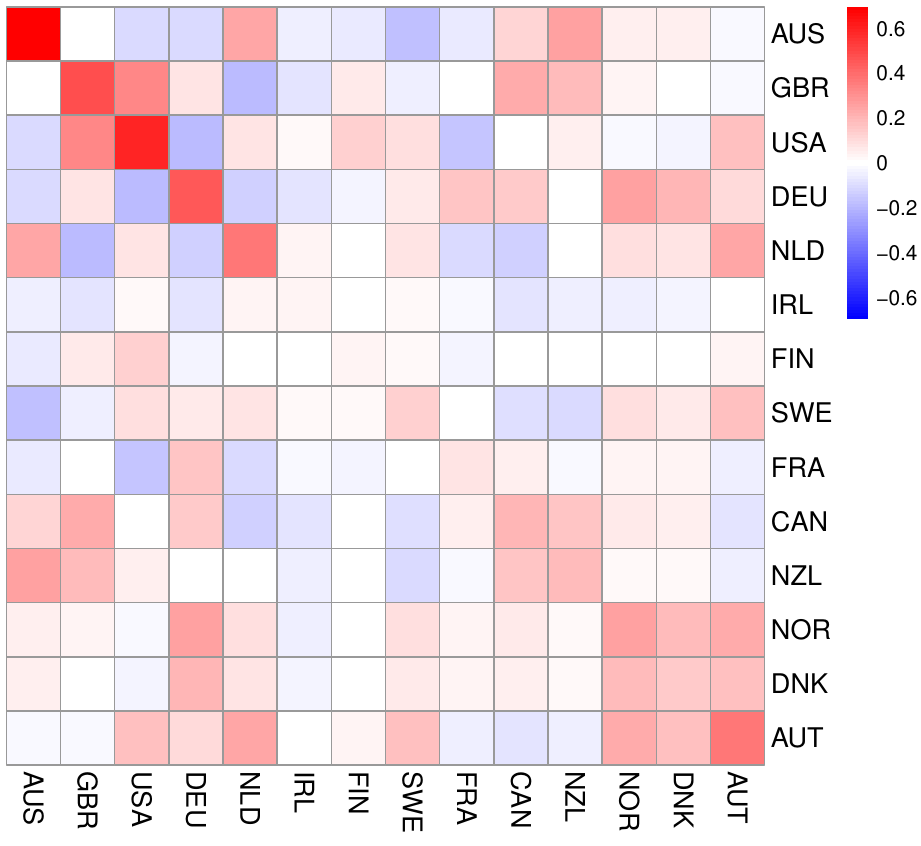}
        \caption{$\widehat{\bm C}_1\widehat{\bm C}_1^\top+\widehat{\bm P}_1\widehat{\bm P}_1^\top$}
        \label{subfig: C1_P1}
    \end{subfigure}
    \begin{subfigure}[b]{0.32\textwidth} 
        \centering
        \includegraphics[width=\textwidth]{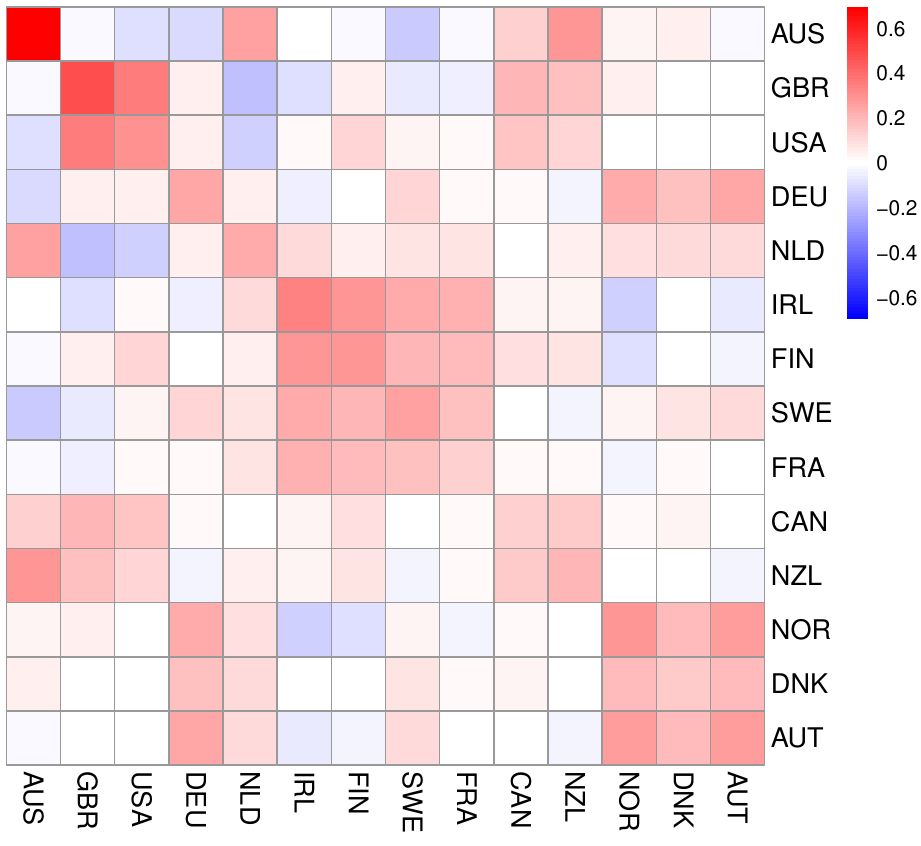}
        \caption{$\widehat{\bm C}_1\widehat{\bm C}_1^\top+\widehat{\bm R}_1\widehat{\bm R}_1^\top$}
        \label{subfig: C1_R1}
    \end{subfigure}
    \caption{Estimates of full predictor and response projection matrices over countries.}
    \label{fig: country 2}
\end{figure}

\paragraph{Indicator-Level Dynamics (Column Space).}
The analysis of economic indicators (Figures \ref{fig: index 1} and \ref{fig: index 2}) reveals clear lead-lag relationships. 
\begin{itemize}
    \item Predictors: Real economic activity measures, including GDP, Manufacturing (PTM), Industrial Production (PTEC), and Exports (ITEX), dominate the predictor-specific subspace (Figure \ref{subfig: P2}). This confirms their utility as leading indicators for future economic states.
    \item Responders: Energy prices (CPIE) are dominant in the response-specific subspace (Figure \ref{subfig: R2}) but negligible in the predictor space. This is consistent with the status of most countries in the sample as energy importers; their domestic energy prices are determined by global markets (the response) but do not themselves drive global cycles.
    \item Commonality: Food prices (CPIF) span both subspaces, suggesting high persistence and integration into both the cost of living and production chains.
\end{itemize}

\begin{figure}[htb]
    \centering
    \begin{subfigure}[b]{0.32\textwidth} 
        \centering
        \includegraphics[width=\textwidth]{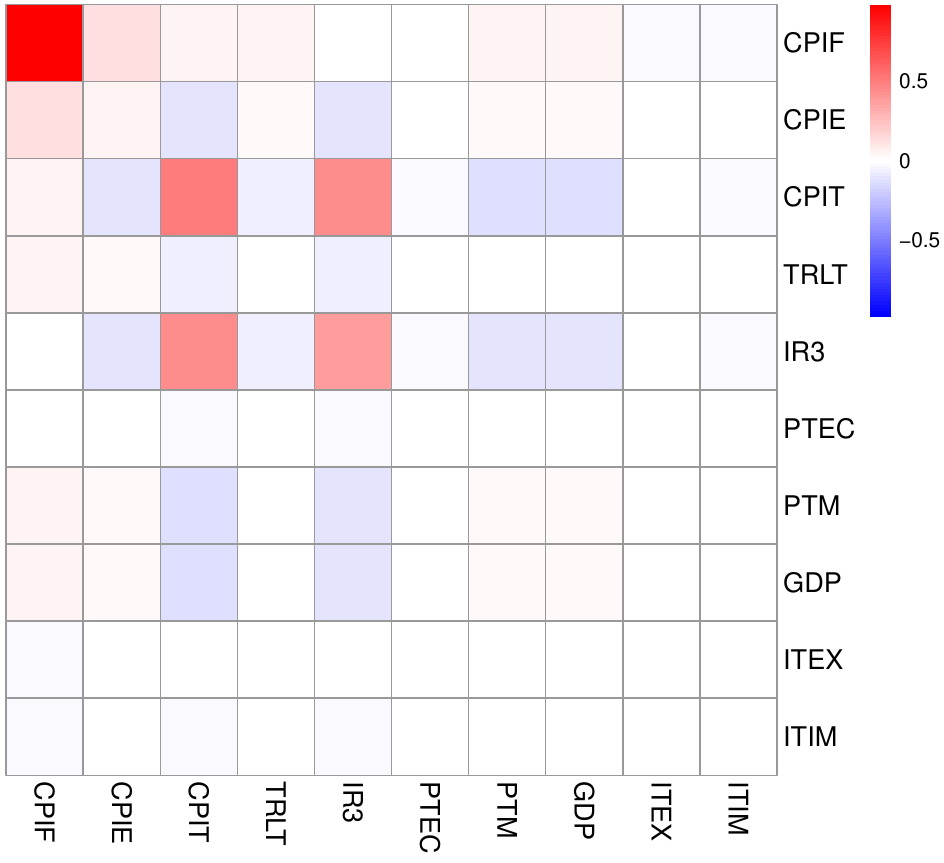}
        \caption{$\widehat{\bm C}_2\widehat{\bm C}_2^\top$}
        \label{subfig: C2}
    \end{subfigure}
    \begin{subfigure}[b]{0.32\textwidth} 
        \centering
        \includegraphics[width=\textwidth]{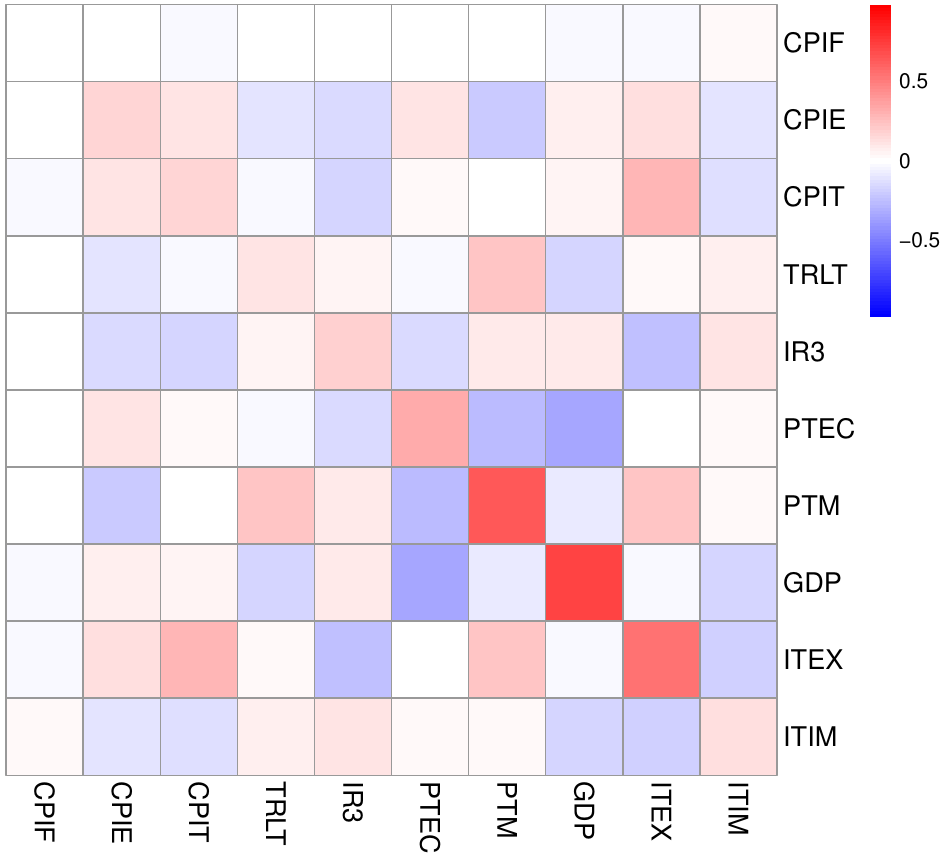}
        \caption{$\widehat{\bm P}_2\widehat{\bm P}_2^\top$}
        \label{subfig: P2}
    \end{subfigure}
    \begin{subfigure}[b]{0.32\textwidth} 
        \centering
        \includegraphics[width=\textwidth]{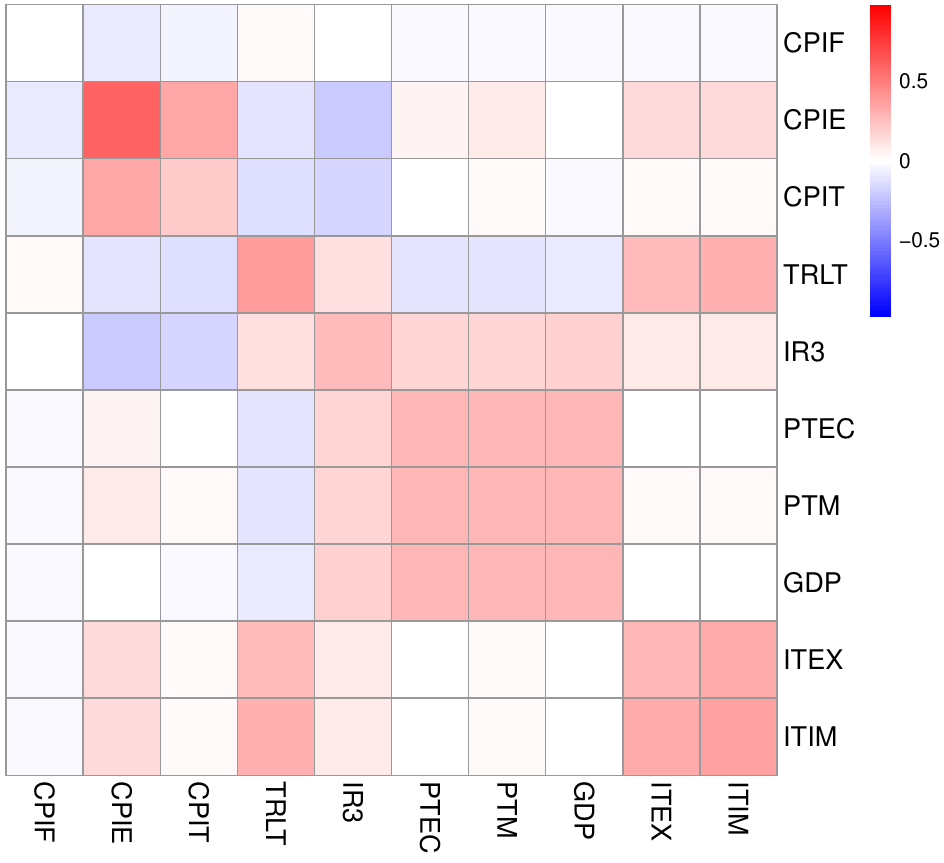}
        \caption{$\widehat{\bm R}_2\widehat{\bm R}_2^\top$}
        \label{subfig: R2}
    \end{subfigure}
    \caption{Estimates of common and specific projection matrices over economic indicators.}
    \label{fig: index 1}
\end{figure}

\begin{figure}[htb]
    \centering
    \begin{subfigure}[b]{0.32\textwidth} 
        \centering
        \includegraphics[width=\textwidth]{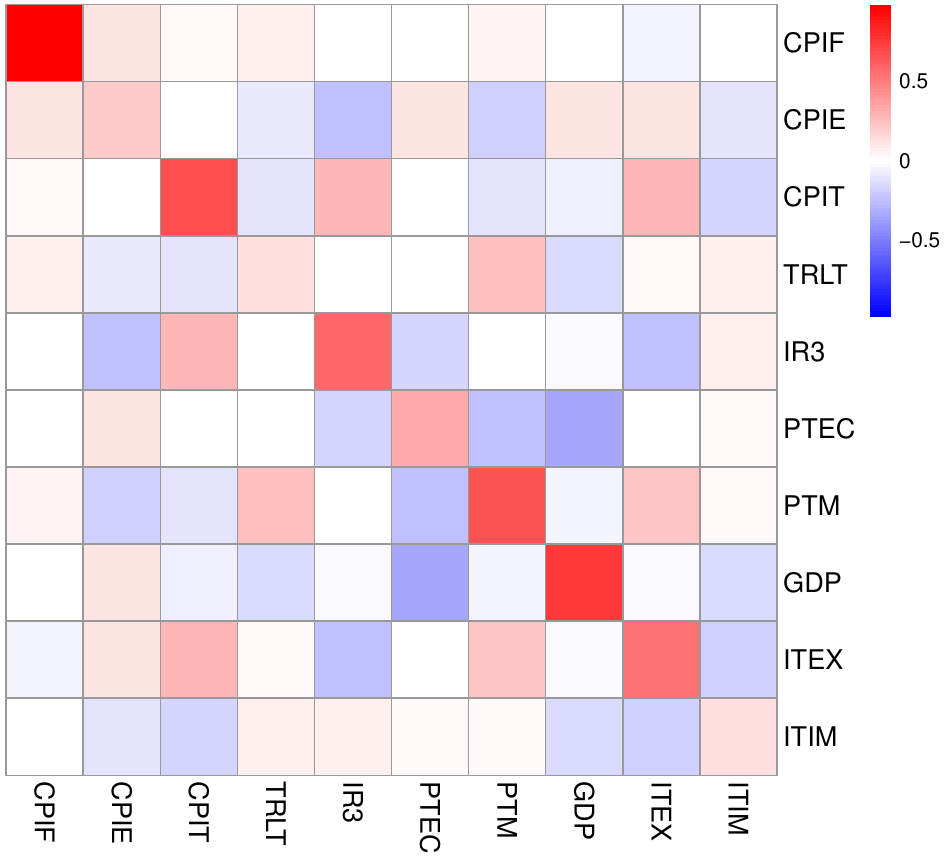}
        \caption{$\widehat{\bm C}_2\widehat{\bm C}_2^\top+\widehat{\bm P}_2\widehat{\bm P}_2^\top$}
        \label{subfig: C2_P2}
    \end{subfigure}
    \begin{subfigure}[b]{0.32\textwidth} 
        \centering
        \includegraphics[width=\textwidth]{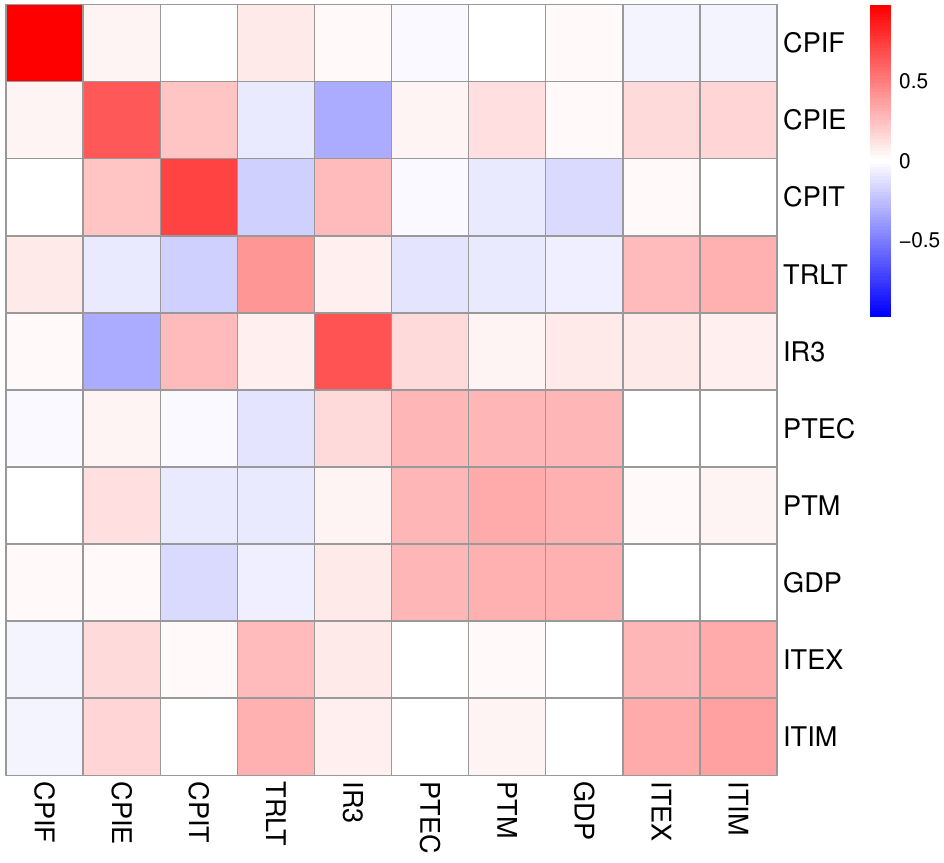}
        \caption{$\widehat{\bm C}_2\widehat{\bm C}_2^\top+\widehat{\bm R}_2\widehat{\bm R}_2^\top$}
        \label{subfig: C2_R2}
    \end{subfigure}
    \caption{Estimates of full predictor and response projection matrices over economic indicators.}
    \label{fig: index 2}
\end{figure}

\subsection{Forecasting Performance}

Finally, we evaluate the out-of-sample forecasting performance. We compute one-step-ahead forecast errors using 15 rolling windows covering the period from 2013-Q1 to 2016-Q3. We compare MARCF against the RRMAR model (with identical total ranks $r_1=4, r_2=5$) and the dynamic MFM under three specifications: (1) rank selected by information criterion ($r_1=r_2=1$) \citep{Han22rkInTSFM}; (2) ranks matching our common dimensions ($r_1=3, r_2=2$); and (3) a high-rank specification ($r_1=r_2=8$). For the dynamic MFM, after estimating the factor loading and process, we use a low-dimensional MAR model \eqref{eq:MAR1} to fit the process for prediction.

Table \ref{table:forecasting} reports the mean and median sums of squared out-of-sample prediction errors. The MARCF model outperforms all benchmarks. Notably, it improves upon RRMAR by exploiting the common subspace structure to reduce estimation variance, and it significantly outperforms dynamic MFM, which is constrained by the rigid assumption of symmetric predictor/response subspaces. In terms of model fit, the mean in-sample prediction errors of MFM(8,8) and RRMAR are slightly lower (239.90 and 242.31, respectively) compared to MARCF (243.25). Their inferior out-of-sample performance indicates overfitting. In contrast, MFM(1,1) and MFM(3,2) suffer from underfitting, yielding higher errors (334.30 and 276.32, respectively). The above empirical evidence supports the theoretical premise that real-world economic data occupies a middle ground, possessing both shared factors and distinct rotational dynamics, that is best captured by the hybrid MARCF framework.

\begin{table}[htp]
    \centering
    \caption{Means and medians of the out-of-sample forecast errors across 15 rolling windows.}
    \begin{tabular}{cccccc}
        \toprule
          & RRMAR & MARCF & MFM(1,1) & MFM(3,2) & MFM(8,8)\\
        \midrule
        Mean & 121.911 & \textbf{121.080} & 131.533 & 124.096 & 123.541\\
        Median & 69.410 & \textbf{67.697} & 83.657 & 77.006 & 75.403\\
        \bottomrule
    \end{tabular}
    \label{table:forecasting}
\end{table}

\section{Conclusion}

In this paper, we have introduced the Matrix Autoregressive model with Common Factors (MARCF), a unified framework for high-dimensional matrix time series that bridges the structural gap between reduced-rank autoregression and dynamic factor models. By explicitly parameterizing the intersection between predictor and response subspaces, the MARCF framework resolves the tension between the structural rigidity of reduced-rank models (which assume distinct subspaces) and the restrictiveness of factor models (which assume identical subspaces).

Our contributions can be summarized as follows:
\begin{itemize}
    \item \textbf{Methodological Unification:} We proposed a flexible decomposition that separates dynamics into common (shared), predictor-specific, and response-specific components. This allows the model to adaptively capture the ``subspace alignment'' of the data, providing a continuum that encompasses both RRMAR and DMF as special cases while enabling the modeling of complex interaction dynamics.
    \item \textbf{Scalable Estimation:} We developed a regularized gradient descent algorithm designed for the non-convex nature of the parameter space. By incorporating balanced regularization terms ($b \asymp \phi^{1/3}$), our approach resolves the scaling indeterminacies inherent in bilinear systems and avoids the computationally expensive matrix operations associated with alternating least squares, ensuring scalability to large-scale economic datasets.
    \item \textbf{Theoretical Rigor:} We established the statistical consistency of the estimator and the rank selection procedure. Our analysis decouples the geometric landscape of the loss function from the stochastic properties of the data, proving that the estimator achieves the optimal convergence rate governed by the effective degrees of freedom.
\end{itemize}

The empirical application to global macroeconomic indicators highlights the practical value of this framework. The MARCF model not only improved forecasting accuracy but also provided interpretable insights into the global economy, distinguishing between ``driver'' countries (dominating the predictor subspace) and ``responder'' countries (dominating the response subspace).

Several promising directions remain for future research. First, while we focused on dense subspace structures, incorporating sparsity constraints on the basis matrices ($\bm{C}, \bm{R}, \bm{P}$) could further enhance interpretability and efficiency, particularly when $p_1$ and $p_2$ are very large. Second, extending this framework to tensor-valued time series with an order $K \ge 3$ would allow for the modeling of more complex structures, though this introduces additional algebraic challenges regarding tensor rank selection. Finally, developing inferential tools, such as confidence intervals for the estimated factors or impulse response functions, remains an important challenge due to the non-standard asymptotics of singular subspace identification.

\setlength{\bibsep}{4pt}
\bibliography{mainbib.bib}

\counterwithin{equation}{section}
\counterwithin{lemma}{section}
\counterwithin{figure}{section}
\renewcommand{\thelemma}{\thesection.\arabic{lemma}}
\renewcommand{\theHlemma}{\thesection.\arabic{lemma}}
\AtAppendix{\counterwithin{lemma}{section}}
\AtAppendix{\counterwithin{proposition}{section}}

\appendix
\newpage 
\begin{center}
  \Large\bfseries Supplementary Material to ``A Hybrid Framework Combining Autoregression and Common Factors for Matrix Time Series"
\end{center}

This is the supplementary material for the main paper. We present detailed proofs of the theorems in the main text. Lemmas, their proofs, along with other technical tools, are also provided. Appendix \ref{append:modeling_procedure} summarizes the entire procedure for the proposed model. Appendix \ref{sec: gradients} provides the detailed expressions of the gradients used in the proposed algorithm. Appendix \ref{sec: computation converge} establishes the computational convergence of the algorithm via a deterministic analysis. Appendix \ref{sec:statistical_convergence} derives the statistical convergence rates of the final estimators, along with verification of the RSC/RSS condition and bounding of the deviation term and initialization errors via a stochastic analysis. Finally, Appendix \ref{sec:rank_selection} establishes the selection consistency of the rank selection method.

\section{Modeling Procedure}\label{append:modeling_procedure}

First, we provide a summary of the entire procedure for constructing the MARCF model on matrix-valued time series data. This is presented in the following Figure \ref{fig: modeling flowchart} to offer a clear and step-by-step guidance for the modeling process.

\begin{figure}[!htb]

\centering

\begin{tikzpicture}\label{fig: modeling flowchart}
\node[io] (input) {\textbf{Input}:\\
 centered and standardized stationary data $\mathcal{D}$; $\bar{r}_1,\bar{r}_2;$ step size $\eta;\lambda_1,\lambda_2,b$};

\node[process, below =1.5cm of input.center] (rank) {
    \textbf{Rank selection}:\\
    1. Estimate an RRMAR model with\\
    \hspace{1em} $\bar{r}_1,\bar{r}_2$ to obtain $\widehat{\bm A}_1^\mathrm{RR}(\bar{r}_1),\widehat{\bm A}_2^\mathrm{RR}(\bar{r}_2)$;\\
    2. Compute singular values and \\
    \hspace{1em} select $\widehat{r}_1,\widehat{r}_2$ by Section \ref{subsec: selection}.
};

\node[process, below=0.8cm of rank] (dim) {
    \textbf{Common dimension selection}:\\
    1. Estimate an RRMAR model with $\widehat{r}_1,\widehat{r}_2$;\\
    2. For every possible pair of $(d_1,d_2)$:\\
    \hspace{1em}\textbullet\hspace{0.05em} initialize $\bbm\Theta^{(0)}$ from $\widehat{\bm A}_1^\mathrm{RR}(\widehat{r}_1),\widehat{\bm A}_2^\mathrm{RR}(\widehat{r}_2)$;\\
    \hspace{1em}\textbullet\hspace{0.05em} run Algorithm \ref{algo: GD} with step size $\eta$\\
    \hspace{1.5em} and parameters $\lambda_1,\lambda_2,b$;\\
    3. Compute BIC and find $\widehat{d}_1,\widehat{d}_2$.
};

\node[process, below=1cm of dim] (estimation) {
    \textbf{Final estimation}: \\ Initialize $\bbm\Theta^{(0)}$ with $\widehat{r}_1,\widehat{r}_2,\widehat{d}_1,\widehat{d}_2$
    and run Algorithm \ref{algo: GD}.

};

\node[io, below=1cm of estimation] (output) {
    \textbf{Output}: \\
    estimates
    $\widehat{\bbm\Theta}=(\widehat{\bm C}_i,\widehat{\bm R}_i,\widehat{\bm P}_i,\widehat{\bm D}_i)_{i=1}^2$
};

\draw[arrow] (input) -- ++(0,-1.5) -| (rank.north);
\draw[arrow] (rank) -- (dim);
\draw[arrow] (dim) -- (estimation);
\draw[arrow] (estimation) -- (output);

\end{tikzpicture}

\caption{Modeling procedure of MARCF model.}
\label{fig: modeling flowchart}

\end{figure}

\newpage

\section{Detailed Expressions of Gradients in Algorithm \ref{algo: GD}}\label{sec: gradients}

In this appendix, we provide detailed expressions of the gradients in Algorithm \ref{algo: GD} which are omitted in the main article.

The objective of Algorithm \ref{algo: GD} is to minimize $\overline{\mathcal{L}}$ in \eqref{eq: loss} with known $r_1,r_2$ and $d_1,d_2$:
\begin{equation}\label{Penalized LS}
    \widehat{\bbm \Theta}:=\argmin \overline{\mathcal{L}}(\bm{\Theta};\lambda_1,\lambda_2,b)=\argmin \left\{\mathcal{L}(\bm{\Theta})+\frac{\lambda_1}{4}\mathcal{R}_1(\bbm\Theta)+\frac{\lambda_2}{2}\mathcal{R}_2(\bbm \Theta;b)\right\}.
\end{equation}
To derive the expressions of the gradients, we first define some useful notations. For any matrix $\bm{M}\in\mathbb{R}^{p_1\times p_1}$ and $\bm{N}\in\mathbb{R}^{p_2\times p_2}$, since all entries of $\bm M\otimes \bm N$ and $\V{\bm M}\V{\bm N}^\top$ are the same up to a permutation, we define the entry-wise re-arrangement operator $\mathcal{P}$ such that
\begin{equation}
    \mathcal{P}(\bm{M}\otimes\bm{N}) = \text{vec}(\bm{M})\text{vec}(\bm{N})^\top,
\end{equation}
where the dimensions of the matrices are omitted for simplicity.
The inverse operator $\mathcal{P}^{-1}$ is defined such that
\begin{equation}
    \mathcal{P}^{-1}(\text{vec}(\bm{M})\text{vec}(\bm{N})^\top) = \bm{M}\otimes\bm{N},
\end{equation}
where we use the notation vec($\cdot$) and mat($\cdot$) to denote the vectorization operator and its inverse matricization operator (with dimensions omitted for brevity), respectively.

The proposed MARCF model is given by
$$\bm{Y}_t = \bm{A}_1\bm{Y}_{t-1}\bm{A}_2^\top + \bm{E}_t, \quad t=1,\dots,T,$$
where $\bm{A}_i=[\bm{C}_i~\bm{R}_i]\bm{D}_i[\bm{C}_i~\bm{P}_i]^\top$, for $i=1,2$. The matrix time series can be vectorized and the model be rewritten as
\begin{equation}
    \bm{y}_t = (\bm{A}_2\otimes\bm{A}_1)\bm{y}_{t-1} + \bm{e}_t,~~~t=1,\dots,T,
\end{equation}
where $\bm y_t:=\V{\bm Y_t}$ and $\bm e_t:=\V{\bm E_t}$. We denote $\bm{A}:=\bm{A}_2\otimes\bm{A}_1$ to be the uniquely identifiable parameter matrix, and we also let $\bm Y:=[\bm y_T, \bm y_{T-1},...,\bm y_1]$ and $\bm X:=[\bm y_{T-1}, \bm y_{T-2},...,\bm y_0]$. 
Then, the loss function $\mathcal{L}$ in \eqref{eq: LS loss} can be viewed as a function of $\bm A$:
\begin{equation}
    \widetilde{\mathcal{L}}(\bm{A}) :=\mathcal{L}(\bbm\Theta)=\frac{1}{2T}\fnorm{\bm{Y}-\bm{A}\bm{X}}^2=\frac{1}{2T}\fnorm{\bm{Y}-(\bm  A_2\otimes \bm A_1)\bm{X}}^2.
\end{equation}

First, the gradient of $\widetilde{\mathcal{L}}$ with respect to $\bm A$ is
\begin{equation}
    \nabla\widetilde{\mathcal{L}}(\bm A):=\frac{\partial\widetilde{\mathcal{L}}}{\partial \bm A}=-\frac{1}{T}\sum_{t=1}^T(\bm{y}_t-\bm{A}\bm{y}_{t-1})\bm{y}_{t-1}^\top=-\frac{1}{T}(\bm{Y}-\bm{A}\bm{X})\bm{X}^\top.
\end{equation}
Let $\bm{a}_i=\V{\bm{A}_i}$ and $\bm{B}:=\V{\bm A_2}\otimes \V{\bm A_1}^\top=\bm{a}_2\bm{a}_1^\top$. Obviously, $\bm{A}=\mathcal{P}^{-1}(\bm{B})$ and $\bm{B}=\mathcal{P}(\bm{A})$. Define $\mathcal{L}_2(\bm{B})=\widetilde{\mathcal{L}}(\bm{A})$. Based on the entry-wise matrix permutation operators $\mathcal{P}$ and $\mathcal{P}^{-1}$, we have
\begin{equation}
    \begin{split}
        \nabla\mathcal{L}_2(\bm B) & = \mathcal{P}(\nabla\widetilde{\mathcal{L}}(\bm A))=\mathcal{P}(-(\bm{Y}-\bm{A}\bm{X})\bm{X}^\top/T),\\
        \nabla_{\bm{a}_1}\mathcal{L}_2 & = \mathcal{P}(-(\bm{Y}-\bm{A}\bm{X})\bm{X}^\top/T)^\top\bm{a}_2,\\
        \nabla_{\bm{a}_2}\mathcal{L}_2 & = \mathcal{P}(-(\bm{Y}-\bm{A}\bm{X})\bm{X}^\top/T)\bm{a}_1,\\
        \nabla_{\bm{A}_1}\widetilde{\mathcal{L}} & = \text{mat}(\mathcal{P}(-(\bm{Y}-\bm{A}\bm{X})\bm{X}^\top/T)^\top\text{vec}(\bm{A}_2))\\
        & = \text{mat}(\mathcal{P}(\nabla\widetilde{\mathcal{L}}(\bm{A}))^\top \text{vec}(\bm{A}_2)), \\
        \text{and  }\nabla_{\bm{A}_2}\widetilde{\mathcal{L}} & = \text{mat}(\mathcal{P}(-(\bm{Y}-\bm{A}\bm{X})\bm{X}^\top/T)\text{vec}(\bm{A}_1))\\
        & = \text{mat}(\mathcal{P}(\nabla\widetilde{\mathcal{L}}(\bm{A})) \text{vec}(\bm{A}_1)).
    \end{split}
\end{equation}

The gradients of $\widetilde{\mathcal{L}}$ with respect to $\bm C_i$, $\bm R_i$, $\bm P_i$, and $\bm D_i$ are:
\begin{equation}
    \begin{split}
        \nabla_{\bm{C}_i}\widetilde{\mathcal{L}} & = \nabla_{\bm{A}_i}\widetilde{\mathcal{L}}(\bm{C}_i\bm{D}_{i,11}^\top+\bm{P}_i\bm{D}_{i,12}^\top) + \nabla_{\bm{A}_i}\widetilde{\mathcal{L}}^\top(\bm{C}_i\bm{D}_{i,11}+\bm{R}_i\bm{D}_{i,21}),\\
        \nabla_{\bm{R}_i}\widetilde{\mathcal{L}} & = \nabla_{\bm{A}_i}\widetilde{\mathcal{L}}[\bm{C}_i~\bm{P}_i][\bm{D}_{i,21}~\bm{D}_{i,22}]^\top,\\
        \nabla_{\bm{P}_i}\widetilde{\mathcal{L}} & = \nabla_{\bm{A}_i}\widetilde{\mathcal{L}}^\top[\bm{C}_i~\bm{R}_i][\bm{D}_{i,12}^\top~\bm{D}_{i,22}^\top]^\top,\\
        \text{and  }\nabla_{\bm{D}_i}\widetilde{\mathcal{L}} & = [\bm{C}_i~\bm{R}_i]^\top\nabla_{\bm{A}_i}\widetilde{\mathcal{L}}[\bm{C}_i~\bm{P}_i].
    \end{split}
\end{equation}

Denote the following block structure of $\bm D_i$ as
$$
\bm D_i=\begin{pmatrix}
    \bm D_{i,11} & \bm D_{i,12}\\
    \bm D_{i,21} & \bm D_{i,22}
\end{pmatrix},
$$
where $\bm D_{i,11}\in\mathbb{R}^{d_i\times d_i}, \bm D_{i,12}\in\mathbb{R}^{d_i\times (r_i-d_i)}, \bm D_{i,21}\in\mathbb{R}^{(r_i-d_i)\times d_i}, \bm D_{i,22}\in\mathbb{R}^{(r_i-d_i)\times (r_i-d_i)}$ are block matrices in $\bm{D}_i$, for $i=1,2$. 
Finally, combined with the partial gradients of the regularization terms, the partial gradients of $\overline{\mathcal{L}}$ are given by:
\begin{equation}
    \begin{aligned}
    \nabla_{\bm C_i}\overline{\mathcal{L}}=&\nabla_{\bm{A}_i}\widetilde{\mathcal{L}}(\bm{C}_i\bm{D}_{i,11}^\top+\bm{P}_i\bm{D}_{i,12}^\top) + \nabla_{\bm{A}_i}\widetilde{\mathcal{L}}^\top(\bm{C}_i\bm{D}_{i,11}+\bm{R}_i\bm{D}_{i,21})\\
    &+(-1)^{i+1}\lambda_1\left(\fnorm{\bm A_1}^2-\fnorm{\bm A_2}^2\right)\left(\bm A_i[\bm C_i~\bm P_i][\bm D_{i,11}~\bm D_{i,12}]^\top+\bm A_i^\top[\bm C_i~\bm R_i][\bm D_{i,11}^\top~\bm D_{i,21}^\top]^\top\right)\\
    &+\lambda_2\left(2\bm{C}_i(\bm{C}_i^\top\bm{C}_i-b^2\bm{I}_{d_i})+\bm{R_i}\bm{R}_i^\top\bm{C}_i+\bm{P_i}\bm{P}_i^\top\bm{C}_i\right),\\
    \nabla_{\bm{R}_i}\overline{\mathcal{L}}=&\nabla_{\bm{A}_i}\widetilde{\mathcal{L}}[\bm{C}_i ~ \bm{P}_i][\bm{D}_{i,21}\; \bm{D}_{i,22}]^\top+(-1)^{i+1}\lambda_1\left(\fnorm{\bm A_1}^2-\fnorm{\bm A_2}^2\right)\bm A_i[\bm C_i~\bm P_i][\bm D_{i,21}~\bm D_{i,22}]^\top\\
    &+\lambda_2\left(\bm{R}_i(\bm{R}_i^\top\bm{R}_i-b^2\bm{I}_{r_i-d_i})+\bm{C}_i\bm{C}_i^\top\bm{R}_i\right),\\
    \nabla_{\bm{P}_i}\overline{\mathcal{L}}=&\nabla_{\bm{A}_i}\widetilde{\mathcal{L}}^\top[\bm{C}_i \; \bm{R}_i][\bm{D}_{i,12}^\top\; \bm{D}_{i,22}^\top]^\top+(-1)^{i+1}\lambda_1\left(\fnorm{\bm A_1}^2-\fnorm{\bm A_2}^2\right)\bm A_i^\top[\bm C_i~\bm R_i][\bm D_{i,12}^\top~\bm D_{i,22}^\top]^\top\\
    &+\lambda_2\left(\bm{P}_i(\bm{P}_i^\top\bm{P}_i-b^2\bm{I}_{r_i-d_i})+\bm{C_i}\bm{C}_i^\top\bm{P}_i\right),\\
    \nabla_{\bm{D}_i}\overline{\mathcal{L}}=&[\bm{C}_i \; \bm{R}_i]^\top\nabla_{\bm{A}_i}\widetilde{\mathcal{L}}[\bm{C}_i \; \bm{P}_i]+(-1)^{i+1}\lambda_1\left(\fnorm{\bm A_1}^2-\fnorm{\bm A_2}^2\right)[\bm C_i~\bm R_i]^\top\bm A_i[\bm C_i~\bm P_i].
    \end{aligned}
\end{equation}

\newpage

\section{Computational Convergence Analysis}\label{sec: computation converge}

In this appendix, we present the proof of Theorem \ref{theorem: computational convergence}, i.e., the computational convergence of regularized gradient decent algorithm proposed in the main article. To make it easy to read, it is divided into several steps. Auxiliary lemmas and their proofs are presented in Appendix \ref{sub:auxiliary_lemmas}.

\subsection{Proof of Theorem \ref{theorem: computational convergence}}

\textit{\textbf{Step 1.}} \textbf{(Notations, conditions, and proof outline)}

\noindent We begin by introducing some important notations and conditions required for the convergence analysis. Other notations not mentioned in this step are inherited from Appendix \ref{sec: gradients}. We then provide an outline of the proof, highlighting the key intermediate results and main ideas.

To begin, we restate the definitions of the measures quantifying the estimation error and the statistical error. Throughout this article, the true values of all parameters are defined by letters with an asterisk superscript, i.e., $\{\bm C_i^*,\bm R_i^*,\bm P_i^*,\bm D_i^*,\bm A_i^*\}_{i=1}^2$. Let $\mathbb{O}^{n_1\times n_2}$ be the set of $n_1\times n_2$ matrices with orthonormal columns. When $n_1=n_2=n$, we use $\mathbb{O}^n$ for short. For the $j$-th iterate, we quantify the combined estiamtion errors up to optimal rotations defined in \eqref{def: error} as
\begin{equation}\label{eq:running_error}
    \begin{split}
        \mathrm{dist}(\bbm\Theta^{(j)},\bbm\Theta^*)^2 = \min_{\substack{\bm O_{i,r},\bm O_{i,p}\in \mathbb{O}^{r_i-d_i}\\\bm O_{i,c}\in \mathbb{O}^{d_i}}}\sum_{i=1,2}\Big\{ & \|\bm{C}_i^{(j)}-\bm{C}^*\bm{O}_{i,c}\|_\text{F}^2 + \|\bm{R}_i^{(j)}-\bm{R}^*\bm{O}_{i,r}\|_\text{F}^2\\
        + \|\bm{P}_i^{(j)}-\bm{P}^*\bm{O}_{i,p}\|_\text{F}^2 + \|&\bm{D}_i^{(j)}-\text{diag}(\bm{O}_{i,c},\bm{O}_{i,r})^\top\bm{D}_i^*\text{diag}(\bm{O}_{i,c},\bm{O}_{i,p})\|_\text{F}^2\Big\},
    \end{split}
\end{equation}
and the corresponding optimal rotations as $\bm{O}_{i,c}^{(j)},\bm{O}_{i,r}^{(j)},\bm{O}_{i,p}^{(j)}$, for $i=1,2$. For simplicity, we let $\bm O_{i, u}:=\text{diag}(\bm{O}_{i,c},\bm{O}_{i,r})$, $\bm O_{i, v}:=\text{diag}(\bm{O}_{i,c},\bm{O}_{i,p})$, and use $\mathrm{dist}^2_{(j)}$ to represent $\mathrm{dist}(\bbm\Theta^{(j)},\bbm\Theta^*)^2$.

In addition, as defined in Definition \ref{def: statistical error xi}, the statistical error is quantified by
\begin{equation}
    \xi(r_1,r_2,d_1,d_2):=\sup_{\substack{[\bm C_i~\bm R_i]\in\mathbb{O}^{p_i\times r_i},\\
    [\bm C_i~\bm P_i]\in\mathbb{O}^{p_i\times r_i},\\\bm D_i\in\mathbb{R}^{r_i},\fnorm{\bm D_i}=1}}\inner{\nabla\widetilde{\mathcal{L}}(\bm{A}^*)}{[\bm{C}_2~\bm{R}_2]\bm{D}_2[\bm{C}_2~\bm{P}_2]^\top\otimes[\bm{C}_1~\bm{R}_1]\bm{D}_1[\bm{C}_1~\bm{P}_1]^\top},
\end{equation}
where $\bm{A}^*=\bm{A}_2^*\otimes\bm{A}_1^*$ and $\widetilde{\mathcal{L}}(\bm A)$ is referred to as the least-squares loss function with respect to the identifiable Kronecker product type parameter $\bm A=\bm A_2\otimes \bm A_1$.

Next, we state three conditions required for convergence analysis. The first condition stipulates that $\widetilde{\mathcal{L}}(\bm A)$ satisfies the RSC/RSS condition defined in \ref{def:RSCRSS}. Specifically, we assume there exist constants $0 < \alpha \leq \beta$ such that for any pair of identifiable parameter matrices $\bm A=\bm A_2\otimes \bm A_1$ and $\bm A'=\bm A_2'\otimes \bm A_1'$ admitting the MARCF structural decomposition, where $\bm A_i=[\bm C_i~\bm R_i]\bm D_i[\bm C_i~\bm P_i]^\top$ and $\bm A_i'=[\bm C_i'~\bm R_i']\bm D_i'[\bm C_i'~\bm P_i']^\top$, the following inequalities hold:
\begin{equation}
    \begin{gathered}
    \mathrm{(RSC)}\quad \frac{\alpha}{2}\fnorm{\bm A-\bm A'}^2\leq \mathcal{L}(\bm A)-\mathcal{L}(\bm A')-\inner{\nabla\mathcal{L}(\bm A')}{\bm A-\bm A'},\\
    \mathrm{(RSS)}\quad \mathcal{L}(\bm A)-\mathcal{L}(\bm A')-\inner{\nabla\mathcal{L}(\bm A')}{\bm A-\bm A'}\leq \frac{\beta}{2}\fnorm{\bm A-\bm A'}^2,
    \end{gathered}
\end{equation}
The condition is a premise of Theorem \ref{theorem: computational convergence}. As in \cite{Nesterov2004}, it implies that 
$$\mathcal{L}\left(\mathbf{A}'\right)-\mathcal{L}(\mathbf{A}) \geq\left\langle\nabla \mathcal{L}(\mathbf{A}), \mathbf{A}'-\mathbf{A}\right\rangle+\frac{1}{2 \beta}\left\|\nabla \mathcal{L}\left(\mathbf{A}'\right)-\nabla \mathcal{L}(\mathbf{A})\right\|_{\mathrm{F}}^2.$$
Combining this inequality with the RSC condition, we have
\begin{equation}\label{RCG}
    \begin{aligned}
    &\left\langle\nabla \mathcal{L}(\mathbf{A})-\nabla \mathcal{L}\left(\mathbf{A}'\right), \mathbf{A}-\mathbf{A}'\right\rangle \geq \frac{\alpha}{2}\left\|\mathbf{A}-\mathbf{A}'\right\|_{\mathrm{F}}^2+\frac{1}{2 \beta}\left\|\nabla \mathcal{L}(\mathbf{A})-\nabla \mathcal{L}\left(\mathbf{A}'\right)\right\|_{\mathrm{F}}^2,
\end{aligned}\end{equation}
which is equivalent to the restricted correlated gradient (RCG) condition in \citep{Han2022}.

The second and the third conditions ensure that the estimates remain within a small neighborhood of the true values during all iterations. They are established inductively in the proof.
Let $\underline{\sigma}:=\min\{\sigma_{r_1}(\bm A_1^*), \sigma_{r_2}(\bm A_2^*)\}$ be the smallest value among all the non-zero singular values of $\bm A_1^*$ and $\bm A_2^*$. Define $\phi:=\fnorm{\bm A_1^*}=\fnorm{\bm A_2^*}$  and $\kappa:=\phi/\underline{\sigma}$. With the notations, the second condition is specified as
\begin{equation}\label{condition:E_upperBound}
    \mathrm{dist}^2_{(j)}\leq \frac{C_D\alpha\phi^{2/3}}{\beta\kappa^2},
\end{equation}
and then
\begin{equation}
    \mathrm{dist}^2_{(j)}\leq \frac{C_D\phi^{2/3}}{\kappa^2}\leq C_D\phi^{2/3},\text{ for all } j=0,1,2,\dots
\end{equation}
where $C_D$ is positive constant small enough to ensure local convergence.

For the decomposition of $\bm A_i^*$, we require $[\bm C_i^*~\bm R_i^*]^\top[\bm C_i^*~\bm R_i^*]=b^2\bm I_{r_i}$, $[\bm C_i^*~\bm P_i^*]^\top[\bm C_i^*~\bm P_i^*]=b^2\bm I_{r_i},$ for $i=1,2$. For simplicity, we consider $b=\phi^{1/3}$, though our proof can be readily extended to the case where $b\asymp \phi^{1/3}$. 

The third condition is:
\begin{equation}\label{condition: upper bound of pieces}\begin{gathered}
     \fnorm{\left[\bm C_i^{(j)}~ \bm R_i^{(j)}\right]}\leq (1+c_b)b,\quad  \fnorm{\left[\bm C_i^{(j)}~ \bm R_i^{(j)}\right]}\leq (1+c_b)b,\\
     \fnorm{\bm D_i^{(j)}}\leq \frac{(1+c_b)\phi}{b^2},\quad \forall i=1,2,\text{ and }j=0,1,2,...
\end{gathered}\end{equation}
By sub-multiplicative property of Frobenius norm, $\fnorm{\bm A_1^{(j)}}\leq (1+c_a)\phi$ and $\fnorm{\bm A_2^{(j)}}\leq (1+c_a)\phi$ with $c_b$ and $c_a$ being two positive constants. We assume that $c_b\leq 0.01$ and then $c_a\leq 0.04$. In fact, the constant $0.01$ reflects the accuracy of the initial estimate, which can be replaced by any small positive numbers. 

Finally, we give an outline of the proof of Theorem \ref{theorem: computational convergence}, which proceeds by induction. Assuming that \eqref{condition:E_upperBound} and \eqref{condition: upper bound of pieces} hold for $\bbm\Theta^{(j)}$, based on on RSC and RSS conditions, we derive an recursive relationship between $\mathrm{dist}^2_{(j+1)}$ and $\mathrm{dist}^2_{(j)}$ with an extra statistical error term. Then, we prove that $\eqref{condition:E_upperBound}$ and $\eqref{condition: upper bound of pieces}$ hold for $\bbm\Theta^{(j+1)}$ and verify the conditions for $\bbm\Theta^{(0)}$, thereby finishing the inductive argument. Finally, built on the recursive relationship, we establish the three upper bounds presented Theorem \ref{theorem: computational convergence}.

Specifically, in Steps 2-4, we focus on the update at $(j+1)$-th iterate to establish a recursive inequality between $\mathrm{dist}_{(j+1)}^2$ and $\mathrm{dist}_{(j)}^2$, given \eqref{RCG}, \eqref{condition:E_upperBound}, and \eqref{condition: upper bound of pieces}. In Step 2, for $\bm R_1$, we use the rule $\bm R_1^{(j+1)}=\bm R_1^{(j)}-\eta\nabla_{\bm R_1}\overline{\mathcal{L}}$ and upper bound the estimation error of $\bm R^{(j+1)}$ as
\begin{equation}
    \begin{split}
        \min_{\bm O_{1,r}\in\mathbb{O}^{r_1-d_1}}\fnorm{\bm R_1^{(j+1)}-\bm R_1^*\bm O_{1,r}}^2 \leq& \min_{\bm O_{1,r}\in\mathbb{O}^{r_1-d_1}}\fnorm{\bm R_1^{(j)}-\bm R_1^*\bm O_{1,r}}^2\\ 
        &-2\eta Q_{\bm R_1,1} -2\lambda_1\eta G_{\bm R_1}-2\lambda_2\eta T_{\bm R_1}+ \eta^2 Q_{\bm R_1,2}.
    \end{split}
\end{equation}
Similarly, we do the same for $\bm P_i^{(j)}, \bm C_i^{(j)}$ and $\bm D_i^{(j)}$ and sum them up to obtain the following first-stage upper bound, which is an informal version of \eqref{E upper bound: Q2Q1GT}.
\begin{equation}
\begin{aligned}
    \mathrm{dist}^2_{(j+1)}&\leq \mathrm{dist}^2_{(j)}+\eta^2Q_2
    -2\eta Q_1
    -2\lambda_1\eta G
    -2\lambda_2\eta T.
\end{aligned}
\end{equation}

In Steps 3.1-3.4, we give further the lower bounds for $Q_1$, $G$, and $T$, whose coefficients are $-2\eta$, and the upper bound for $Q_2$, whose coefficient is $\eta^2$. They lead to an intermediate upper bound in \eqref{E upper bound 1}, whose right hand side is negatively correlated to the estimation error of $\bm A_2^{(j)}\otimes \bm A_1^{(j)}$. Thus, in Step 3.5, we construct a lower bound with respect to $\mathrm{dist}^2_{(j)}$ and the regularization terms, $\mathcal{R}_1(\bbm\Theta^{(j)})^2$ and $\mathcal{R}_2({\bbm\Theta^{(j)}})^2$. Plugging it into the intermediate bound, we have the second-stage upper bound of $\mathrm{dist}_{(j+1)}^2$ as follows, which is an informal version of \eqref{E upper bound 2}
\begin{equation}
    \begin{aligned}
        \mathrm{dist}^2_{(j+1)}\leq&\left(1-2\eta\rho_1\right)\mathrm{dist}^2_{(j)}+\rho_2\xi^2\\
        &+\rho_3\fnorm{\nabla\widetilde{\mathcal{L}}(\bm A^{(j)})-\nabla\widetilde{\mathcal{L}}(\bm A^*)}^2+\rho_4\mathcal{R}_2\left({\bbm\Theta^{(j)}}\right)^2+\rho_5\mathcal{R}_1\left(\bbm\Theta^{(j)}\right)^2.
    \end{aligned}
\end{equation}
In the above inequality, $\{\rho_i\}_{i=1}^5$ are coefficients related to $\eta$, $\lambda_1$, $\lambda_2$, $\alpha$, $\beta$, $\phi$, and $\kappa$.

Next, in Step 4, we impose some sufficient conditions on $\eta$, $\lambda_1,\lambda_2$, $\alpha$, and $\beta$ for algorithm convergence. Then, we obtain the final recursive relationship \eqref{recursive of E}:
\begin{equation}
    \mathrm{dist}^2_{(j+1)}\leq (1-C_0\eta_0\alpha\beta^{-1}\kappa^{-2})\mathrm{dist}^2_{(j)}+C\kappa^4\alpha^{-2}\phi^{-10/3}\xi^2.
\end{equation}

To complete the inductive argument, in Step 5, we first verify the conditions \eqref{condition:E_upperBound} and \eqref{condition: upper bound of pieces} for $\bbm\Theta^{(0)}$. Given that \eqref{condition:E_upperBound} is an assumption of Theorem 3, it suffices to verify \eqref{condition: upper bound of pieces}. Then, we prove that both conditions hold for $\bbm\Theta^{(j+1)}$ and complete the induction. Finally, in Step 6, we establish the three upper bounds presented in Theorem \ref{theorem: computational convergence} based on the recursive relationship and the approximate equivalence of the three metrics, thereby completing the whole proof.\\

\noindent \textbf{\textit{Step 2.}} (\textbf{Upper bound of $\mathrm{dist}^2_{(j+1)}-\mathrm{dist}^2_{(j)}$})

\noindent By definition,
\begin{equation}
    \begin{split}
        & \mathrm{dist}^2_{(j+1)}\\
        = & \sum_{i=1,2}\Big\{ \|\bm{R}_i^{(j+1)}-\bm{R}_i^*\bm{O}_{1,r}^{(j+1)}\|_\text{F}^2 + \|\bm{P}_i^{(j+1)}-\bm{P}_i^*\bm{O}_{1,p}^{(j+1)}\|_\text{F}^2 + \|\bm{C}_i^{(j+1)}-\bm{C}^*\bm{O}_{1,c}^{(j+1)}\|_\text{F}^2\\
        & + \|\bm{D}_i^{(j+1)}-\bm{O}_{1,u}^{(j+1)\top}\bm{D}^*\bm{O}_{1,v}^{(j+1)}\|_\text{F}^2\Big\}\\
        \leq & \sum_{i=1,2}\Big\{ \|\bm{R}_i^{(j+1)}-\bm{R}_i^*\bm{O}_{1,r}^{(j)}\|_\text{F}^2 + \|\bm{P}_i^{(j+1)}-\bm{P}_i^*\bm{O}_{1,p}^{(j)}\|_\text{F}^2 + \|\bm{C}_i^{(j+1)}-\bm{C}^*\bm{O}_{1,c}^{(j)}\|_\text{F}^2\\
        & + \|\bm{D}_i^{(j+1)}-\bm{O}_{1,u}^{(j)\top}\bm{D}^*\bm{O}_{1,v}^{(j)}\|_\text{F}^2\Big\}.
    \end{split}
\end{equation}
In the following substeps, we derive upper bounds for the errors of $\bm C,\bm R,\bm P$ and $\bm D$ separately. Then, we combine them to give a first-stage upper bound \eqref{E upper bound: Q2Q1GT}.

~\newline
\noindent\textit{Step 2.1} (Upper bounds for the errors of $\bm{R}_i$ and $\bm{P}_i$)

\noindent By definition, we have
\begin{equation}\label{eq:R_split}
    \begin{aligned}
        & \|\bm{R}_1^{(j+1)}-\bm{R}_1^*\bm{O}_{1,r}^{(j)}\|_\text{F}^2 \\
        = & \Big\|\bm{R}_1^{(j)}-\bm{R}_1^*\bm{O}_{1,r}^{(j)} -\eta\Big(\nabla_{\bm{A}_1}\widetilde{\mathcal{L}}(\bm{A}^{(j)})[\bm{C}_1^{(j)}~\bm{P}_1^{(j)}][\bm{D}_{1,21}^{(j)}~\bm{D}_{1,22}^{(j)}]^\top\\
        &+\lambda_1\left(\fnorm{\bm A_1^{(j)}}^2-\fnorm{\bm A_2^{(j)}}^2\right)\bm A_i^{(j)}[\bm C_i^{(j)}~\bm P_i^{(j)}][\bm D_{1,21}^{(j)}~\bm D_{1,22}^{(j)}]^\top \\
         & + \lambda_2\bm{R}_1^{(j)}(\bm{R}_1^{(j)\top}\bm{R}_1^{(j)}-b^2\bm{I}_{r_1-d_1}) + \lambda_2\bm{C}_1^{(j)}\bm{C}_1^{(j)\top}\bm{R}_1^{(j)} \Big)\Big\|_\text{F}^2\\
        = & \|\bm{R}_1^{(j)}-\bm{R}_1^*\bm{O}_{1,r}^{(j)}\|_\text{F}^2 + \eta^2\Big\|\nabla_{\bm{A}_1}\widetilde{\mathcal{L}}(\bm{A}^{(j)})[\bm{C}_1^{(j)}~\bm{P}_1^{(j)}][\bm{D}_{1,21}^{(j)}~\bm{D}_{1,22}^{(j)}]^\top \\
        &+\lambda_1\left(\fnorm{\bm A_1^{(j)}}^2-\fnorm{\bm A_2^{(j)}}^2\right)\bm A_i^{(j)}[\bm C_i^{(j)}~\bm P_i^{(j)}][\bm D_{1,21}^{(j)}~\bm D_{1,22}^{(j)}]^\top\\
        & + \lambda_2\bm{R}_1^{(j)}(\bm{R}_1^{(j)\top}\bm{R}_1^{(j)}-b^2\bm{I}_{r_1-d_1}) + \lambda_2\bm{C}_1^{(j)}\bm{C}_1^{(j)\top}\bm{R}_1^{(j)}\Big\|_\text{F}^2\\
        - & 2\eta\Big\langle\bm{R}_1^{(j)}-\bm{R}_1^*\bm{O}_{1,r}^{(j)},\nabla_{\bm{A}_1}\widetilde{\mathcal{L}}(\bm{A}^{(j)})[\bm{C}_1^{(j)}~\bm{P}_1^{(j)}][\bm{D}_{1,21}^{(j)}~\bm{D}_{1,22}^{(j)}]^\top\Big\rangle\\
        - & 2\lambda_1\eta\inner{\bm{R}_1^{(j)}-\bm{R}_1^*\bm{O}_{1,r}^{(j)}}{\left(\fnorm{\bm A_1^{(j)}}^2-\fnorm{\bm A_2^{(j)}}^2\right)\bm A_i^{(j)}[\bm C_i^{(j)}~\bm P_i^{(j)}][\bm D_{1,21}^{(j)}~\bm D_{1,22}^{(j)}]^\top}\\
        - & 2\lambda_2\eta\Big\langle\bm{R}_1^{(j)}-\bm{R}_1^*\bm{O}_{1,r}^{(j)}, \bm{R}_1^{(j)}(\bm{R}_1^{(j)\top}\bm{R}_1^{(j)}-b^2\bm{I}_{r_1-d_1})\Big\rangle\\
        - & 2\lambda_2\eta\Big\langle\bm{R}_1^{(j)}-\bm{R}_1^*\bm{O}_{1,r}^{(j)}, \bm{C}_1^{(j)}\bm{C}_1^{(j)\top}\bm{R}_1^{(j)}\Big\rangle\\
        :=&\fnorm{\bm{R}_1^{(j)}-\bm{R}_1^*\bm{O}_{1,r}^{(j)}}^2+\eta^2 I_{\bm R_1,2}-2\eta I_{\bm R_1,1}.
    \end{aligned}
{}\end{equation}

For $I_{\bm R_1,2}$ (the second term in \eqref{eq:R_split}), by Cauchy's inequality,
\begin{equation}
    \begin{aligned}
        I_{\bm R_1,2}=& \Big\|\nabla_{\bm{A}_1}\widetilde{\mathcal{L}}(\bm{A}^{(j)})[\bm{C}_1^{(j)}~\bm{P}_1^{(j)}][\bm{D}_{1,21}^{(j)}~\bm{D}_{1,22}^{(j)}]^\top\\
        &+\lambda_1\left(\fnorm{\bm A_1^{(j)}}^2-\fnorm{\bm A_2^{(j)}}^2\right)\bm A_i^{(j)}[\bm C_i^{(j)}~\bm P_i^{(j)}][\bm D_{1,21}^{(j)}~\bm D_{1,22}^{(j)}]^\top\\
        &+ \lambda_2\bm{R}_1^{(j)}(\bm{R}_1^{(j)\top}\bm{R}_1^{(j)}-b^2\bm{I}_{r_1-d_1}) + \lambda_2\bm{C}_1^{(j)}\bm{C}_1^{(j)\top}\bm{R}_1^{(j)}\Big\|_\text{F}^2 \\ 
        \leq & 4 \|\nabla_{\bm{A}_1}\widetilde{\mathcal{L}}(\bm{A}^{(j)})[\bm{C}_1^{(j)}~\bm{P}_1^{(j)}][\bm{D}_{1,21}^{(j)}~\bm{D}_{1,22}^{(j)}]^\top\|_\text{F}^2\\
        & + 4\lambda_1^2\left(\fnorm{\bm A_1^{(j)}}^2-\fnorm{\bm A_2^{(j)}}^2\right)^2\fnorm{\bm A_i^{(j)}[\bm C_i^{(j)}~\bm P_i^{(j)}][\bm D_{1,21}^{(j)}~\bm D_{1,22}^{(j)}]^\top}^2\\
        & + 4\lambda_2^2\|\bm{R}_1^{(j)}(\bm{R}_1^{(j)\top}\bm{R}_1^{(j)}-b^2\bm{I}_{r_1-d_1})\|_\text{F}^2+4\lambda_2^2\|\bm{C}_1^{(j)}\bm{C}_1^{(j)\top}\bm{R}_1^{(j)}\|_\text{F}^2,
    \end{aligned}
\end{equation}
where the first term in the RHS of $I_{\bm R_1,2}$ can be bounded by
\begin{equation}
    \begin{split}
        & \|\nabla_{\bm{A}_1}\widetilde{\mathcal{L}}(\bm{A}^{(j)})[\bm{C}_1^{(j)}~\bm{P}_1^{(j)}][\bm{D}_{1,21}^{(j)}~\bm{D}_{1,22}^{(j)}]^\top\|_\text{F}^2\\
        = & \|\text{mat}(\mathcal{P}(\nabla\widetilde{\mathcal{L}}(\bm{A}^{(j)}))^\top\text{vec}(\bm{A}_2^{(j)}))[\bm{C}_1^{(j)}~\bm{P}_1^{(j)}][\bm{D}_{1,21}^{(j)}~\bm{D}_{1,22}^{(j)}]^\top\|_\text{F}^2 \\
        \leq & 2\|\text{mat}(\mathcal{P}(\nabla\widetilde{\mathcal{L}}(\bm{A}^{*}))^\top\text{vec}(\bm{A}_2^{(j)}))[\bm{C}_1^{(j)}~\bm{P}_1^{(j)}][\bm{D}_{1,21}^{(j)}~\bm{D}_{1,22}^{(j)}]^\top\|_\text{F}^2 \\
        & + 2\|\text{mat}(\mathcal{P}(\nabla\widetilde{\mathcal{L}}(\bm{A}^{(j)})-\nabla\widetilde{\mathcal{L}}(\bm{A}^{*}))^\top\text{vec}(\bm{A}_2^{(j)}))[\bm{C}_1^{(j)}~\bm{P}_1^{(j)}][\bm{D}_{1,21}^{(j)}~\bm{D}_{1,22}^{(j)}]^\top\|_\text{F}^2.\\
    \end{split}
\end{equation}
By duality of the Frobenius norm and definition of $\xi$, we have
\begin{equation}
    \begin{split}
        & \Big\|\text{mat}(\mathcal{P}(\nabla\widetilde{\mathcal{L}}(\bm{A}^*))^\top\text{vec}(\bm{A}_2^{(j)}))[\bm{C}_1^{(j)}~\bm{P}_1^{(j)}][\bm{D}_{1,21}^{(j)}~\bm{D}_{1,22}^{(j)}]^\top\Big\|_\text{F}\\
        = & \sup_{\|\bm{W}\|_\text{F}=1}\left\langle\nabla\widetilde{\mathcal{L}}(\bm{A}^*),\bm{A}_2^{(j)}\otimes [\bm{C}_1^{(j)}~\bm{W}]\begin{bmatrix} \bm{0} & \bm{0} \\ \bm{D}_{1,21}^{(j)} & \bm{D}_{1,22}^{(j)} \end{bmatrix} [\bm{C}_1^{(j)}~\bm{P}_1^{(j)}]^\top \right\rangle\\
        \leq & \|\bm{A}_2^{(j)}\|_\text{F}\cdot\|[\bm{D}_{1,21}^{(j)}~\bm{D}_{1,22}^{(j)}]\|_\text{op}\cdot\|[\bm{C}_{1}^{(j)}~\bm{P}_{1}^{(j)}]\|_\text{op}\cdot\xi(r_1,r_2,d_1,d_2).
    \end{split}
\end{equation}
Thus, the first term can be bounded by
\begin{equation}
    \begin{split}
        & \|\nabla_{\bm{A}_1}\widetilde{\mathcal{L}}(\bm{A}^{(j)})[\bm{C}_1^{(j)}~\bm{P}_1^{(j)}][\bm{D}_{1,21}^{(j)}~\bm{D}_{1,22}^{(j)}]^\top\|_\text{F}^2\\
        \leq & 2\|\text{mat}(\mathcal{P}(\nabla\widetilde{\mathcal{L}}(\bm{A}^{*}))^\top\text{vec}(\bm{A}_2^{(j)}))[\bm{C}_1^{(j)}~\bm{P}_1^{(j)}][\bm{D}_{1,21}^{(j)}~\bm{D}_{1,22}^{(j)}]^\top\|_\text{F}^2 \\
        & + 2\|\text{mat}(\mathcal{P}(\nabla\widetilde{\mathcal{L}}(\bm{A}^{(j)})-\nabla\widetilde{\mathcal{L}}(\bm{A}^{*}))^\top\text{vec}(\bm{A}_2^{(j)}))\|_\text{F}^2\cdot\|[\bm{C}_1^{(j)}~\bm{P}_1^{(j)}]\|_\text{op}^2\cdot\|[\bm{D}_{1,21}^{(j)}~\bm{D}_{1,22}^{(j)}]\|_\text{F}^2 \\
        \leq & 4b^{-2}\phi^4\cdot[\xi^2(r_1,r_2,d_1,d_2)+\|\nabla\widetilde{\mathcal{L}}(\bm{A}^{(j)})-\nabla\widetilde{\mathcal{L}}(\bm{A}^*)\|_\text{F}^2].
    \end{split}
\end{equation}
The second term can be bounded by
\begin{equation}
    \begin{aligned}
        & \lambda_1^2\left(\fnorm{\bm A_1^{(j)}}^2-\fnorm{\bm A_2^{(j)}}^2\right)^2\fnorm{\bm A_i^{(j)}[\bm C_i^{(j)}~\bm P_i^{(j)}][\bm D_{1,21}^{(j)}~\bm D_{1,22}^{(j)}]^\top}^2 \\ 
        & \leq 2\lambda_1^2b^{-2}\phi^4\left(\fnorm{\bm A_1^{(j)}}^2-\fnorm{\bm A_2^{(j)}}^2\right)^2.
    \end{aligned}
\end{equation}
The third term can be bounded by
\begin{equation}
    \begin{split}
        & \lambda_2^2\|\bm{R}_1^{(j)}(\bm{R}_1^{(j)\top}\bm{R}_1^{(j)}-b^2\bm{I}_{r_1-d_1})\|_\text{F}^2 \leq \lambda_2^2\|\bm{R}_1^{(j)}\|_\text{op}^2\|\bm{R}_1^{(j)\top}\bm{R}_1^{(j)}-b^2\bm{I}_{r_1-d_1}\|_\text{F}^2\\
        \leq & 2\lambda_2^2b^2\|\bm{R}_1^{(j)\top}\bm{R}_1^{(j)}-b^2\bm{I}_{r_1-d_1}\|_\text{F}^2,
    \end{split}
\end{equation}
and the fourth term can be bounded by
\begin{equation}
    \lambda_2^2\|\bm{C}_1^{(j)}\bm{C}_1^{(j)\top}\bm{R}_1^{(j)}\|_\text{F}^2 \leq 2\lambda_2^2b^2\|\bm{C}_1^{(j)\top}\bm{R}_1^{(j)}\|_\text{F}^2.
\end{equation}
Combining these four upper bounds, we have
\begin{equation}\label{eq:IR12_bound}
    \begin{split}
        I_{\bm R_1,2}
        \leq& 16b^{-2}\phi^4\left(\xi^2(r_1,r_2,d_1,d_2)+\|\nabla\widetilde{\mathcal{L}}(\bm{A}^{(j)})-\nabla\widetilde{\mathcal{L}}(\bm{A}^{*})\|_\text{F}^2\right)\\
        &+ 8\lambda_2^2b^2\left(\|\bm{R}_1^{(j)\top}\bm{R}_1^{(j)}-b^2\bm{I}_{r_1-d_1}\|_\text{F}^2 + \|\bm{C}_1^{(j)\top}\bm{R}_1^{(j)}\|_\text{F}^2\right)\\
        &+8\lambda_1^2b^{-2}\phi^4\left(\fnorm{\bm A_1^{(j)}}^2-\fnorm{\bm A_2^{(j)}}^2\right)^2\\
        :=&Q_{\bm R_1,2}.
    \end{split}
\end{equation}

For $I_{\bm R_1,1}$ defined in \eqref{eq:R_split}, rewrite its first term:
\begin{equation}
    \begin{split}
        & \left\langle\bm{R}_1^{(j)}-\bm{R}_1^*\bm{O}_{1,r}^{(j)},\nabla_{\bm{A}_1}\widetilde{\mathcal{L}}(\bm{A}^{(j)})[\bm{C}_1^{(j)}~\bm{P}_1^{(j)}][\bm{D}_{1,21}^{(j)}~\bm{D}_{1,22}^{(j)}]^\top\right\rangle \\
        = & \left\langle\bm{R}_1^{(j)}[\bm{D}_{1,21}^{(j)}~\bm{D}_{1,22}^{(j)}][\bm{C}_1^{(j)}~\bm{P}_1^{(j)}]^\top-\bm{R}_1^*\bm{O}_{1,r}^{(j)}[\bm{D}_{1,21}^{(j)}~\bm{D}_{1,22}^{(j)}][\bm{C}_1^{(j)}~\bm{P}_1^{(j)}]^\top,\nabla_{\bm{A}_1}\widetilde{\mathcal{L}}(\bm{A}^{(j)})\right\rangle \\
        := & \left\langle\bm{A}_{1,r}^{(j)},\nabla_{\bm{A}_1}\widetilde{\mathcal{L}}(\bm{A}^{(j)})\right\rangle\\
        =&\inner{\bm{A}_{1,r}^{(j)}}{\text{mat}(\mathcal{P}(\nabla\widetilde{\mathcal{L}}(\bm{A}^{(j)}))^\top \text{vec}(\bm{A}_2^{(j)}))}\\
        =&\inner{\V{\bm{A}_{1,r}^{(j)}}}{\mathcal{P}(\nabla\widetilde{\mathcal{L}}(\bm{A}^{(j)}))^\top\V{\bm{A}_2^{(j)}}}\\
        =&\inner{\V{\bm A_2^{(j)}}\V{\bm{A}_{1,r}^{(j)}}^\top}{\mathcal{P}(\nabla\widetilde{\mathcal{L}}(\bm{A}^{(j)}))}\\
        =&\inner{\bm A_2^{(j)}\otimes\bm A_{1,r}^{(j)}}{\nabla\widetilde{\mathcal{L}}(\bm{A}^{(j)})}\\
        :=&Q_{\bm R_1,1}
    \end{split}
\end{equation}
For the second term of $I_{\bm R_1,1}$, define
\begin{equation}
    \label{eq:G_R1}
    G_{\bm R_1}:=\inner{\bm{R}_1^{(j)}-\bm{R}_1^*\bm{O}_{1,r}^{(j)}}{\left(\fnorm{\bm A_1^{(j)}}^2-\fnorm{\bm A_2^{(j)}}^2\right)\bm A_i^{(j)}[\bm C_i^{(j)}~\bm P_i^{(j)}][\bm D_{1,21}^{(j)}~\bm D_{1,22}^{(j)}]^\top}.    
\end{equation}    
    
For the third and fourth terms of $I_{\bm R_1,1}$, define
\begin{equation}
    T_{\bm R_1} := \left\langle\bm{R}_1^{(j)}-\bm{R}_1^*\bm{O}_{1,r}^{(j)},\bm{R}_1^{(j)}(\bm{R}_1^{(j)\top}\bm{R}_1^{(j)}-b^2\bm{I}_{r_1-d_1}) + \bm{C}_1^{(j)}\bm{C}_1^{(j)\top}\bm{R}_1^{(j)}\right\rangle.
\end{equation}
Therefore, we can rewrite $I_{\bm R_1,1}$ as
\begin{equation}
    I_{\bm R_1,1}=Q_{\bm R_1,1}+\lambda_1 G_{\bm R_1}+\lambda_2T_{\bm R_1}.
\end{equation}
Combining the bound for the $I_{\bm R_1,2}$ in \eqref{eq:IR12_bound}, we have
\begin{equation}
    \|\bm{R}_1^{(j+1)}-\bm{R}_1^*\bm{O}_{1,r}^{(j)}\|_\text{F}^2 - \|\bm{R}_1^{(j)}-\bm{R}_1^*\bm{O}_{1,r}^{(j)}\|_\text{F}^2 \leq -2\eta Q_{\bm R_1,1} -2\lambda_1\eta G_{\bm R_1}-2\lambda_2\eta T_{\bm R_1}+ \eta^2 Q_{\bm R_1,2}.
\end{equation}
Similarly, for $\bm{R}_2^{(j+1)}$, $\bm{P}_1^{(j+1)}$, and $\bm{P}_2^{(j+1)}$, we can define the similar quantities and show that
\begin{equation}\label{eq:R_P_bound}
    \begin{split}
        &\|\bm{R}_2^{(j+1)}-\bm{R}_2^*\bm{O}_{2,r}^{(j)}\|_\text{F}^2 - \|\bm{R}_2^{(j)}-\bm{R}_2^*\bm{O}_{2,r}^{(j)}\|_\text{F}^2 \leq -2\eta Q_{\bm R_2,1} -2\lambda_1\eta G_{\bm R_2}-2\lambda_2\eta T_{\bm R_2}+ \eta^2 Q_{\bm R_2,2},\\
        &\|\bm{P}_1^{(j+1)}-\bm{P}_1^*\bm{O}_{1,p}^{(j)}\|_\text{F}^2 - \|\bm{P}_1^{(j)}-\bm{P}_1^*\bm{O}_{1,p}^{(j)}\|_\text{F}^2 \leq -2\eta Q_{\bm P_1,1} -2\lambda_1\eta G_{\bm P_1}-2\lambda_2\eta T_{\bm P_1}+ \eta^2 Q_{\bm P_1,2},\\
        \text{and  }&\|\bm{P}_2^{(j+1)}-\bm{P}_2^*\bm{O}_{2,p}^{(j)}\|_\text{F}^2 - \|\bm{P}_2^{(j)}-\bm{P}_2^*\bm{O}_{2,p}^{(j)}\|_\text{F}^2 \leq -2\eta Q_{\bm P_2,1} -2\lambda_1\eta G_{\bm P_2}-2\lambda_2\eta T_{\bm P_2}+ \eta^2 Q_{\bm P_2,2}.
    \end{split}
\end{equation}

~\newline
\noindent\textit{Step 2.2} (Upper bound for the errors of $\bm C_i$)

\noindent For $\bm C_1$, note that
\begin{equation}
    \nabla_{\bm{C}_i}\widetilde{\mathcal{L}}^{(j)}  = \nabla_{\bm{A}_i}\widetilde{\mathcal{L}}^{(j)}(\bm{C}_i^{(j)}\bm{D}_{i,11}^{(j)\top}+\bm{P}_i^{(j)}\bm{D}_{i,12}^{(j)\top}) + \nabla_{\bm{A}_i}\widetilde{\mathcal{L}}^{(j)\top}(\bm{C}_i^{(j)}\bm{D}_{i,11}^{(j)}+\bm{R}_i^{(j)}\bm{D}_{i,21}^{(j)}),
\end{equation}
and then,
\begin{equation}\label{eq:C_split}
\begin{aligned}
& \left\|\bm C_1^{(j+1)}-\bm C_1^* \bm O_{1,c}^{(j)}\right\|_{\mathrm{F}}^2 \\
= & \left\|\bm C_1^{(j)}-\bm C_1^* \bm O_{1,c}^{(j)}-\eta\Big\{\nabla_{\bm C_1} \mathcal{L}^{(j)}+\lambda_1\left(\fnorm{\bm A_1^{(j)}}^2-\fnorm{\bm A_2^{(j)}}^2\right)\bm A_i^{(j)}[\bm C_i^{(j)}~\bm P_i^{(j)}][\bm D_{i,11}^{(j)}~\bm D_{i,12}^{(j)}]^\top\right.\\
&+\lambda_1\left(\fnorm{\bm A_1^{(j)}}^2-\fnorm{\bm A_2^{(j)}}^2\right)\bm A_i^{(j)^\top}[\bm C_i^{(j)}~\bm R_i^{(j)}][\bm D_{i,11}^{(j)^\top}~\bm D_{i,21}^{(j)\top}]^\top\\
&\left.+2 \lambda_2 \bm C_1^{(j)}\left(\bm C_1^{(j) \top} \bm C_1^{(j)}-b^2 \bm I_{d_1}\right)+\lambda_2 \bm R_1^{(j)} \bm R_1^{(j) \top} \bm C_1^{(j)}+\lambda_2 \bm P_1^{(j)} \bm P_1^{(j) \top} \bm C_1^{(j)}\Big\} \right\|_\mathrm{F}^2\\
= & \left\|\bm C_1^{(j)}-\bm C_1^* \bm O_{1,c}^{(j)}\right\|_{\mathrm{F}}^2 \\
+& \eta^2\bigg\|\nabla_{\bm C_1} \mathcal{L}^{(j)}+\lambda_1\left(\fnorm{\bm A_1^{(j)}}^2-\fnorm{\bm A_2^{(j)}}^2\right)\bm A_i^{(j)}[\bm C_i^{(j)}~\bm P_i^{(j)}][\bm D_{i,11}^{(j)}~\bm D_{i,12}^{(j)}]^\top\\
&\quad +\lambda_1\left(\fnorm{\bm A_1^{(j)}}^2-\fnorm{\bm A_2^{(j)}}^2\right)\bm A_i^{(j)^\top}[\bm C_i^{(j)}~\bm R_i^{(j)}][\bm D_{i,11}^{(j)^\top}~\bm D_{i,21}^{(j)\top}]^\top\\
&\quad +2 \lambda_2 \bm C_1^{(j)}\left(\bm C_1^{(j) \top} \bm C_1^{(j)}-b^2 \bm I_{d_1}\right)+\lambda_2 \bm R_1^{(j)} \bm R_1^{(j) \top} \bm C_1^{(j)}+\lambda_2 \bm P_1^{(j)} \bm P_1^{(j) \top} \bm C_1^{(j)}\bigg\|_{\mathrm{F}}^2 \\
-& 2 \eta\left\langle\bm C_1^{(j)}-\bm C_1^* \bm O_{1,c}^{(j)}, \nabla_{\bm C_1} \mathcal{L}^{(j)}\right\rangle \\
-&2\lambda_1\eta\inner{\bm C_1^{(j)}-\bm C_1^* \bm O_{1,c}^{(j)}}{\left(\fnorm{\bm A_1^{(j)}}^2-\fnorm{\bm A_2^{(j)}}^2\right)\bm A_i^{(j)}[\bm C_i^{(j)}~\bm P_i^{(j)}][\bm D_{i,11}^{(j)}~\bm D_{i,12}^{(j)}]^\top}\\
-&2\lambda_1\eta\inner{\bm C_1^{(j)}-\bm C_1^* \bm O_{1,c}^{(j)}}{\left(\fnorm{\bm A_1^{(j)}}^2-\fnorm{\bm A_2^{(j)}}^2\right)\bm A_i^{(j)^\top}[\bm C_i^{(j)}~\bm R_i^{(j)}][\bm D_{i,11}^{(j)^\top}~\bm D_{i,21}^{(j)\top}]^\top}\\
-& 2 \lambda_2 \eta\left\langle\bm C_1^{(j)}-\bm C_1^* \bm O_{1,c}^{(j)}, 2 \bm C_1^{(j)}\left(\bm C_1^{(j) \top} \bm C_1^{(j)}-b^2 \bm I_{d_1}\right)\right\rangle \\
-& 2 \lambda_2 \eta\left\langle\bm C_1^{(j)}-\bm C_1^* \bm O_{1,c}^{(j)}, \bm R_1^{(j)} \bm R_1^{(j) \top} \bm C_1^{(j)}+\bm P_1^{(j)} \bm P_1^{(j) \top} \bm C_1^{(j)}\right\rangle\\
:=&\left\|\bm C_1^{(j)}-\bm C_1^* \bm O_{1,c}^{(j)}\right\|_{\mathrm{F}}^2+\eta^2I_{\bm C_1,2}-2\eta I_{\bm C_1,1}.
\end{aligned}
\end{equation}

For $I_{\bm C_1,2}$ in \eqref{eq:C_split},
$$\begin{aligned} I_{\bm C_1,2}\leq &5\left\|\nabla_{\bm C_1} \mathcal{L}^{(j)}\right\|_{\mathrm{F}}^2+40\lambda_1^2b^{-2}\phi^4\left(\fnorm{\bm A_1^{(j)}}^2-\fnorm{\bm A_2^{(j)}}^2\right)^2\\
&+40\lambda_2^2b^2\left\|\bm C_1^{(j) \top} \bm C_1^{(j)}-b^2 \bm I_{d_1}\right\|_{\mathrm{F}}^2+10\lambda_2^2b^2\left\|\bm R_1^{(j) \top} \bm C_1^{(j)}\right\|_{\mathrm{F}}^2+10\lambda_2^2b^2\left\|\bm P_1^{(j) \top} \bm C_1^{(j)}\right\|_{\mathrm{F}}^2.\end{aligned}$$
The first term of the RHS can be bounded as
$$
\begin{aligned}
&\left\|\nabla_{\bm C_1} \mathcal{L}^{(j)}\right\|_{\mathrm{F}}^2 \\
& \leq 2\left\|\nabla_{\bm A_1} \widetilde{\mathcal{L}}(\bm A^{(j)})\left[\bm C_1^{(j)}~ \bm P_1^{(j)}\right]\left[\bm D_{1,11}^{(j)}~ \bm D_{1,12}^{(j)}\right]^{\top}\right\|_{\mathrm{F}}^2\\
&+2\left\|\nabla_{\bm A_1} \widetilde{\mathcal{L}}(\bm A^{(j)})^\top\left[\bm C_1^{(j)}~ \bm R_1^{(j)}\right]\left[\bm D_{1,11}^{(j)\top} ~\bm D_{1,21}^{(j)\top}\right]^{\top}\right\|_{\mathrm{F}}^2 \\
&=2\left\| \text{mat}\left(\mathcal{P}\left(\nabla \widetilde{\mathcal{L}}\left(\bm A^{(j)}\right)\right)^{\top} \V{\bm A_2^{(j)}\right)}\left[\bm C_1^{(j)}~  \bm P_1^{(j)}\right]\left[\bm D_{1,11}^{(j)}~ \bm D_{1,12}^{(j)}\right]^{\top}\right\|_\mathrm{F}^2\\
&+2 \left\|\text{mat}\left(\mathcal{P}\left(\nabla \widetilde{\mathcal{L}}\left(\bm A^{(j)}\right)\right)^{\top} \V{\bm A_2^{(j)}}\right)^{\top}\left[\bm C_1^{(j)}~ \bm R_1^{(j)}\right]\left[\bm D_{1,11}^{(j)\top}~  \bm D_{1,21}^{(j)\top}\right]^{\top} \right\|_\mathrm{F}^2\\
& \leq 4\left\|\text{mat}\left(\mathcal{P}\left(\nabla \widetilde{\mathcal{L}}\left(\bm A^*\right)\right)^{\top} \V{\bm A_2^{(j)}}\right)\left[\bm C_1^{(j)}~ \bm P_1^{(j)}\right]\left[\bm D_{1,11}^{(j)}~ \bm D_{1,12}^{(j)}\right]^{\top}\right\|_{\mathrm{F}}^2\\
&+4\left\|\text{mat}\left(\mathcal{P}\left(\nabla \widetilde{\mathcal{L}}\left(\bm A^*\right)\right)^{\top} \V{\bm A_2^{(j)}\right)}^\top\left[\bm C_1^{(j)}~ \bm R_1^{(j)}\right]\left[\bm D_{1,11}^{(j)\top}~ \bm D_{1,21}^{(j)\top}\right]^{\top}\right\|_{\mathrm{F}}^2 \\
&+4\left\|\text{mat}\left(\mathcal{P}\left(\nabla\widetilde{\mathcal{L}}(\bm A^{(j)})-\nabla\widetilde{\mathcal{L}}(\bm A^*)\right)^\top\V{\bm A_2^{(j)}}\right)\left[\bm C_1^{(j)}~ \bm P_1^{(j)}\right]\left[\bm D_{1,11}^{(j)}~ \bm D_{1,12}^{(j)}\right]^{\top}\right\|_{\mathrm{F}}^2 \\
&+4\left\|\text{mat}\left(\mathcal{P}\left(\nabla\widetilde{\mathcal{L}}(\bm A^{(j)})-\nabla\widetilde{\mathcal{L}}(\bm A^*)\right)^\top\V{\bm A_2^{(j)}}\right)^{\top}\left[\bm C_1^{(j)}~ \bm R_1^{(j)}\right]\left[\bm D_{1,11}^{(j)\top}~ \bm D_{1,21}^{(j)\top}\right]^{\top}\right\|_{\mathrm{F}}^2 \\
& \leq 16b^{-2} \phi^4\xi^2(r_1,r_2,d_1,d_2)+16b^{-2} \phi^4\left\|\nabla\widetilde{\mathcal{L}}(\bm A^{(j)})-\nabla\widetilde{\mathcal{L}}(\bm A^*)\right\|_\text{F}^2\\
& = 16b^{-2} \phi^4\left(\xi^2(r_1,r_2,d_1,d_2)+\left\|\nabla\widetilde{\mathcal{L}}(\bm A^{(j)})-\nabla\widetilde{\mathcal{L}}(\bm A^*)\right\|_\text{F}^2\right).
\end{aligned}
$$
Thus, $I_{\bm C_1,2}$ can be upper bounded by
\begin{equation}
    \label{eq:I_C12_bound}
    \begin{aligned}
        I_{\bm C_1,2} \leq& 80b^{-2} \phi^4\left(\xi^2(r_1,r_2,d_1,d_2)+\left\|\nabla\widetilde{\mathcal{L}}(\bm A^{(j)})-\nabla\widetilde{\mathcal{L}}(\bm A^*)\right\|_\text{F}^2\right)\\
        &+40\lambda_1^2b^{-2}\phi^4\left(\fnorm{\bm A_1^{(j)}}^2-\fnorm{\bm A_2^{(j)}}^2\right)^2\\
        &+40\lambda_1^2b^2\left\|\bm C_1^{(j) \top} \bm C_1^{(j)}-b^2 \bm I_{d_1}\right\|_{\mathrm{F}}^2+10a^2b^2\left\|\bm R_1^{(j) \top} \bm C_1^{(j)}\right\|_{\mathrm{F}}^2+10\lambda_1^2b^2\left\|\bm P_1^{(j) \top} \bm C_1^{(j)}\right\|_{\mathrm{F}}^2\\
        :=&Q_{\bm C_1,2}.
    \end{aligned}    
\end{equation}

For $I_{\bm C_1,1}$ in \eqref{eq:C_split}, similarly to $\bm R_1$ step, we rewrite its first term as the following:
$$
\begin{aligned}
& \left\langle\bm C_1^{(j)}-\bm C_1^* \bm O_{1,c}^{(j)}, \nabla_{\bm C_1} \mathcal{L}^{(j)}\right\rangle \\
= & \left\langle\bm R_1^{(j)} \bm D_{1,21}^{(j)} \bm C_1^{(j) \top}-\bm R_1^{(j)} \bm D_{1,21}^{(j)} \bm O_{1,c}^{(j) \top} \bm C_1^{* \top}, \nabla_{\bm A_1} \widetilde{\mathcal{L}}\left(\mathbf{A}^{(j)}\right)\right\rangle \\
+ & \left\langle\bm C_1^{(j)} \bm D_{1,12}^{(j)} \bm P_1^{(j) \top}-\bm C_1^* \bm O_{1,c}^{(j)} \bm D_{1,12}^{(j)} \bm P_1^{(j) \top}, \nabla_{\bm A_1} \widetilde{\mathcal{L}}\left(\mathbf{A}^{(j)}\right)\right\rangle \\
+ & \left\langle\bm C_1^{(j)} \bm D_{1,11}^{(j)} \bm C_1^{(j) \top}-\bm C_1^* \bm O_{1,c}^{(j)} \bm D_{1,11}^{(j)} \bm C_1^{(j) \top}, \nabla_{\bm A_1} \widetilde{\mathcal{L}}\left(\mathbf{A}^{(j)}\right)\right\rangle \\
+ & \left\langle\bm C_1^{(j)} \bm D_{1,11}^{(j)} \bm C_1^{(j) \top}-\bm C_1^{(j)} \bm D_{1,11}^{(j)} \bm O_{1,c}^{(j)} \bm C_1^{* \top}, \nabla_{\bm A_1} \widetilde{\mathcal{L}}\left(\mathbf{A}^{(j)}\right)\right\rangle \\
:= & \left\langle\mathbf{A}_{\bm C_1}^{(j)}, \nabla_{\bm A_1} \widetilde{\mathcal{L}}\left(\mathbf{A}^{(j)}\right)\right\rangle\\
=&\left\langle\bm A_2^{(j)}\otimes \bm A_{\bm C_1}^{(j)},\nabla\widetilde{\mathcal{L}}(\bm A^{(j)})\right\rangle\\
:=&Q_{\bm C_1,1}.
\end{aligned}
$$
For the last three terms, denote
$$
\begin{aligned}G_{\bm C_1}:=&\inner{\bm C_1^{(j)}-\bm C_1^* \bm O_{1,c}^{(j)}}{\left(\fnorm{\bm A_1^{(j)}}^2-\fnorm{\bm A_2^{(j)}}^2\right)\bm A_i^{(j)}[\bm C_i^{(j)}~\bm P_i^{(j)}][\bm D_{i,11}^{(j)}~\bm D_{i,12}^{(j)}]^\top}\\
+&\inner{\bm C_1^{(j)}-\bm C_1^* \bm O_{1,c}^{(j)}}{\left(\fnorm{\bm A_1^{(j)}}^2-\fnorm{\bm A_2^{(j)}}^2\right)\bm A_i^{(j)\top}[\bm C_i^{(j)}~\bm R_i^{(j)}][\bm D_{i,11}^{(j)\top}~\bm D_{i,21}^{(j)\top}]^\top}
\end{aligned}$$
and
$$
T_{\bm C_1}:=\left\langle\bm C_1^{(j)}-\bm C_1^* \bm O_{1,c}^{(j)}, 2 \bm C_1^{(j)}\left(\bm C_1^{(j) \top} \bm C_1^{(j)}-b^2 \bm I_{d_1}\right)+\bm R_1^{(j)} \bm R_1^{(j) \top} \bm C_1^{(j)}+\bm P_1^{(j)} \bm P_1^{(j) \top} \bm C_1^{(j)}\right\rangle .
$$
Therefore, we can rewrite $I_{\bm C_1,1}$ as
$$I_{\bm C_1,1}=Q_{\bm C_1,1}+\lambda_1 G_{\bm C_1}+\lambda_2 T_{\bm C_1}.$$

Combining these bounds for $I_{\bm C_1,2}$ in \eqref{eq:I_C12_bound}, we have
$$
\left\|\bm C_1^{(j+1)}-\bm C_1^* \bm O_{1,c}^{(j)}\right\|_{\mathrm{F}}^2-\left\|\bm C_1^{(j)}-\bm C_1^* \bm O_{1,c}^{(j)}\right\|_{\mathrm{F}}^2 \leq -2\eta Q_{\bm C_1,1}-2\lambda_1\eta G_{\bm C_1,1}-2\lambda_2\eta T_{\bm C_1,1}+\eta^2 Q_{\bm C_1, 2}.
$$
Similarly, we also have
$$
\left\|\bm C_2^{(j+1)}-\bm C_2^* \bm O_{2,c}^{(j)}\right\|_{\mathrm{F}}^2-\left\|\bm C_2^{(j)}-\bm C_2^* \bm O_{2,c}^{(j)}\right\|_{\mathrm{F}}^2 \leq -2\eta Q_{\bm C_2,1}-2\lambda_1\eta G_{\bm C_2,1}-2\lambda_2\eta T_{\bm C_2,1}+\eta^2 Q_{\bm C_2, 2}.
$$

~\newline
\noindent \textit{Step 2.3} (Upper bound for the errors of $\bm D_i$)

\begin{equation}\label{eq:D_split}
\begin{aligned}
    & \left\|\bm D_1^{(j+1)}-\bm O_{1,u}^{(j) \top} \bm D_1^* \bm O_{1,v}^{(j)}\right\|_{\mathrm{F}}^2 \\
    = & \left\|\bm D_1^{(j)}-\bm O_{1,u}^{(j) \top} \bm D_1^* \bm O_{1,v}^{(j)}-\eta\left\{\left[\bm C_1^{(j)}~ \bm R_1^{(j)}\right]^{\top} \nabla_{\bm A_1}\widetilde{\mathcal{L}}(\bm A^{(j)})\left[\bm C_1^{(j)}~ \bm P_1^{(j)}\right]\right.\right.\\
    &\left.\left.+\lambda_1\left(\fnorm{\bm A_1^{(j)}}^2-\fnorm{\bm A_2^{(j)}}^2\right)[\bm C_i^{(j)}~\bm R_i^{(j)}]^\top\bm A_i^{(j)}[\bm C_i^{(j)}~\bm P_i^{(j)}]\right\}\right\|_{\mathrm{F}}^2 \\
    = & \left\|\bm D_1^{(j)}-\bm O_{1,u}^{(j) \top} \bm D_1^* \bm O_{1,v}^{(j)}\right\|_{\mathrm{F}}^2\\
    +&\eta^2\left\|\left[\bm C_1^{(j)}~ \bm R_1^{(j)}\right]^{\top} \nabla_{\bm A_1}\widetilde{\mathcal{L}}(\bm A^{(j)})\left[\bm C_1^{(j)}~ \bm P_1^{(j)}\right]\right.
    \\&\qquad +\left.\lambda_1\left(\fnorm{\bm A_1^{(j)}}^2-\fnorm{\bm A_2^{(j)}}^2\right)[\bm C_i^{(j)}~\bm R_i^{(j)}]^\top\bm A_i^{(j)}[\bm C_i^{(j)}~\bm P_i^{(j)}]\right\|_{\mathrm{F}}^2 \\
    - & 2 \eta\left\langle\bm D_1^{(j)}-\bm O_{1,u}^{(j) \top} \bm D_1^* \bm O_{1,v}^{(j)},\left[\bm C_1^{(j)}~ \bm R_1^{(j)}\right]^{\top} \nabla_{\bm A_1}\widetilde{\mathcal{L}}(\bm A^{(j)})\left[\bm C_1^{(j)}~ \bm P_1^{(j)}\right]\right\rangle\\
    -&2\lambda_1\eta\inner{\bm D_1^{(j)}-\bm O_{1,u}^{(j) \top} \bm D_1^* \bm O_{1,v}^{(j)}}{\left(\fnorm{\bm A_1^{(j)}}^2-\fnorm{\bm A_2^{(j)}}^2\right)\left[\bm C_i^{(j)}~\bm R_i^{(j)}\right]^\top\bm A_i^{(j)}\left[\bm C_i^{(j)}~\bm P_i^{(j)}\right]}\\
    :=&\left\|\bm D_1^{(j)}-\bm O_{1,u}^{(j) \top} \bm D_1^* \bm O_{1,v}^{(j)}\right\|_{\mathrm{F}}^2+\eta^2 I_{\bm D_1,2}-2\eta I_{\bm D_1,1}.
\end{aligned}
\end{equation}

For $I_{\bm D_1,2}$:
$$\begin{aligned}
    I_{\bm D_1,2}\leq& 2\fnorm{\left[\bm C_1^{(j)}~ \bm R_1^{(j)}\right]^{\top} \nabla_{\bm A_1}\widetilde{\mathcal{L}}(\bm A^{(j)})\left[\bm C_1^{(j)}~ \bm P_1^{(j)}\right]}^2+4\lambda_1^2b^4\phi^2\left(\fnorm{\bm A_1^{(j)}}^2-\fnorm{\bm A_2^{(j)}}^2\right)^2.
\end{aligned}$$
For the first term,
$$\begin{aligned}
    & \left\|\left[\bm C_1^{(j)}~ \bm R_1^{(j)}\right]^{\top} \nabla_{\bm A_1}\widetilde{\mathcal{L}}(\bm A^{(j)})\left[\bm C_1^{(j)}~ \bm P_1^{(j)}\right]\right\|_{\mathrm{F}}^2 \\
    =&\left\|\left[\bm C_1^{(j)}~ \bm R_1^{(j)}\right]^{\top} \text{mat}\left(\mathcal{P}\left(\nabla \widetilde{\mathcal{L}}\left(\bm A^{(j)}\right)\right)^{\top} \V{\bm A_2^{(j)}\right)}\left[\bm C_1^{(j)}~ \bm P_1^{(j)}\right]\right\|_\text{F}^2\\
    \leq & 2\left\|\left[\bm C_1^{(j)}~ \bm R_1^{(j)}\right]^{\top} \text{mat}\left(\mathcal{P}\left(\nabla \widetilde{\mathcal{L}}\left(\bm A^*\right)\right)^{\top} \V{\bm A_2^{(j)}\right)}\left[\bm C_1^{(j)}~ \bm P_1^{(j)}\right]\right\|_{\mathrm{F}}^2 \\
    + & 2\left\|\left[\bm C_1^{(j)}~ \bm R_1^{(j)}\right]^{\top}\text{mat}\left(\mathcal{P}\left[\nabla \widetilde{\mathcal{L}}\left(\bm A^{(j)}\right)-\nabla\widetilde{\mathcal{L}}(\bm A^*)\right]^{\top} \V{\bm A_2^{(j)}}\right)\left[\bm C_1^{(j)}~ \bm P_1^{(j)}\right]\right\|_{\mathrm{F}}^2 \\
    \leq & 4b^4\phi^2\left(\xi^2(r_1,r_2,d_1,d_2)+\left\|\nabla\widetilde{\mathcal{L}}(\bm A^{(j)})-\nabla\widetilde{\mathcal{L}}(\bm A^*)\right\|_\text{F}^2\right)
\end{aligned}$$
Then $I_{\bm D_1,2}$ can be upper bounded as 
$$\begin{aligned}
&I_{\bm D_1,2}\\
\leq& 4b^4\phi^2\left(\xi^2(r_1,r_2,d_1,d_2)+\left\|\nabla\widetilde{\mathcal{L}}(\bm A^{(j)})-\nabla\widetilde{\mathcal{L}}(\bm A^*)\right\|_\text{F}^2\right)+4\lambda_1^2b^4\phi^2\left(\fnorm{\bm A_1^{(j)}}^2-\fnorm{\bm A_2^{(j)}}^2\right)^2\\
:=&Q_{\bm D_1,2}
\end{aligned}$$

For $I_{\bm D_1,1}$, similarly we define
$$
\begin{aligned}
& \left\langle\bm D_1^{(j)}-\bm O_{1,u}^{(j) \top} \bm D_1^* \bm O_{1,v}^{(j)}, \left[\bm C_1^{(j)}~ \bm R_1^{(j)}\right]^{\top} \nabla_{\bm A_1}\widetilde{\mathcal{L}}(\bm A^{(j)})\left[\bm C_1^{(j)}~ \bm P_1^{(j)}\right]\right\rangle \\
= & \left\langle\bm A_1^{(j)}-\left[\bm C_1^{(j)}~ \bm R_1^{(j)}\right] \bm O_{1,u}^{(j) \top} \bm D_1^* \bm O_{1,v}^{(j)}\left[\bm C_1^{(j)}~ \bm P_1^{(j)}\right]^{\top}, \nabla_{\bm A_1}\widetilde{\mathcal{L}}(\bm A^{(j)})\right\rangle \\
:= & \left\langle\bm A_{\bm D_1}^{(j)}, \nabla_{\bm A_1}\widetilde{\mathcal{L}}(\bm A^{(j)})\right\rangle\\
=&\left\langle\bm A_2^{(j)}\otimes\bm A_{\bm D_1}^{(j)},\nabla\widetilde{\mathcal{L}}(\bm A^{(j)})\right\rangle\\
:=&Q_{\bm D_1, 1},
\end{aligned}
$$
and 
$$G_{\bm D_1}:=\inner{\bm D_1^{(j)}-\bm O_{1,u}^{(j) \top} \bm D_1^* \bm O_{1,v}^{(j)}}{\left(\fnorm{\bm A_1^{(j)}}^2-\fnorm{\bm A_2^{(j)}}^2\right)\left[\bm C_1^{(j)}~\bm R_1^{(j)}\right]^\top\bm A_1^{(j)}\left[\bm C_1^{(j)}~\bm P_1^{(j)}\right]}$$

Hence, we have
$$
\left\|\bm D_1^{(j+1)}-\bm O_{1,u}^{(j) \top} \bm D_1^* \bm O_{1,v}^{(j)}\right\|_{\mathrm{F}}^2-\left\|\bm D_1^{(j)}-\bm O_{1,u}^{(j) \top} \bm D_1^* \bm O_{1,v}^{(j)}\right\|_{\mathrm{F}}^2 \leq-2 \eta Q_{\bm D_1, 1}-2\lambda_1\eta G_{\bm D_1}+\eta^2 Q_{\bm D_1, 2},
$$
and
$$\left\|\bm D_2^{(j+1)}-\bm O_{2,u}^{(j) \top} \bm D_2^* \bm O_{2,v}^{(j)}\right\|_{\mathrm{F}}^2-\left\|\bm D_2^{(j)}-\bm O_{2,u}^{(j) \top} \bm D_2^* \bm O_{2,v}^{(j)}\right\|_{\mathrm{F}}^2 \leq-2 \eta Q_{\bm D_2, 1}-2\lambda_1\eta G_{\bm D_2}+\eta^2 Q_{\bm D_2, 2}.$$

Combining the pieces above,
\begin{equation}\label{E upper bound: Q2Q1GT}
\begin{aligned}
    \mathrm{dist}^2_{(j+1)}&\leq \mathrm{dist}^2_{(j)}+\eta^2\sum_{i=1}^{2}\left(Q_{\mathrm{D}_i, 2}+Q_{\mathrm{R}_i, 2}+Q_{\mathrm{P}_i, 2}+Q_{\mathrm{C}_i, 2}\right)\\
    &-2\eta\sum_{i=1}^{2}\left(Q_{\mathrm{D}_i, 1}+Q_{\mathrm{R}_i, 1}+Q_{\mathrm{P}_i, 1}+Q_{\mathrm{C}_i, 1}\right)\\
    &-2\lambda_1\eta\sum_{i=1}^{2}\left(G_{D_i}+G_{R_i}+G_{P_i}+G_{C_i}\right)\\
    &-2\lambda_2\eta\sum_{i=1}^{2}\left(T_{R_i}+T_{P_i}+T_{C_i}\right).
\end{aligned}
\end{equation}\\

\noindent\textbf{\textit{Step 3.}} (\textbf{Recursive relationship between $\mathrm{dist}^2_{(j+1)}$ and $\mathrm{dist}^2_{(j)}$})

\noindent In this step, we will derive upper bounds of $Q_{\cdot,1},Q_{\cdot,2},G$ and $T$ terms, and finally obtain a upper bound as in \eqref{E upper bound 2}.\\

\noindent\textit{Step 3.1} (Lower bound for $\sum_{i=1}^{2}\left(Q_{\bm D_i, 1}+Q_{\bm R_i, 1}+Q_{\bm P_i, 1}+Q_{\bm C_i, 1}\right)$ )

\noindent By definition, 
\begin{equation}\label{Q_1}
\begin{aligned}
& \sum_{i=1}^{2}\left(Q_{\bm D_i, 1}+Q_{\bm R_i, 1}+Q_{\bm P_i, 1}+Q_{\bm C_i, 1}\right)\\
=&\left\langle
    \begin{matrix}
        \bm A_2^{(j)}\otimes\left(\mathbf{A}_{\bm D_1}^{(j)}+\mathbf{A}_{\bm R_1}^{(j)}+\mathbf{A}_{\bm P_1}^{(j)}+\mathbf{A}_{\bm C_1}^{(j)}\right)\\
        +\left(\mathbf{A}_{\bm D_2}^{(j)}+\mathbf{A}_{\bm R_2}^{(j)}+\mathbf{A}_{\bm P_2}^{(j)}+\mathbf{A}_{\bm C_2}^{(j)}\right)\otimes \bm A_1^{(j)}
    \end{matrix},\nabla \widetilde{\mathcal{L}}\left(\mathbf{A}^{(j)}\right)\right\rangle.
\end{aligned}
\end{equation}

Noting that
$$
\begin{aligned}
& \mathbf{A}_{\bm D_1}^{(j)}+\mathbf{A}_{\bm R_1}^{(j)}+\mathbf{A}_{\bm P_1}^{(j)}+\mathbf{A}_{\bm C_1}^{(j)} \\
= & 3 \bm A_1^{(j)}-\left[\bm C_1^{(j)}~ \bm R_1^{(j)}\right] \bm O_{1,u}^{(j) \top} \bm D_1^* \bm O_{1,v}^{(j)}\left[\bm C_1^{(j)}~ \bm P_1^{(j)}\right]^{\top}-\left[\bm C_1^{(j)}~ \bm R_1^{(j)}\right] \bm D_1^{(j)} \bm O_{1,v}^{(j) \top}\left[\bm C_1^*~ \bm P_1^*\right]^{\top} \\
- & {\left[\bm C_1^*~ \bm R_1^*\right] \bm O_{1,u}^{(j)} \bm D_1^{(j)}\left[\bm C_1^{(j)}~ \bm P_1^{(j)}\right]^{\top} },
\end{aligned}
$$
we define 
$$
\bm H_1^{(j)}:=\bm A_2^{(j)}\otimes\left(\mathbf{A}_{\bm D_1}^{(j)}+\mathbf{A}_{\bm R_1}^{(j)}+\mathbf{A}_{\bm P_1}^{(j)}+\mathbf{A}_{\bm C_1}^{(j)}\right)
+\left(\mathbf{A}_{\bm D_2}^{(j)}+\mathbf{A}_{\bm R_2}^{(j)}+\mathbf{A}_{\bm P_2}^{(j)}+\mathbf{A}_{\bm C_2}^{(j)}\right)\otimes \bm A_1^{(j)},
$$
which contains all the first order perturbation terms. By Lemma \ref{lemma:H_upperBound}, we know that
$$\bm H_1^{(j)}=\bm A_2^{(j)}\otimes\bm A_1^{(j)}-\bm A_2^*\otimes\bm A_1^*+\bm H^{(j)}.$$
Hence, \eqref{Q_1} can be simplified as
$\left\langle\bm A_2^{(j)}\otimes\bm A_1^{(j)}-\bm A_2^*\otimes\bm A_1^*+\bm H^{(j)},\nabla \widetilde{\mathcal{L}}\left(\mathbf{A}^{(j)}\right)\right\rangle.$
With conditions \eqref{condition: upper bound of pieces} \eqref{condition:E_upperBound} and Lemma \ref{lemma:H_upperBound}, we have
$$\fnorm{\bm H^{(j)}}\leq C_h\phi^{4/3}\mathrm{dist}^2_{(j)},$$
where $C_h$ is a constant of moderate size.

Then, for \eqref{Q_1}, with RCG condition \eqref{RCG}, we have
\begin{equation}\label{Q1_inner_lowerBound}
\begin{aligned}
    &\left\langle
        \begin{matrix}\bm A_2^{(j)}\otimes\left(\mathbf{A}_{\bm D_1}^{(j)}+\mathbf{A}_{\bm R_1}^{(j)}+\mathbf{A}_{\bm P_1}^{(j)}+\mathbf{A}_{\bm C_1}^{(j)}\right)\\
        +\left(\mathbf{A}_{\bm D_2}^{(j)}+\mathbf{A}_{\bm R_2}^{(j)}+\mathbf{A}_{\bm P_2}^{(j)}+\mathbf{A}_{\bm C_2}^{(j)}\right)\otimes \bm A_1^{(j)}
        \end{matrix}
        , \nabla \widetilde{\mathcal{L}}\left(\mathbf{A}^{(j)}\right)\right\rangle\\
    =&\left\langle\bm A_2^{(j)}\otimes\bm A_1^{(j)}-\bm A_2^*\otimes\bm A_1^*+\bm H^{(j)},\nabla \widetilde{\mathcal{L}}\left(\mathbf{A}^{(j)}\right)\right\rangle\\
    =&\left\langle\bm A_2^{(j)}\otimes\bm A_1^{(j)}-\bm A_2^*\otimes\bm A_1^*,\nabla\widetilde{\mathcal{L}}(\bm A^{(j)})-\nabla\widetilde{\mathcal{L}}(\bm A^*)\right\rangle+\left\langle\bm H^{(j)},\nabla\widetilde{\mathcal{L}}(\bm A^{(j)})-\nabla\widetilde{\mathcal{L}}(\bm A^*)\right\rangle\\
    &+\left\langle\bm A_2^{(j)}\otimes\bm A_1^{(j)}-\bm A_2^*\otimes\bm A_1^*+\bm H^{(j)},\nabla \widetilde{\mathcal{L}}\left(\mathbf{A}^*\right)\right\rangle\\
    \geq&\frac{\alpha}{2}\left\|\bm A_2^{(j)}\otimes\bm A_1^{(j)}-\bm A_2^*\otimes\bm A_1^*\right\|_F^2+\frac{1}{2\beta}\left\|\nabla\widetilde{\mathcal{L}}(\bm A^{(j)})-\nabla\widetilde{\mathcal{L}}(\bm A^*)\right\|_F^2\\
    &-\fnorm{\bm H^{(j)}}\fnorm{\nabla\widetilde{\mathcal{L}}(\bm A^{(j)})-\nabla\widetilde{\mathcal{L}}(\bm A^*)}\\
    &-\left|\left\langle\bm A_2^{(j)}\otimes\left(\mathbf{A}_{\bm D_1}^{(j)}+\mathbf{A}_{\bm R_1}^{(j)}+\mathbf{A}_{\bm P_1}^{(j)}+\mathbf{A}_{\bm C_1}^{(j)}\right),\nabla \widetilde{\mathcal{L}}\left(\mathbf{A}^*\right)\right\rangle\right|\\
    &-\left|\left\langle\left(\mathbf{A}_{\bm D_2}^{(j)}+\mathbf{A}_{\bm R_2}^{(j)}+\mathbf{A}_{\bm P_2}^{(j)}+\mathbf{A}_{\bm C_2}^{(j)}\right)\otimes \bm A_1^{(j)},\nabla \widetilde{\mathcal{L}}\left(\mathbf{A}^*\right)\right\rangle\right|.
\end{aligned}
\end{equation}

For the fourth term in \eqref{Q1_inner_lowerBound}, plugging in $b=\phi^{1/3}$:
\begin{equation}
\begin{aligned}
        &\left|\left\langle\bm A_2^{(j)}\otimes\left(\mathbf{A}_{\bm D_1}^{(j)}+\mathbf{A}_{\bm R_1}^{(j)}+\mathbf{A}_{\bm P_1}^{(j)}+\mathbf{A}_{\bm C_1}^{(j)}\right),\nabla \widetilde{\mathcal{L}}\left(\mathbf{A}^*\right)\right\rangle\right|\\
    \leq&\left|\left\langle\bm A_2^{(j)}\otimes\left(\mathbf{A}_1^{(j)}-\left[\bm C_1^*~ \bm R_1^*\right] \bm O_{1,u}^{(j)} \bm D_1^{(j)}\left[\bm C_1^{(j)}~ \bm P_1^{(j)}\right]^{\top}\right), \nabla \widetilde{\mathcal{L}}\left(\mathbf{A}^*\right)\right\rangle\right| \\
        &+\left|\left\langle\bm A_2^{(j)}\otimes\left(\mathbf{A}_1^{(j)}-\left[\bm C_1^{(j)}~ \bm R_1^{(j)}\right] \bm D_1^{(j)} \bm O_{1,v}^{(j) \top}\left[\bm C_1^*~ \bm P_1^*\right]^{\top}\right), \nabla \widetilde{\mathcal{L}}\left(\mathbf{A}^*\right)\right\rangle\right| \\
    &+\left|\left\langle\bm A_2^{(j)}\otimes\left(\mathbf{A}_1^{(j)}-\left[\bm C_1^{(j)}~ \bm R_1^{(j)}\right] \bm O_{1,u}^{(j) \top} \bm D_1^* \bm O_{1,v}^{(j)}\left[\bm C_1^{(j)}~ \bm P_1^{(j)}\right]^{\top}\right), \nabla \widetilde{\mathcal{L}}\left(\mathbf{A}^*\right)\right\rangle\right| \\
        \leq&\xi(r_1,r_2,d_1,d_2)\left(\left\|\bm A_2^{(j)}\right\|_\mathrm{F}\left\|\bm D_1^{(j)}\right\|_{\text {op }} \cdot\left\|\left[\bm C_1^{(j)}~ \bm P_1^{(j)}\right]\right\|_{\mathrm{op}} \cdot\left\|\left[\bm C_1^{(j)}~ \bm R_1^{(j)}\right]-\left[\bm C_1^*~ \bm R_1^*\right] \bm O_{1,u}^{(j)}\right\|_{\mathrm{F}}\right)\\
    &+ \xi(r_1,r_2,d_1,d_2)\left(\left\|\bm A_2^{(j)}\right\|_\mathrm{F}\left\|\bm D_1^{(j)}\right\|_{\text {op }} \cdot\left\|\left[\bm C_1^{(j)}~ \bm R_1^{(j)}\right]\right\|_{\mathrm{op}} \cdot\left\|\left[\bm C_1^{(j)}~ \bm P_1^{(j)}\right]-\left[\bm C_1^*~ \bm P_1^*\right] \bm O_{1,v}^{(j)}\right\|_{\mathrm{F}}\right) \\
        &+ \xi(r_1,r_2,d_1,d_2)\left(\left\|\bm A_2^{(j)}\right\|_\mathrm{F}\left\|\left[\bm C_1^{(j)}~ \bm R_1^{(j)}\right]\right\|_{\mathrm{op}} \cdot\left\|\left[\bm C_1^{(j)}~ \bm P_1^{(j)}\right]\right\|_{\mathrm{op}} \cdot\left\|\bm D_1^{(j)}-\bm O_{1,u}^{(j) \top} \bm D_1^* \bm O_{1,v}^{(j)}\right\|_{\mathrm{F}}\right) \\
    \leq& \left(2b^2+4 \phi / b\right)\phi \xi(r_1,r_2,d_1,d_2)\mathrm{dist}_{(j)} \\
    =&6\phi^{5/3}\xi(r_1,r_2,d_1,d_2)\mathrm{dist}_{(j)}.
\end{aligned}
\end{equation}
Applying the same analysis on the last term in \eqref{Q1_inner_lowerBound}, the last two terms in \eqref{Q1_inner_lowerBound} can be upper bounded by
\begin{align}
    &\left|\left\langle\bm A_2^{(j)}\otimes\left(\mathbf{A}_{\bm D_1}^{(j)}+\mathbf{A}_{\bm R_1}^{(j)}+\mathbf{A}_{\bm P_1}^{(j)}+\mathbf{A}_{\bm C_1}^{(j)}\right),\nabla \widetilde{\mathcal{L}}\left(\mathbf{A}^*\right)\right\rangle\right|\\
    &+\left|\left\langle\left(\mathbf{A}_{\bm D_2}^{(j)}+\mathbf{A}_{\bm R_2}^{(j)}+\mathbf{A}_{\bm P_2}^{(j)}+\mathbf{A}_{\bm C_2}^{(j)}\right)\otimes \bm A_1^{(j)},\nabla \widetilde{\mathcal{L}}\left(\mathbf{A}^*\right)\right\rangle\right|\\
    \leq&6\phi^{5/3}\xi(r_1,r_2,d_1,d_2)\mathrm{dist}_{(j)}+6\phi^{5/3}\xi(r_1,r_2,d_1,d_2)\mathrm{dist}_{(j)}\\
    =&12\phi^{5/3}\xi(r_1,r_2,d_1,d_2)\mathrm{dist}_{(j)}\\
    \leq&36c\phi^{10/3}\mathrm{dist}^2_{(j)}+\frac{1}{c}\xi(r_1,r_2,d_1,d_2)^2, \quad\forall c >0.
\end{align}
The last inequality stems from the fact that $x^2+y^2\geq 2xy$.

For the third term in \eqref{Q1_inner_lowerBound}, we know from Lemma \ref{lemma:H_upperBound} that $\fnorm{\bm H^{(j)}}\leq C_h\phi^{4/3}\mathrm{dist}^2_{(j)}$, and with condition \eqref{condition:E_upperBound}, $\mathrm{dist}^2_{(j)}\leq C_D\phi^{2/3}\alpha\beta^{-1}\kappa^{-2}$, and 
$$
\begin{aligned}
    &\fnorm{\bm H^{(j)}}\fnorm{\nabla\widetilde{\mathcal{L}}(\bm A^{(j)})-\nabla\widetilde{\mathcal{L}}(\bm A^*)}\\
    \leq &\frac{1}{4\beta}\fnorm{\nabla\widetilde{\mathcal{L}}(\bm A^{(j)})-\nabla\widetilde{\mathcal{L}}(\bm A^*)}^2+\beta\fnorm{\bm H^{(j)}}^2\\
    \leq &\frac{1}{4\beta}\fnorm{\nabla\widetilde{\mathcal{L}}(\bm A^{(j)})-\nabla\widetilde{\mathcal{L}}(\bm A^*)}^2+\beta(C_h^2\phi^{8/3}\mathrm{dist}^2_{(j)})\mathrm{dist}^2_{(j)}\\
    \leq &\frac{1}{4\beta}\fnorm{\nabla\widetilde{\mathcal{L}}(\bm A^{(j)})-\nabla\widetilde{\mathcal{L}}(\bm A^*)}^2+\frac{C_H\alpha\phi^{10/3}}{\kappa^2} \mathrm{dist}^2_{(j)},
\end{aligned}
$$
where $C_H=C_h^2C_D$ is a small positive constant.
Consequently, putting together the bounds, we have:
\begin{equation}
\begin{aligned}
    &\left\langle
        \begin{matrix}\bm A_2^{(j)}\otimes\left(\mathbf{A}_{\bm D_1}^{(j)}+\mathbf{A}_{\bm R_1}^{(j)}+\mathbf{A}_{\bm P_1}^{(j)}+\mathbf{A}_{\bm C_1}^{(j)}\right)\\
        +\left(\mathbf{A}_{\bm D_2}^{(j)}+\mathbf{A}_{\bm R_2}^{(j)}+\mathbf{A}_{\bm P_2}^{(j)}+\mathbf{A}_{\bm C_2}^{(j)}\right)\otimes \bm A_1^{(j)}
        \end{matrix}
        , \nabla \widetilde{\mathcal{L}}\left(\mathbf{A}^{(j)}\right)\right\rangle\\
    \geq&\frac{\alpha}{2}\fnorm{\bm A_2^{(j)}\otimes\bm A_1^{(j)}-\bm A_2^*\otimes\bm A_1^*}^2+\frac{1}{4\beta}\fnorm{\nabla\widetilde{\mathcal{L}}(\bm A^{(j)})-\nabla\widetilde{\mathcal{L}}(\bm A^*)}^2\\
    &-\frac{C_H\alpha\phi^{10/3}}{\kappa^2} \mathrm{dist}^2_{(j)}-36c\phi^{10/3}\mathrm{dist}^2_{(j)}-\frac{1}{c}\xi(r_1,r_2,d_1,d_2)^2.
\end{aligned}
\end{equation}

~\\
\noindent\textit{Step 3.2.} (Lower bound for $\sum_{i=1}^{2}\left(G_{\bm D_i}+G_{\bm R_i}+G_{\bm P_i}+G_{\bm C_i}\right)$)

\noindent Recall the definitions in \eqref{eq:G_R1} and \eqref{eq:R_P_bound}. It can be easily verified that, by adding the terms together, we have 
\begin{equation}\label{eq: G1}
    \begin{aligned}
    &G_{\bm D_1}+G_{\bm R_1}+G_{\bm P_1}+G_{\bm C_1}\\
    =&\inner{\begin{matrix}
        \bm A_1^{(j)}-\left[\bm C_1^*~\bm R_1^*\right]\bm O_{1,u}\bm D_1^{(j)}\sqmat{\bm C_1^{(j)}~\bm P_1^{(j)}}^\top\\
        +\bm A_1^{(j)}-\left[\bm C_1^{(j)}~\bm R_1^{(j)}\right]\bm O_{1,u}^\top\bm D_1^*\bm O_{1,v}\sqmat{\bm C_1^{(j)}~\bm P_1^{(j)}}^\top\\
        +\bm A_1^{(j)}-\left[\bm C_1^{(j)}~\bm R_1^{(j)}\right]\bm D_1^{(j)}\bm O_{1,v}^\top\sqmat{\bm C_1^*~\bm P_1^*}^\top
    \end{matrix}}
    {\left(\fnorm{\bm A_1^{(j)}}^2-\fnorm{\bm A_2^{(j)}}^2\right)\bm A_1^{(j)}}\\
    :=&\inner{\bm H_{\bm A_1,1}}{\left(\fnorm{\bm A_1^{(j)}}^2-\fnorm{\bm A_2^{(j)}}^2\right)\bm A_1^{(j)}}.
\end{aligned}\end{equation}
The left side of the inner product, $\bm H_{\bm A_1,1}$, contains the first order perturbation terms with respect to $\bm A_1$. By Lemma A.1 in \cite{wang2023commonFactor}, it is exactly $\bm A_1^{(j)}-\bm A_1^*+\bm H_{\bm A_1,2}$, where $\bm H_{\bm A_1,2}$ comprises all the second and third-order perturbation terms. Applying it with $b=\phi^{1/3}$, we have $\fnorm{\bm H_{\bm A_1,2}}\leq 5\phi^{1/3}\mathrm{dist}^2_{(j)}$. Similarly for $\bm A_2$, we have
\begin{equation}\label{eq: G2}
    \begin{aligned}
    &G_{\bm D_2}+G_{\bm R_2}+G_{\bm P_2}+G_{\bm C_2}\\
    =&\inner{\begin{matrix}
        \bm A_2^{(j)}-\left[\bm C_2^*~\bm R_2^*\right]\bm O_{2,u}\bm D_2^{(j)}\sqmat{\bm C_2^{(j)}~\bm P_2^{(j)}}^\top\\
        +\bm A_2^{(j)}-\left[\bm C_2^{(j)}~\bm R_2^{(j)}\right]\bm O_{2,u}^\top\bm D_2^*\bm O_{2,v}\sqmat{\bm C_2^{(j)}~\bm P_2^{(j)}}^\top\\
        +\bm A_2^{(j)}-\left[\bm C_2^{(j)}~\bm R_2^{(j)}\right]\bm D_2^{(j)}\bm O_{2,v}^\top\sqmat{\bm C_2^*~\bm P_2^*}^\top
    \end{matrix}}
    {-\left(\fnorm{\bm A_1^{(j)}}^2-\fnorm{\bm A_2^{(j)}}^2\right)\bm A_2^{(j)}}\\
    :=&\inner{\bm H_{\bm A_2,1}}{-\left(\fnorm{\bm A_1^{(j)}}^2-\fnorm{\bm A_2^{(j)}}^2\right)\bm A_2^{(j)}}\\
    =&\inner{\bm A_2^{(j)}-\bm A_2^*+\bm H_{\bm A_2,2}}{-\left(\fnorm{\bm A_1^{(j)}}^2-\fnorm{\bm A_2^{(j)}}^2\right)\bm A_2^{(j)}},
\end{aligned}\end{equation}
and $\fnorm{\bm H_{\bm A_2,2}}\leq 5\phi^{1/3}\mathrm{dist}^2_{(j)}$.

Now we define 
$$\bm Z=\begin{pmatrix}
    \V{\bm A_1^{(j)}}\\
    \V{\bm A_2^{(j)}}
\end{pmatrix},
\widetilde{\bm Z}=\begin{pmatrix}
    \V{\bm A_1^{(j)}}\\
    -\V{\bm A_2^{(j)}}
\end{pmatrix},
\bm Z^*=\begin{pmatrix}
    \V{\bm A_1^*}\\
    \V{\bm A_2^*}
\end{pmatrix},
\widetilde{\bm Z}^*=\begin{pmatrix}
    \V{\bm A_1^*}\\
    -\V{\bm A_2^*}
\end{pmatrix}.$$
Then, one can show that
$$\fnorm{\bm A_1^{(j)}}^2-\fnorm{\bm A_2^{(j)}}^2=\bm Z^\top\widetilde{\bm Z}=\widetilde{\bm Z}^\top\bm Z.$$
Vectorizing the matrices and putting \eqref{eq: G1} and \eqref{eq: G2} together, we have
$$\begin{aligned}
    &\sum_{i=1}^{2}\left(G_{\bm D_i}+G_{\bm R_i}+G_{\bm P_i}+G_{\bm C_i}\right)\\
    =&\inner{\bm Z-\bm Z^*}{\widetilde{\bm Z}\widetilde{\bm Z}^\top\bm Z}\\
    +&\inner{\bm H_{\bm A_1,2}}{\left(\fnorm{\bm A_1^{(j)}}^2-\fnorm{\bm A_2^{(j)}}^2\right)\bm A_1^{(j)}}+\inner{\bm H_{\bm A_2,2}}{-\left(\fnorm{\bm A_1^{(j)}}^2-\fnorm{\bm A_2^{(j)}}^2\right)\bm A_2^{(j)}}.
\end{aligned}$$

The second and third terms can be lower bouned by
$$\begin{aligned}
    &\inner{\bm H_{\bm A_1,2}}{\left(\fnorm{\bm A_1^{(j)}}^2-\fnorm{\bm A_2^{(j)}}^2\right)\bm A_1^{(j)}}+\inner{\bm H_{\bm A_2,2}}{-\left(\fnorm{\bm A_1^{(j)}}^2-\fnorm{\bm A_2^{(j)}}^2\right)\bm A_2^{(j)}}\\
    \geq &-\left|\fnorm{\bm A_1^{(j)}}^2-\fnorm{\bm A_2^{(j)}}^2\right|\left(\fnorm{\bm A_1^{(j)}}\fnorm{\bm H_{\bm A_1,2}}+\fnorm{\bm A_2^{(j)}}\fnorm{\bm H_{\bm A_2,2}}\right)\\
    \geq &-20\phi^{4/3}\mathrm{dist}^2_{(j)}\left|\fnorm{\bm A_1^{(j)}}^2-\fnorm{\bm A_2^{(j)}}^2\right|.
\end{aligned}$$

For the first term, noting that $\bm Z^{*\top}\widetilde{\bm Z}^*=\widetilde{\bm Z}^{*\top}\bm Z^*=0$ and $\widetilde{\bm Z}^\top\bm Z^*=\bm Z^\top\widetilde{\bm Z}^*$, we have
$$\begin{aligned}
    &\inner{\bm Z-\bm Z^*}{\widetilde{\bm Z}\widetilde{\bm Z}^\top\bm Z}\\
    =&\inner{\widetilde{\bm Z}^\top\bm Z-\widetilde{\bm Z}^\top\bm Z^*}{\widetilde{\bm Z}^\top\bm Z}\\
    =&\frac{1}{2}\fnorm{\widetilde{\bm Z}^\top\bm Z}^2+\frac{1}{2}\inner{\widetilde{\bm Z}^\top\bm Z-2\widetilde{\bm Z}^\top\bm Z^*}{\widetilde{\bm Z}^\top\bm Z}\\
    =&\frac{1}{2}\fnorm{\widetilde{\bm Z}^\top\bm Z}^2+\frac{1}{2}\inner{\widetilde{\bm Z}^\top(\bm Z-\bm Z^*)}{\widetilde{\bm Z}^\top\bm Z}+\frac{1}{2}\inner{-\widetilde{\bm Z}^\top\bm Z^*}{\widetilde{\bm Z}^\top\bm Z}\\
    =&\frac{1}{2}\fnorm{\widetilde{\bm Z}^\top\bm Z}^2+\frac{1}{2}\inner{\widetilde{\bm Z}^\top(\bm Z-\bm Z^*)}{\widetilde{\bm Z}^\top\bm Z}+\frac{1}{2}\inner{\bm Z^{*\top}\widetilde{\bm Z}^*-\bm Z^\top\widetilde{\bm Z}^*}{\widetilde{\bm Z}^\top\bm Z}\\
    =&\frac{1}{2}\fnorm{\widetilde{\bm Z}^\top\bm Z}^2+\frac{1}{2}\inner{(\bm Z-\bm Z^*)^\top(\widetilde{\bm Z}-\widetilde{\bm Z}^*)}{\widetilde{\bm Z}^\top\bm Z}\\
    \geq &\frac{1}{2}\fnorm{\widetilde{\bm Z}^\top\bm Z}^2-\frac{1}{2}\fnorm{\bm Z-\bm Z^*}\fnorm{\widetilde{\bm Z}-\widetilde{\bm Z}^*}\fnorm{\widetilde{\bm Z}^\top\bm Z}
\end{aligned}$$
Note that $\fnorm{\bm Z-\bm Z^*}^2=\fnorm{\widetilde{\bm Z}-\widetilde{\bm Z}^*}^2=\fnorm{\bm A_1^{(j)}-\bm A_1^*}^2+\fnorm{\bm A_2^{(j)}-\bm A_2^*}^2$. By Lemma \ref{lemma:A1A2 and Error} with $c_b\leq 0.01$, we have $\fnorm{\bm A_1^{(j)}-\bm A_1^*}^2+\fnorm{\bm A_2^{(j)}-\bm A_2^*}^2\leq 50\phi^{4/3}\mathrm{dist}^2_{(j)}$, and hence,
$$
\inner{\bm Z-\bm Z^*}{\widetilde{\bm Z}\widetilde{\bm Z}^\top\bm Z} \geq \frac{1}{2}\left(\fnorm{\bm A_1^{(j)}}^2-\fnorm{\bm A_2^{(j)}}^2\right)^2-25\phi^{4/3}\mathrm{dist}^2_{(j)}\left|\fnorm{\bm A_1^{(j)}}^2-\fnorm{\bm A_2^{(j)}}^2\right|.
$$
Combining the two pieces, we have the lower bound of $G$ terms:
\begin{equation}
\begin{aligned}
    &\sum_{i=1}^{2}\left(G_{\bm D_i}+G_{\bm R_i}+G_{\bm P_i}+G_{\bm C_i}\right)\\
    \geq & \frac{1}{2}\left(\fnorm{\bm A_1^{(j)}}^2-\fnorm{\bm A_2^{(j)}}^2\right)^2-45\phi^{4/3}\mathrm{dist}^2_{(j)}\left|\fnorm{\bm A_1^{(j)}}^2-\fnorm{\bm A_2^{(j)}}^2\right|\\
    \geq &\frac{1}{4}\left(\fnorm{\bm A_1^{(j)}}^2-\fnorm{\bm A_2^{(j)}}^2\right)^2-1200\phi^{8/3}\mathrm{dist}^4_{(j)}\\
    \geq & \frac{1}{4}\left(\fnorm{\bm A_1^{(j)}}^2-\fnorm{\bm A_2^{(j)}}^2\right)^2-\frac{1200C_D}{\kappa^2}\phi^{10/3}\mathrm{dist}^2_{(j)}.
\end{aligned}
\end{equation}

~\\
\noindent\textit{Step 3.3.} (Lower bound for $\sum_{i=1}^{2}\left(T_{\bm R_i}+T_{\bm P_i}+T_{\bm C_i}\right)$)

\noindent Similarly to \citet{wang2023commonFactor}, it can be verified that
\begin{equation}\label{T_RPC}
\begin{aligned}
    &T_{\bm R_1}+T_{\bm P_1}+T_{\bm C_1}+T_{\bm R_2}+T_{\bm P_2}+T_{\bm C_2}\\
    =&\left\langle\left[\bm C_1^{(j)}~ \bm R_1^{(j)}\right]-\left[\bm C_1^*~ \bm R_1^*\right] \bm O_{1,u}^{(j)},\left[\bm C_1^{(j)}~ \bm R_1^{(j)}\right]\left(\left[\bm C_1^{(j)}~ \bm R_1^{(j)}\right]^{\top}\left[\bm C_1^{(j)}~ \bm R_1^{(j)}\right]-b^2 \bm{I}_{r_1}\right)\right\rangle \\ 
    +& \left\langle\left[\bm C_1^{(j)}~ \bm P_1^{(j)}\right]-\left[\bm C_1^*~ \bm P_1^*\right] \bm O_{1,v}^{(j)},\left[\bm C_1^{(j)}~ \bm P_1^{(j)}\right]\left(\left[\bm C_1^{(j)}~ \bm P_1^{(j)}\right]^{\top}\left[\bm C_1^{(j)}~ \bm P_1^{(j)}\right]-b^2 \bm{I}_{r_1}\right)\right\rangle\\
    +&\left\langle\left[\bm C_2^{(j)}~ \bm R_2^{(j)}\right]-\left[\bm C_2^*~ \bm R_2^*\right] \bm O_{2,u}^{(j)},\left[\bm C_2^{(j)}~ \bm R_2^{(j)}\right]\left(\left[\bm C_2^{(j)}~ \bm R_2^{(j)}\right]^{\top}\left[\bm C_2^{(j)}~ \bm R_2^{(j)}\right]-b^2 \bm{I}_{r_2}\right)\right\rangle \\ 
    +& \left\langle\left[\bm C_2^{(j)}~ \bm P_2^{(j)}\right]-\left[\bm C_2^*~ \bm P_2^*\right] \bm O_{2,v}^{(j)},\left[\bm C_2^{(j)}~ \bm P_2^{(j)}\right]\left(\left[\bm C_2^{(j)}~ \bm P_2^{(j)}\right]^{\top}\left[\bm C_2^{(j)}~ \bm P_2^{(j)}\right]-b^2 \bm{I}_{r_2}\right)\right\rangle.
\end{aligned}
\end{equation}
Denote $\bm U_1^{(j)}=[\bm C_1^{(j)}~\bm R_1^{(j)}], \bm V_1^{(j)}=[\bm C_1^{(j)}~\bm P_1^{(j)}],\bm U_1^*=[\bm C_1^*~\bm R_1^*], \bm V_1^*=[\bm C_1^*~\bm P_1^*]$. Recall that $\bm U_1^{*\top}\bm U_1^*=b^2\bm I_{r_1}$ and $\bm V_1^{*\top}\bm V_1^*=b^2\bm I_{r_1}$, for the first term we have 
$$\begin{aligned}
    &\left\langle\bm U_1^{(j)}-\bm U_1^* \bm O_{1,u}^{(j)},\bm U_1^{(j)}\left(\bm U_1^{(j)\top}\bm U_1^{(j)}-b^2 \bm{I}_{r_1}\right)\right\rangle\\
    =&\frac{1}{2}\inner{\bm U_1^{(j)\top}\bm U_1^{(j)}-\bm U_1^{*\top}\bm U_1^*}{\bm U_1^{(j)\top}\bm U_1^{(j)}-b^2 \bm{I}_{r_1}}+\frac{1}{2}\inner{\bm U_1^{(j)\top}(\bm U_1^{(j)}-\bm U_1^* \bm O_{1,u}^{(j)})}{\bm U_1^{(j)\top}\bm U_1^{(j)}-b^2 \bm{I}_{r_1}}\\
    +&\frac{1}{2}\inner{\bm U_1^{*\top}\bm U_1^*-\bm U_1^{(j)\top}\bm U_1^* \bm O_{1,u}^{(j)}}{\bm U_1^{(j)\top}\bm U_1^{(j)}-b^2 \bm{I}_{r_1}}.
\end{aligned}$$
Since $\bm U_1^{(j)\top}\bm U_1^{(j)}-b^2 \bm{I}_{r_1}$ is symmetric, we have
$$\begin{aligned}
    &\inner{\bm U_1^{*\top}\bm U_1^*-\bm U_1^{(j)\top}\bm U_1^* \bm O_{1,u}^{(j)}}{\bm U_1^{(j)\top}\bm U_1^{(j)}-b^2 \bm{I}_{r_1}}\\
    =&\inner{\bm U_1^{*\top}\bm U_1^*- \bm O_{1,u}^{(j)\top}\bm U_1^{*\top}\bm U_1^{(j)}}{\bm U_1^{(j)\top}\bm U_1^{(j)}-b^2 \bm{I}_{r_1}}\\
    =&\inner{\bm O_{1,u}^{(j)\top}\bm U_1^{*\top}(\bm U_1^*\bm O_{1,u}^{(j)}-\bm U_1^{(j)})}{\bm U_1^{(j)\top}\bm U_1^{(j)}-b^2 \bm{I}_{r_1}}.
\end{aligned}$$
Therefore, with condition \eqref{condition:E_upperBound},
$$
\begin{aligned}
    &\left\langle\bm U_1^{(j)}-\bm U_1^* \bm O_{1,u}^{(j)},\bm U_1^{(j)}\left(\bm U_1^{(j)\top}\bm U_1^{(j)}-b^2 \bm{I}_{r_1}\right)\right\rangle\\
    =&\frac{1}{2}\fnorm{\bm U_1^{(j)\top}\bm U_1^{(j)}-b^2\bm I_{r_1}}^2+\frac{1}{2}\inner{(\bm U_1^{(j)}-\bm U_1^* \bm O_{1,u}^{(j)})^\top(\bm U_1^{(j)}-\bm U_1^* \bm O_{1,u}^{(j)})}{\bm U_1^{(j)\top}\bm U_1^{(j)}-b^2\bm I_{r_1}}\\
    \geq &\frac{1}{2}\fnorm{\bm U_1^{(j)\top}\bm U_1^{(j)}-b^2\bm I_{r_1}}^2-\frac{1}{2}\mathrm{dist}^2_{(j)}\fnorm{\bm U_1^{(j)\top}\bm U_1^{(j)}-b^2\bm I_{r_1}}\\
    \geq &\frac{1}{4}\fnorm{\bm U_1^{(j)\top}\bm U_1^{(j)}-b^2\bm I_{r_1}}^2-\frac{1}{4}\mathrm{dist}^4_{(j)}\\
    \geq &\frac{1}{4}\fnorm{\bm U_1^{(j)\top}\bm U_1^{(j)}-b^2\bm I_{r_1}}^2-\frac{C_D}{4}\phi^{2/3}\mathrm{dist}^2_{(j)}.
\end{aligned}
$$
Adding the lower bounds of the four terms in \eqref{T_RPC} together, we have
$$\begin{aligned}
    &T_{\bm R_1}+T_{\bm P_1}+T_{\bm C_1}+T_{\bm R_2}+T_{\bm P_2}+T_{\bm C_2}\\
    \geq &\frac{1}{4}\sum_{i=1}^{2}\left(\fnorm{\left[\bm C_i^{(j)}~ \bm R_i^{(j)}\right]^{\top}\left[\bm C_i^{(j)}~ \bm R_i^{(j)}\right]-b^2 \bm I_{r_i}}^2+\fnorm{\left[\bm C_i^{(j)}~ \bm P_i^{(j)}\right]^{\top}\left[\bm C_i^{(j)}~ \bm P_i^{(j)}\right]-b^2 \bm I_{r_i}}^2\right)\\
    -&C_D\phi^{2/3}\mathrm{dist}^2_{(j)}.
\end{aligned}$$

~\\
\noindent\textit{Step 3.4.} (Upper bound for $\sum_{i=1}^{2}\left(Q_{\bm D_i, 2}+Q_{\bm R_i, 2}+Q_{\bm P_i, 2}+Q_{\bm C_i, 2}\right)$)

\noindent Following the definitions and plugging in $b=\phi^{1/3}$, we have
\begin{equation}\label{Q2_upperBound}
\begin{aligned}
    &\sum_{i=1}^{2}\left(Q_{\bm D_i, 2}+Q_{\bm R_i, 2}+Q_{\bm P_i, 2}+Q_{\bm C_i, 2}\right)\\
    =&232\phi^{10/3}\left(\xi^2(r_1,r_2,d_1,d_2)+\fnorm{\nabla\widetilde{\mathcal{L}}(\bm A^{(j)})-\nabla\widetilde{\mathcal{L}}(\bm A^*)}^2\right)\\
    +&\lambda_2^2\phi^{2/3}\sum_{i=1}^2\left(18\fnorm{\bm C_i^{(j)\top}\bm R_i^{(j)}}^2+18\fnorm{\bm C_i^{(j)\top}\bm P_i^{(j)}}^2+40\fnorm{\bm C_i^{(j)\top}\bm C_i-b^2\bm I_{d_i}}^2\right.\\
    +&\left.8\fnorm{\bm R_i^{(j)\top}\bm R_i-b^2\bm I_{r_i-d_i}}^2+8\fnorm{\bm P_i^{(j)\top}\bm P_i-b^2\bm I_{r_i-d_i}}^2\right)\\
    +&120\lambda_1^2\phi^{10/3}\left(\fnorm{\bm A_1^{(j)}}^2-\fnorm{\bm A_2^{(j)}}^2\right)^2\\
    \leq &232\phi^{10/3}\left(\xi^2(r_1,r_2,d_1,d_2)+\fnorm{\nabla\widetilde{\mathcal{L}}(\bm A^{(j)})-\nabla\widetilde{\mathcal{L}}(\bm A^*)}^2\right)\\
    +&20\lambda_2^2\phi^{2/3}\sum_{i=1}^2\left(\fnorm{\left[\bm C_i^{(j)}~ \bm R_i^{(j)}\right]^{\top}\left[\bm C_i^{(j)}~ \bm R_i^{(j)}\right]-b^2 \bm I_{r_i}}^2+\fnorm{\left[\bm C_i^{(j)}~ \bm P_i^{(j)}\right]^{\top}\left[\bm C_i^{(j)}~ \bm P_i^{(j)}\right]-b^2 \bm I_{r_i}}^2\right)\\
    +&120\lambda_1^2\phi^{10/3}\left(\fnorm{\bm A_1^{(j)}}^2-\fnorm{\bm A_2^{(j)}}^2\right)^2.
\end{aligned}
\end{equation}

\noindent\textit{Step 3.5.} (Lower bound for $\fnorm{\bm A_2^{(j)}\otimes\bm A_1^{(j)}-\bm A_2^*\otimes\bm A_1^*}^2$)

\noindent So far, we have derived bounds for all parts in \eqref{E upper bound: Q2Q1GT}. Combining these pieces, we have
\begin{equation}\label{E upper bound 1}
\begin{aligned}
    &\mathrm{dist}^2_{(j+1)}-\mathrm{dist}^2_{(j)}\\
    \leq&-\alpha\eta\fnorm{\bm A_2^{(j)}\otimes\bm A_1^{(j)}-\bm A_2^*\otimes\bm A_1^*}^2\\
    -&2\eta\left(-\frac{C_H\alpha\phi^{10/3}}{\kappa^2} -36c\phi^{10/3}-C_D\lambda_2\phi^{2/3}-\frac{1200\lambda_1 C_D}{\kappa^2}\phi^{10/3}\right)\mathrm{dist}^2_{(j)}\\
    +&\left(232\eta^2 \phi^{10/3}-\frac{\eta}{2\beta}\right)\fnorm{\nabla\widetilde{\mathcal{L}}(\bm A^{(j)})-\nabla\widetilde{\mathcal{L}}(\bm A^*)}^2\\
    +&\left(232\eta^2 \phi^{10/3}+\frac{2\eta}{c}\right)\xi^2(r_1,r_2,d_1,d_2)\\
    +&\left(20\lambda_2^2\eta^2\phi^{2/3}-\frac{\lambda_2\eta}{2}\right)\\
    \times &\sum_{i=1}^2\left(\fnorm{\left[\bm C_i^{(j)}~ \bm R_i^{(j)}\right]^{\top}\left[\bm C_i^{(j)}~ \bm R_i^{(j)}\right]-b^2 \bm I_{r_i}}^2+\fnorm{\left[\bm C_i^{(j)}~ \bm P_i^{(j)}\right]^{\top}\left[\bm C_i^{(j)}~ \bm P_i^{(j)}\right]-b^2 \bm I_{r_i}}^2\right)\\
    +&\left(120\lambda_1^2\eta^2\phi^{10/3}-\frac{1}{2}\lambda_1\eta\right)\left(\fnorm{\bm A_1^{(j)}}^2-\fnorm{\bm A_2^{(j)}}^2\right)^2.
\end{aligned}
\end{equation}

Next, we construct a lower bound of $\fnorm{\bm A_2^{(j)}\otimes\bm A_1^{(j)}-\bm A_2^*\otimes\bm A_1^*}^2$ related to $\mathrm{dist}^2_{(j)}$. Viewing $\V{\bm A_2}$ and $\V{\bm A_1}$ as the factors in $\V{\bm A_2}\V{\bm A_1}^\top$, our first regularizer 
\begin{equation}
    \left(\fnorm{\bm A_1}^2-\fnorm{\bm A_2}^2\right)^2=\left(\V{\bm A_1}^\top\V{\bm A_1}-\V{\bm A_2}^\top\V{\bm A_2}\right)^2
\end{equation}
can be regarded as a special case of \cite{Guquanquan2017}, as well as a generalization of \cite{Tu2016}. Define the distance of two vectors up to a sign switch as
$$\mathrm{dist}^2(\V{\bm A_i^{(j)}},\V{\bm A_i^*})=\min_{s=\pm 1}\fnorm{\V{\bm A_i^{(j)}}-\V{\bm A_i^*}s}^2, \quad i=1,2.$$
For more details of this type of distance, see \cite{Cai2018perturbtion}. 
By Lemma \ref{lemma:A1A2 and Error}, since $C_D$ is small,
$$\fnorm{\bm A_1^{(j)}-\bm A_1^*}^2+\fnorm{\bm A_2^{(j)}-\bm A_2^*}^2\leq 30\phi^{4/3}D\leq 30C_D\phi^2\leq \phi^2, \quad i=1,2.$$
Then, we have $\fnorm{\bm A_i^{(j)}-\bm A_i^*}\leq \fnorm{\bm A_i^*}$. Hence, we know that the angle between $\V{\bm A_i^{(j)}}$ and $\V{\bm A_i^*}$ is an acute angle. Thus,
$$\mathrm{dist}^2(\V{\bm A_i^{(j)}},\V{\bm A_i^*})=\fnorm{\V{\bm A_i^{(j)}}-\V{\bm A_i^*}}^2,\quad i=1,2.$$ 
Meanwhile, applying the permutation operator $\mathcal{P}$, the Frobenius norm remains unchanged
$$\fnorm{\bm A_2^{(j)}\otimes\bm A_1^{(j)}-\bm A_2^*\otimes\bm A_1^*}^2=\fnorm{\V{\bm A_2^{(j)}}\V{\bm A_1^{(j)}}^\top-\V{\bm A_2^*}\V{\bm A_1^*}^\top}^2.$$

Recalling the notations defined in Step 3.5 and applying Lemma \ref{lemma: Tu16, dist(U,X) and UU'-XX'}, we have
\begin{equation}\label{A1A2 and kronecker product}
\begin{aligned}
    &\fnorm{\bm A_1^{(j)}-\bm A_1^*}^2+\fnorm{\bm A_2^{(j)}-\bm A_2^*}^2\\
    =&\mathrm{dist}^2(\bm Z,\bm Z^*)\\
    \leq &\frac{1}{4(\sqrt{2}-1)\phi^2}\fnorm{\bm Z\bm Z^\top-\bm Z^*\bm Z^{*\top}}^2\\
    =&\frac{1}{4(\sqrt{2}-1)\phi^2}\left(2\fnorm{\bm A_2^{(j)}\otimes\bm A_1^{(j)}-\bm A_2^*\otimes\bm A_1^*}^2\right.\\
    &+\left.\fnorm{\V{\bm A_1^{(j)}}\V{\bm A_1^{(j)}}^\top-\V{\bm A_1^*}\V{\bm A_1^*}^\top}^2\right.\\
    &+\left.\fnorm{\V{\bm A_2^{(j)}}\V{\bm A_2^{(j)}}^\top-\V{\bm A_2^*}\V{\bm A_2^*}^\top}^2\right).
\end{aligned}
\end{equation}
To erase the last two terms, we note that
\begin{equation}\label{regularizer1 and kronecker product}
\begin{aligned}
&\left(\fnorm{\bm A_1^{(j)}}^2-\fnorm{\bm A_2^{(j)}}^2\right)^2\\
=&\fnorm{\widetilde{\bm Z}^\top\bm Z}^2\\
=&\inner{\bm Z\bm Z^\top}{\widetilde{\bm Z}\widetilde{\bm Z}^\top}\\
=&\inner{\bm Z\bm Z^\top-\bm Z^*\bm Z^{*^\top}}{\widetilde{\bm Z}\widetilde{\bm Z}^\top-\widetilde{\bm Z}^*\widetilde{\bm Z}^{*\top}}+\inner{\bm Z\bm Z^\top}{\widetilde{\bm Z}^*\widetilde{\bm Z}^{*\top}}+\inner{\bm Z^*\bm Z^{*^\top}}{\widetilde{\bm Z}\widetilde{\bm Z}^\top}-\inner{\bm Z^*\bm Z^{*^\top}}{\widetilde{\bm Z}^*\widetilde{\bm Z}^{*\top}}\\
=&\inner{\bm Z\bm Z^\top-\bm Z^*\bm Z^{*^\top}}{\widetilde{\bm Z}\widetilde{\bm Z}^\top-\widetilde{\bm Z}^*\widetilde{\bm Z}^{*\top}}+2\fnorm{\bm Z^\top\widetilde{\bm Z}^*}^2-\underbrace{\fnorm{\bm Z^{*\top}\widetilde{\bm Z}^*}^2}_{=0}\\
\geq &\inner{\bm Z\bm Z^\top-\bm Z^*\bm Z^{*^\top}}{\widetilde{\bm Z}\widetilde{\bm Z}^\top-\widetilde{\bm Z}^*\widetilde{\bm Z}^{*\top}}\\
=&\fnorm{\V{\bm A_1^{(j)}}\V{\bm A_1^{(j)}}^\top-\V{\bm A_1^*}\V{\bm A_1^*}^\top}^2+\fnorm{\V{\bm A_2^{(j)}}\V{\bm A_2^{(j)}}^\top-\V{\bm A_2^*}\V{\bm A_2^*}^\top}^2\\
-&2\fnorm{\bm A_2^{(j)}\otimes\bm A_1^{(j)}-\bm A_2^*\otimes\bm A_1^*}^2.
\end{aligned}
\end{equation}
Then, combine \eqref{A1A2 and kronecker product} and \eqref{regularizer1 and kronecker product}, we have
\begin{equation}\label{A1A2 and kron and regularizer1}
    \fnorm{\bm A_1^{(j)}-\bm A_1^*}^2+\fnorm{\bm A_2^{(j)}-\bm A_2^*}^2\leq \frac{4}{\phi^2}\fnorm{\bm A_2^{(j)}\otimes\bm A_1^{(j)}-\bm A_2^*\otimes\bm A_1^*}^2+\frac{1}{\phi^2}\left(\fnorm{\bm A_1^{(j)}}^2-\fnorm{\bm A_2^{(j)}}^2\right)^2.
\end{equation}
Next, by Lemma \ref{lemma:A1A2 and Error} and $c_b\leq 0.01$, we have
$$\begin{aligned}
    \mathrm{dist}^2_{(j)} &\leq 100\phi^{-4/3}\kappa^2\left(\left\|\mathbf{A}_1^{(j)}-\mathbf{A}_1^*\right\|_{\mathrm{F}}^2+\left\|\mathbf{A}_2^{(j)}-\mathbf{A}_2^*\right\|_{\mathrm{F}}^2\right) \\
    & +24\phi^{-2/3}\sum_{i=1}^{2}\left(\fnorm{\left[\bm C_i^{(j)}~ \bm R_i^{(j)}\right]^{\top}\left[\bm C_i^{(j)}~ \bm R_i^{(j)}\right]-b^2 \bm I_{r_i}}^2+\fnorm{\left[\bm C_i^{(j)}~\bm P_i^{(j)}\right]^{\top}\left[\bm C_i^{(j)}~\bm P_i^{(j)}\right]-b^2 \bm I_{r_i}}^2\right).
\end{aligned}$$
Then, we obtain a lower bound:
\begin{equation}\label{LB of kronecker w.r.t E}
\begin{aligned}
    &\fnorm{\bm A_2^{(j)}\otimes\bm A_1^{(j)}-\bm A_2^*\otimes\bm A_1^*}^2\\
    \geq& \frac{\phi^{10/3}}{400\kappa^2}\mathrm{dist}^2_{(j)}-\frac{1}{4}\left(\fnorm{\bm A_1^{(j)}}^2-\fnorm{\bm A_2^{(j)}}^2\right)^2\\
    -&\frac{3\phi^{8/3}}{50\kappa^2}\sum_{i=1}^{2}\left(\fnorm{\left[\bm C_i^{(j)}~ \bm R_i^{(j)}\right]^{\top}\left[\bm C_i^{(j)}~ \bm R_i^{(j)}\right]-b^2 \bm I_{r_i}}^2+\fnorm{\left[\bm C_i^{(j)}~\bm P_i^{(j)}\right]^{\top}\left[\bm C_i^{(j)}~\bm P_i^{(j)}\right]-b^2 \bm I_{r_i}}^2\right).
\end{aligned}
\end{equation}

Finally, plugging the lower bound \eqref{LB of kronecker w.r.t E} into \eqref{E upper bound 1}, we have the following upper bound:

\begin{equation}\label{E upper bound 2}
    \begin{aligned}
        \mathrm{dist}^2_{(j+1)}\leq&\left(1-2\eta\left(\frac{\alpha\phi^{10/3}}{800\kappa^2}-\frac{C_h^2C_D\alpha\phi^{10/3}}{\kappa^2} -36c\phi^{10/3}-C_D\lambda_2\phi^{2/3}-\frac{1200\lambda_1 C_D}{\kappa^2}\phi^{10/3}\right)\right)\mathrm{dist}^2_{(j)}\\
        +&\eta\left(\eta 232\phi^{10/3}+\frac{2}{c}\right)\xi^2(r_1,r_2,d_1,d_2)\\
        +&\eta\left(232\eta \phi^{10/3}-\frac{1}{2\beta}\right)\fnorm{\nabla\widetilde{\mathcal{L}}(\bm A^{(j)})-\nabla\widetilde{\mathcal{L}}(\bm A^*)}^2\\
        +&\frac{1}{2}\eta\left(40\lambda_2^2\eta\phi^{2/3}+\frac{3\alpha\phi^{8/3}}{25\kappa^2}-\lambda_2\right)\\
        &\times\sum_{i=1}^2\left(\fnorm{\left[\bm C_i^{(j)}~ \bm R_i^{(j)}\right]^{\top}\left[\bm C_i^{(j)}~ \bm R_i^{(j)}\right]-b^2 \bm I_{r_i}}^2+\fnorm{\left[\bm C_i^{(j)}~ \bm P_i^{(j)}\right]^{\top}\left[\bm C_i^{(j)}~ \bm P_i^{(j)}\right]-b^2 \bm I_{r_i}}^2\right)\\
        +&\frac{1}{4}\eta\left(\alpha+480\lambda_1^2\eta\phi^{10/3}-2\lambda_1\right)\left(\fnorm{\bm A_1^{(j)}}^2-\fnorm{\bm A_2^{(j)}}^2\right)^2.
    \end{aligned}
\end{equation}

\newpage
\noindent\textbf{\textit{Step 4.}} (\textbf{Convergence analysis of $\mathrm{dist}^2_{(j+1)}$})

\noindent For the tuning parameters $\eta, \lambda_2$, $\lambda_1$ and $c$, let
$$\eta=\frac{\eta_0}{\beta\phi^{10/3}},\quad \lambda_2=\frac{\alpha\phi^{8/3}}{\kappa^2},\quad \lambda_1=\alpha,\quad\text{and}\quad c=\frac{C_D\alpha}{\kappa^2}.$$
We see that with $\eta_0\leq 1/480$ and $C_D$ being small enough, the error is reduced by iteration.

For the coefficient of the third term of \eqref{E upper bound 2},
$$
232\eta \phi^{10/3}-\frac{1}{2\beta}\leq\frac{232}{480\beta}-\frac{1}{2\beta}\leq 0.$$
For the fourth term of \eqref{E upper bound 2},
$$\begin{aligned}
    40\lambda_2^2\eta\phi^{2/3}+\frac{3\alpha\phi^{8/3}}{25\kappa^2}-\lambda_2\leq& \frac{40\alpha^2\phi^{8/3}}{480\kappa^4\beta}+\frac{3\alpha\phi^{8/3}}{25\kappa^2}-\frac{\alpha\phi^{8/3}}{\kappa^2}\\
    \leq &\frac{\alpha\phi^{8/3}}{\kappa^2}\left(\frac{1}{12}+\frac{3}{25}-1\right)\\
    \leq &0.
\end{aligned}$$
For the fifth term of \eqref{E upper bound 2},
$$\begin{aligned}
    \alpha+480\lambda_1^2\eta\phi^{10/3}-2\lambda_1\leq& \alpha+\frac{\alpha^2}{\beta}-2\alpha\leq (1+1-2)\alpha\leq 0.
\end{aligned}$$

Form now on, for the sake of simplicity, $C$ will denote a constant whose exact value may change in different contexts. For the second term of \eqref{E upper bound 2},
$$
    \eta\left(\eta 232\phi^{10/3}+\frac{2}{c}\right)\leq \frac{\eta_0}{\beta\phi^{10/3}}\left(\frac{232}{480\beta}+\frac{2\kappa^2}{C_D\alpha}\right)\leq \frac{C\eta_0\kappa^2}{\alpha\beta\phi^{10/3}}.
$$
For the first term, i.e., the coefficient of $\mathrm{dist}^2_{(j)}$,
$$\begin{aligned}
    &\frac{\alpha\phi^{10/3}}{800\kappa^2}-\frac{C_h^2C_D\alpha\phi^{10/3}}{\kappa^2} -36c\phi^{10/3}-C_D\lambda_2\phi^{2/3}-\frac{1200\lambda_1 C_D}{\kappa^2}\phi^{10/3}\\
    =&\frac{\alpha\phi^{10/3}}{\kappa^2}\left(\frac{1}{800}-\left(C_h^2+1237\right)C_D\right)\\
    :=&\frac{\alpha\phi^{10/3}}{\kappa^2}\times \frac{C_0}{2}.
\end{aligned}$$
Therefore, we can derive the following recursive relationship:
\begin{equation}\label{recursive of E}
    \mathrm{dist}^2_{(j+1)}\leq (1-C_0\eta_0\alpha\beta^{-1}\kappa^{-2})\mathrm{dist}^2_{(j)}+C\eta_0\kappa^2\alpha^{-1}\beta^{-1}\phi^{-10/3}\xi^2(r_1,r_2,d_1,d_2).
\end{equation}
When $C_D$ is small enough, $C_0$ is a positive constant and $1-C_0\eta_0\alpha\beta^{-1}\kappa^{-2}$ is a constant smaller than 1. Hence, the computational error can be reduced through the proposed gradient decent algorithm.\\

\noindent \textbf{\textit{Step 5.}} (\textbf{Verification of the conditions})

\noindent In this step, we verify that the conditions \eqref{condition: upper bound of pieces} and \eqref{condition:E_upperBound} in the convergence analysis hold recursively.

For $j=0$, since initialization condition $\mathrm{dist}^2_{(0)}\leq C_D\alpha\beta^{-1}\kappa^{-2}\phi^{2/3}$ holds, we have
\begin{equation}
    \begin{aligned}
        \fnorm{[\bm C_i^{(0)}~ \bm R_i^{(0)}]}&\leq \fnorm{[\bm C_i^{(0)}~ \bm R_i^{(0)}]-[\bm C_i^*~ \bm R_i^*]\bm O_{i,u}^{(0)}}+\fnorm{[\bm C_i^*~ \bm R_i^*]\bm O_{i,u}^{(0)}}\\
        &\leq \mathrm{dist}_{(0)}+\phi^{1/3}\\
        &\leq (1+c_b)b^{1/3},\\
        \fnorm{[\bm C_i^{(0)}~ \bm P_i^{(0)}]}&\leq \fnorm{[\bm C_i^{(0)}~ \bm P_i^{(0)}]-[\bm C_i^*~ \bm P_i^*]\bm O_{i,v}^{(0)}}+\fnorm{[\bm C_i^*~ \bm P_i^*]\bm O_{i,v}^{(0)}}\\
        &\leq \mathrm{dist}_{(0)}+\phi^{1/3}\\
        &\leq (1+c_b)b^{1/3},
    \end{aligned}
\end{equation} 
and
\begin{equation}
\begin{aligned}
    \fnorm{\bm D_i^{(0)}}&\leq \fnorm{\bm D_i^{(0)}-\bm O_{i,u}^{(0)\top}\bm D_i^*\bm O_{i,v}^{(0)}}+\fnorm{\bm O_{i,u}^{(0)\top}\bm D_i^*\bm O_{i,v}^{(0)}}\\
    &\leq \sqrt{\mathrm{dist}^2_{(0)}}+\phi^{1/3}\\
    &\leq \frac{(1+c_b)\phi}{b^2}, \quad \mathrm{for~} i=1,2.
\end{aligned}
\end{equation}

Then, suppose \eqref{condition:E_upperBound} and \eqref{condition: upper bound of pieces} hold at step $j$, for $j+1$, we have
\begin{equation}
\begin{aligned}
    \mathrm{dist}^2_{(j+1)}&\leq(1-C_0\eta_0\alpha\beta^{-1}\kappa^{-2})\mathrm{dist}^2_{(j)}+C\eta_0\kappa^2\alpha^{-1}\beta^{-1}\phi^{-10/3}\xi^2(r_1,r_2,d_1,d_2)\\
    &\leq \frac{C_D\alpha\phi^{2/3}}{\beta\kappa^2}-\eta_0\phi^{2/3}\alpha^2\beta^{-1}\kappa^2\left(\frac{C_0 C_D}{\beta\kappa^6}-\frac{C\xi^2(r_1,r_2,d_1,d_2)}{\alpha^3\phi^{4}}\right).
\end{aligned}
\end{equation}
Since $\phi^4\geq C\beta\kappa^6\xi^2\alpha^{-3}$ for some universally big constant, we can verify that
$$\frac{C_2}{\beta\kappa^6}-\frac{C_1\xi^2(r_1,r_2,d_1,d_2)}{\alpha^3\phi^{4}}\geq 0,$$
then 
$$\mathrm{dist}^2_{(j+1)}\leq \frac{C_D\alpha\phi^{2/3}}{\beta\kappa^2}.$$
By the same argument as $j=0$, we can verify that condition \eqref{condition: upper bound of pieces} holds. Then the induction is completed.\\

\noindent \textbf{\textit{Step 6.}} (\textbf{Construction of the upper bounds})

In the preceding steps, we proved that \eqref{condition:E_upperBound}, \eqref{condition: upper bound of pieces}, and the recursive relationship \eqref{recursive of E} hold for any $j\geq 0$. By iterating this relationship and summing the resulting geometric series, noting that $\sum_{j=0}^{\infty}(1-C_0\eta_0\alpha\beta^{-1}\kappa^{-2})^j = C_0^{-1}\eta_0^{-1}\alpha^{-1}\beta\kappa^2$, we obtain the upper bound of $\mathrm{dist}^2_{(j)}$:
\begin{equation}
    \mathrm{dist}^2_{(j)}\leq (1-C_0\eta_0\alpha\beta^{-1}\kappa^{-2})^{j}\mathrm{dist}^2_{(0)}+C\kappa^4\alpha^{-2}\phi^{-10/3}\xi^2(r_1,r_2,d_1,d_2),
\end{equation}

For the error bound of $\bm A_1$ and $\bm A_2$, by Lemma \ref{lemma:A1A2 and Error}, we have
$$\begin{aligned}
&\fnorm{\bm A_1^{(j)}-\bm A_1^*}^2+\fnorm{\bm A_2^{(j)}-\bm A_2^*}^2\\
\leq& C\phi^{4/3}\mathrm{dist}^2_{(j)}\\
\leq& C\phi^{4/3}(1-C_0\eta_0\alpha\beta^{-1}\kappa^{-2})^{j}\mathrm{dist}^2_{(0)}+C\kappa^4\alpha^{-2}\phi^{-2}\xi^2(r_1,r_2,d_1,d_2)\\
\leq&C\kappa^2(1-C_0\eta_0\alpha\beta^{-1}\kappa^{-2})^{j}\left(\fnorm{\bm A_1^{(0)}-\bm A_1^*}^2+\fnorm{\bm A_2^{(0)}-\bm A_2^*}^2\right)\\
&+C\kappa^4\alpha^{-2}\phi^{-2}\xi^2(r_1,r_2,d_1,d_2).
\end{aligned}$$
Also, for the error bound of $\fnorm{\bm A_2^{(j)}\otimes\bm A_1^{(j)}-\bm A_2^*\otimes\bm A_1^*}^2$, we have 
$$\begin{aligned}
    &\fnorm{\bm A_2^{(j)}\otimes\bm A_1^{(j)}-\bm A_2^*\otimes\bm A_1^*}^2\\
    \leq &2\fnorm{\left(\bm A_2^{(j)}-\bm A_2^*\right)\otimes \bm A_1^{(j)}}^2+2\fnorm{\bm A_2^*\otimes \left(\bm A_1^{(j)}-\bm A_1^*\right)}^2\\
    \leq &8\phi^{2}\left(\fnorm{\bm A_1^{(j)}-\bm A_1^*}^2+\fnorm{\bm A_2^{(j)}-\bm A_2^*}^2\right)\\
    \leq &C\phi^2\kappa^2(1-C_0\eta_0\alpha\beta^{-1}\kappa^{-2})^{j}\left(\fnorm{\bm A_1^{(0)}-\bm A_1^*}^2+\fnorm{\bm A_2^{(0)}-\bm A_2^*}^2\right)\\
    &+C\kappa^4\alpha^{-2}\xi^2(r_1,r_2,d_1,d_2)\\
    \leq &C\kappa^2(1-C_0\eta_0\alpha\beta^{-1}\kappa^{-2})^{j}\left(\fnorm{\bm A_2^{(0)}\otimes\bm A_1^{(0)}-\bm A_2^*\otimes\bm A_1^*}^2\right)\\
    &+C\kappa^4\alpha^{-2}\xi^2(r_1,r_2,d_1,d_2).
\end{aligned}$$
The last inequality is from \eqref{A1A2 and kron and regularizer1} with the fact that $\fnorm{\bm A_1^{(0)}}=\fnorm{\bm A_2^{(0)}}$.

\subsection{Auxiliary Lemmas}\label{sub:auxiliary_lemmas}

The first lemma states that when the Frobenius norm of estimated parameter matrices are close to their true values, the Frobenius norm of high-order perturbation terms can be controlled by the running error in \eqref{eq:running_error}.
\begin{lemma}\label{lemma:H_upperBound}
Define  the following matrices
$$\begin{aligned}
    &\bm A_1^*=[\bm C_1^*~ \bm R_1^*]\bm D_1^*[\bm C_1^*~ \bm P_1^*]^\top, \bm A_2^*=[\bm C_2^*~ \bm R_2^*]\bm D_2^*[\bm C_2^*~ \bm P_2^*]^\top,\\
    &\bm A_1=[\bm C_1~ \bm R_1]\bm D_1[\bm C_1~ \bm P_1]^\top, \text{ and } \bm A_2=[\bm C_2~ \bm R_2]\bm D_2[\bm C_2~ \bm P_2]^\top,
\end{aligned}$$ 
where $\bm D_i^*$, $\bm D_i \in\mathbb{R}^{r_i \times r_i}$, $\bm C_i^*, \bm C_i \in \mathbb{R}^{p_i \times d_i}$, and $\bm R_i, \bm R_i^*, \bm P_i, \bm P_i^* \in \mathbb{R}^{p_i \times(r_i-d_i)}$, for $i=1,2$. Meanwhile, $\fnorm{\bm A_1^*}=\fnorm{\bm A_2^*}=\phi$ and $\left[\mathbf{C}_i^*~ \mathbf{R}_i^*\right]^{\top}\left[\mathbf{C}_i^*~ \mathbf{R}_i^*\right]=\left[\mathbf{C}_i^*~ \mathbf{P}_i^*\right]^{\top}\left[\mathbf{C}_i^*~ \mathbf{P}_i^*\right]=b^2\mathbf{I}_{r_i}$ for some balance scalar $b$. Suppose that there exists a constant $c$, such that
\begin{equation}\label{LemmaA1Condition}
    \fnorm{\left[\bm C_i~ \bm R_i\right]}\leq (1+c)b,\quad \fnorm{\left[\bm C_i~ \bm P_i\right]}\leq (1+c)b,
    \text{and }\fnorm{\bm D_i}\leq \frac{(1+c)\phi}{b^2},\quad \text{for } i=1,2.
\end{equation}

For $i=1,2$, let
$$\begin{gathered}
    \mathcal{E}_{i,u}=[\bm C_i^*~ \bm R_i^*]\bm O_{i,u}-[\bm C_i~ \bm R_i],\\
    \mathcal{E}_{i,v}=[\bm C_i^*~ \bm P_i^*]\bm O_{i,v}-[\bm C_i~ \bm P_i],\\
    \mathcal{E}_{i,D}=\bm O_{i,u}^\top\bm D_i^*\bm O_{i,v}-\bm D_i,
\end{gathered}$$
where $\bm O_{i,c} \in \mathbb{O}^{d_i \times d_i}, \bm O_{i,r}, \bm O_{i,p} \in \mathbb{O}^{(r_i-d_i) \times(r_i-d_i)}, \bm O_{i,u}=$ $\operatorname{diag}\left(\bm O_{i,c}, \bm O_{i,r}\right), \bm O_{i,v}=\operatorname{diag}\left(\bm O_{i,c}, \bm O_{i,p}\right).$
Then, let
$$\begin{aligned}
    \bm H_1:=&-\bm A_2\otimes\left([\bm C_1~ \bm R_1]\bm D_1\mathcal{E}_{1,v}^\top+[\bm C_1~ \bm R_1]\mathcal{E}_{1,D}[\bm C_1~ \bm P_1]^\top+\mathcal{E}_{1,u}\bm D_1[\bm C_1~ \bm P_1]^\top\right)\\
    &-\left([\bm C_2~ \bm R_2]\bm D_2\mathcal{E}_{2,v}^\top+[\bm C_2~ \bm R_2]\mathcal{E}_{2,D}[\bm C_2~ \bm P_2]^\top+\mathcal{E}_{2,u}\bm D_2[\bm C_2~ \bm P_2]^\top\right)\otimes\bm A_1,
\end{aligned}$$
which contains the first-order perturbation terms, and
$$
    \bm H:=\bm A_2^*\otimes\bm A_1^*-\bm A_2\otimes\bm A_1+\bm H_1,
$$
which represents second- and higher-order terms perturbated from $\bm A_2\otimes\bm A_1$.

Meanwhile, define the distance as
$$\begin{aligned}
    D:=&\min _{\substack{\mathbf{O}_{i,c}' \in \mathbb{O}^{d_i \times d} \\ \mathbf{O}_{i,r}', \mathbf{O}_{i,p}'\in \mathbb{O}^{(r_i-d_i) \times(r_i-d_i)}\\i=1,2}}\sum_{i=1,2}\Big\{  \|\bm{C}_i-\bm{C}^*\bm{O}_{i,c}'\|_\mathrm{F}^2 + \|\bm{R}_i-\bm{R}^*\bm{O}_{i,r}'\|_\mathrm{F}^2 + \|\bm{P}_i-\bm{P}^*\bm{O}_{i,p}'\|_\mathrm{F}^2\\
         &\qquad\qquad\qquad\qquad\qquad\qquad\qquad + \|\bm{D}_i-\mathrm{diag}(\bm{O}_{i,c}',\bm{O}_{i,r}')^\top\bm{D}_i^*\mathrm{diag}(\bm{O}_{i,c}',\bm{O}_{i,p}')\|_\mathrm{F}^2\Big\}\\
         :=&\sum_{i=1,2}\Big\{  \|\bm{C}_i-\bm{C}^*\bm{O}_{i,c}\|_\mathrm{F}^2 + \|\bm{R}_i-\bm{R}^*\bm{O}_{i,r}\|_\mathrm{F}^2 + \|\bm{P}_i-\bm{P}^*\bm{O}_{i,p}\|_\mathrm{F}^2+\|\bm{D}_i-\bm O_{i,u}^\top\bm{D}_i^*\bm O_{i,v}\|_\mathrm{F}^2\Big\}.\\
\end{aligned}$$

Assume that $D\leq C_D\phi^{2/3}$ for some constant $C_D$. If $b\asymp\phi^{1/3}$, there exists a constant $C_h$, such that
$$
\fnorm{\bm H}\leq C_h\phi^{4/3}D.
$$
\end{lemma}
\begin{proof}[Proof of Lemma \ref{lemma:H_upperBound}]
    We start from the decomposing the perturbated matrix $\bm A_2^*\otimes\bm A_1^*$. Using notations defined above, we split the perturbated matrix to 64 terms and classify them into three types: zeroth-order perturbation, fisrt-order perturbation, and high-order perturbation. Then we control the Frobenius norm of high-order perturbation, $\fnorm{\bm H}$.
    \begin{equation}\begin{aligned}
        &\bm A_2^*\otimes\bm A_1^*\\
        =&[\bm C_2^*~ \bm R_2^*]\bm D_2^*[\bm C_2^*~ \bm P_2^*]^\top\otimes[\bm C_1^*~ \bm R_1^*]\bm D_1^*[\bm C_1^*~ \bm P_1^*]^\top\\
        =&\left([\bm C_2~ \bm R_2]+\mathcal{E}_{2,u}\right)\left(\bm D_2+\mathcal{E}_{2,D}\right)\left([\bm C_2~ \bm P_2]+\mathcal{E}_{2,v}\right)^\top\otimes\left([\bm C_1~ \bm R_1]+\mathcal{E}_{1,u}\right)\left(\bm D_1+\mathcal{E}_{1,D}\right)\left([\bm C_1~ \bm P_1]+\mathcal{E}_{1,v}\right)^\top\\
        =&\begin{pmatrix}
            [\bm C_2~ \bm R_2]\bm D_2[\bm C_2~ \bm P_2]^\top+[\bm C_2~ \bm R_2]\bm D_2\mathcal{E}_{2,v}^\top+[\bm C_2~ \bm R_2]\mathcal{E}_{2,D}[\bm C_2~ \bm P_2]+[\bm C_2~ \bm R_2]\mathcal{E}_{2,D}\mathcal{E}_{2,v}^\top\\
            +\mathcal{E}_{2,u}\bm D_2[\bm C_2~ \bm P_2]^\top+\mathcal{E}_{2,u}\bm D_2\mathcal{E}_{2,v}^\top+\mathcal{E}_{2,u}\mathcal{E}_{2,D}[\bm C_2~ \bm P_2]^\top+\mathcal{E}_{2,u}\mathcal{E}_{2,D}\mathcal{E}_{2,v}^\top
        \end{pmatrix}\\
        &\otimes \begin{pmatrix}
            [\bm C_1~ \bm R_1]\bm D_1[\bm C_1~ \bm P_1]^\top+[\bm C_1~ \bm R_1]\bm D_1\mathcal{E}_{1,v}^\top+[\bm C_1~ \bm R_1]\mathcal{E}_{1,D}[\bm C_1~ \bm P_1]+[\bm C_1~ \bm R_1]\mathcal{E}_{1,D}\mathcal{E}_{1,v}^\top\\
            +\mathcal{E}_{1,u}\bm D_1[\bm C_1~ \bm P_1]^\top+\mathcal{E}_{1,u}\bm D_1\mathcal{E}_{1,v}^\top+\mathcal{E}_{1,u}\mathcal{E}_{1,D}[\bm C_1~ \bm P_1]^\top+\mathcal{E}_{1,u}\mathcal{E}_{1,D}\mathcal{E}_{1,v}^\top
        \end{pmatrix}\\
        =&\bm A_2\otimes \bm A_1+\bm A_2\otimes \left([\bm C_1~ \bm R_1]\bm D_1\mathcal{E}_{1,v}^\top+[\bm C_1~ \bm R_1]\mathcal{E}_{1,D}[\bm C_1~ \bm P_1]^\top+\mathcal{E}_{1,u}\bm D_1[\bm C_1~ \bm P_1]^\top\right)\\
        &+\left([\bm C_2~ \bm R_2]\bm D_2\mathcal{E}_{2,v}^\top+[\bm C_2~ \bm R_2]\mathcal{E}_{2,D}[\bm C_2~ \bm P_2]^\top+\mathcal{E}_{2,u}\bm D_2[\bm C_2~ \bm P_2]^\top\right)\otimes\bm A_1\\
        &+(\text{57 terms containing 2 or more } \mathcal{E}\text{s})\\
        =&\bm A_2\otimes \bm A_1-\bm H_1+(\text{57 terms containing 2 or more } \mathcal{E}\text{s}).
    \end{aligned}
    \end{equation}
    Therefore, $\bm H$ is the summation of 57 terms of higher order perturbation. Next, we upper bound $\fnorm{\bm H}$ by upper bounding every piece of $\bm H$. 
    $$
    \fnorm{\bm H}\leq \sum_{i=1}^{57}\fnorm{\text{the }i \text{th term of }\bm H}.
    $$
    It can be easily verified that for every $\mathcal{E}$ defined above (we just ignore the subscripts), $\fnorm{\mathcal{E}}^2\leq D$. Writing out everyone of the 57 terms we find that the Frobenius of every term can be upper bounded by one of values in $\{CD^3$, $CbD^{5/2}$, $C\phi b^{-2}D^{5/2}$, $Cb^2D^2$, $C\phi b^{-1}D^2$, $C\phi^2b^{-4}D^2$, $Cb^2D^2$, $C\phi^2b^{-3}D^{3/2}$, $C\phi D^{3/2}$, $Cb^3D^{3/2}$, $C\phi bD$, $C\phi^2b^{-2}D$, $Cb^4D\}$. Using $b=\phi^{1/3}$ and $D\leq C_D\phi^{2/3}$, it is clear that all those values are less than or equal to $C\phi^{4/3}D$, where $C$ is another constant. Finally, adding together the 57 upper bounds gives
    $$\fnorm{\bm H}\leq C_h\phi^{4/3}D.$$
\end{proof}

The following lemma is adapted from Lemma A.2 in \citep{wang2023commonFactor}, with the spectral norm replaced by the Frobenius norm. The proof follows analogously and is omitted here. It establishes the approximate equivalence between the estimation error of components and that of the aggregated matrix $\bm B$.
\begin{lemma}\label{lemma:lemmaA2inWang2023}
    \citep{wang2023commonFactor} Suppose that $\mathbf{B}^*=\left[\mathbf{C}^*~ \mathbf{R}^*\right] \mathbf{D}^*\left[\mathbf{C}^*~ \mathbf{P}^*\right]^{\top}$, $\left[\mathbf{C}^*~ \mathbf{R}^*\right]^{\top}\left[\mathbf{C}^*~ \mathbf{R}^*\right]=b^2\mathbf{I}_{r_1}$, $\left[\mathbf{C}^*~ \mathbf{P}^*\right]^{\top}\left[\mathbf{C}^*~ \mathbf{P}^*\right]=$ $b^2\mathbf{I}_{r_2}$, $\phi=\left\|\mathbf{B}^*\right\|_{\mathrm{F}}$, and $\sigma_r=\sigma_r\left(\mathbf{B}^*\right)$. Let $\mathbf{B}=[\mathbf{C}~ \mathbf{R}] \mathbf{D}[\mathbf{C}~ \mathbf{P}]^{\top}$ with $\|[\mathbf{C}~ \mathbf{R}]\|_{\mathrm{F}} \leq\left(1+c\right) b$, $\|[\mathbf{C}~ \mathbf{P}]\|_{\mathrm{F}} \leq\left(1+c\right) b$, and $\|\mathbf{D}\|_{\mathrm{F}} \leq\left(1+c\right) \phi / b^2$ for some constant $c>0$. Define
    $$
    \begin{aligned}
    E:= & \min _{\substack{\mathbf{O}_c \in \mathbb{O}^{d \times d} \\
    \mathbf{O}_r, \mathbf{O}_p \in \mathbb{O}^{(r-d) \times(r-d)}}}\left(\left\|\left[\mathbf{C}~ \mathbf{R}\right]-\left[\mathbf{C}^*~ \mathbf{R}^*\right] \operatorname{diag}\left(\mathbf{O}_c, \mathbf{O}_r\right)\right\|_{\mathrm{F}}^2+\|\left[\mathbf{C}~ \mathbf{P}\right]\right. \\
    & \left.-\left[\mathbf{C}^*~ \mathbf{P}^*\right] \operatorname{diag}\left(\mathbf{O}_c, \mathbf{O}_p\right)\left\|_{\mathrm{F}}^2+\right\| \mathbf{D}-\operatorname{diag}\left(\mathbf{O}_c, \mathbf{O}_r\right)^{\top} \mathbf{D}^* \operatorname{diag}\left(\mathbf{O}_c, \mathbf{O}_p\right) \|_{\mathrm{F}}^2\right).
    \end{aligned}
    $$
    Then, we have
    $$
    \begin{aligned}
    & E \leq\left(4 b^{-4}+\frac{8 b^2}{\sigma_r^2} C_b\right)\left\|\mathbf{B}-\mathbf{B}^*\right\|_{\mathrm{F}}^2 \\
    & +2 b^{-2} C_b\left(\left\|[\mathbf{C ~ R}]^{\top}[\mathbf{C ~ R}]-b^2 \mathbf{I}_r\right\|_{\mathrm{F}}^2+\left\|[\mathbf{C ~ P}]^{\top}[\mathbf{C ~ P}]-b^2 \mathbf{I}_r\right\|_{\mathrm{F}}^2\right)
    \end{aligned}
    $$
    and
    $$
    \fnorm{\mathbf{B}-\mathbf{B}^*}^2 \leq 3 b^4\left[1+4 \phi^2 b^{-6}\left(1+c\right)^4\right] E,
    $$
    where $C_b=1+4 \phi^2 b^{-6}\left(\left(1+c\right)^4+\left(1+c\right)^2\left(2+c\right)^2 / 2\right)$.
\end{lemma}

Applying Lemma \ref{lemma:lemmaA2inWang2023} on $\bm A_1$ and $\bm A_2$, we have the following equivalence relationship between the combined distance $D$ defined in Lemma \ref{lemma:H_upperBound} and the squared errors of $\bm A_1$ and $\bm A_2$.

\begin{lemma}\label{lemma:A1A2 and Error}
Consider $\bm A_1^*, \bm A_1, \bm A_2^*, \bm A_2$ and $D$ defined as in Lemma \ref{lemma:H_upperBound}, and \eqref{LemmaA1Condition} holds. Define $\underline{\sigma}:=\min(\sigma_{1,r_1},\sigma_{2,r_2}).$ Then, we have
\begin{equation}
\begin{aligned}
    & D \leq \left(4 b^{-4}+\frac{8 b^2}{\underline{\sigma}^2} C_b\right)\left(\left\|\mathbf{A}_1-\mathbf{A}_1^*\right\|_{\mathrm{F}}^2+\left\|\mathbf{A}_2-\mathbf{A}_2^*\right\|_{\mathrm{F}}^2\right) \\
    & +2C_bb^{-2}\sum_{i=1}^{2}\left(\fnorm{\left[\bm C_i~ \bm R_i\right]^{\top}\left[\bm C_i~ \bm R_i\right]-b^2 \bm I_{r_i}}^2+\fnorm{\left[\bm C_i~ \bm P_i\right]^{\top}\left[\bm C_i~ \bm P_i\right]-b^2 \bm I_{r_i}}^2\right).
\end{aligned}
\end{equation}
In addition, 
$$
    \fnorm{\bm A_1-\bm A_1^*}^2+\fnorm{\bm A_2-\bm A_2^*}^2\leq 6b^4D+24\phi^2b^{-2}(1+c)^4D.
$$
\end{lemma}

\begin{proof}[Proof of Lemma \ref{lemma:A1A2 and Error}]
    The conditions specified in \eqref{LemmaA1Condition} satisfy the premises of Lemma \ref{lemma:lemmaA2inWang2023}. Applying Lemma \ref{lemma:lemmaA2inWang2023} to the pairs $(\bm A_1, \bm A_1^*)$ and $(\bm A_2, \bm A_2^*)$ separately yields the desired bounds. Specifically, summing the resulting inequalities for $i=1$ and $i=2$ establishes the first inequality. The second inequality follows directly by applying the reverse bound from Lemma \ref{lemma:lemmaA2inWang2023} and summing over $i=1, 2$. This concludes the proof.
\end{proof}

The last lemma is useful in analyzing the upper bound of the distance between our estimates and their true values.
\begin{lemma}\label{lemma: Tu16, dist(U,X) and UU'-XX'}
(Lemma 5.4., \cite{Tu2016}) For any $\bm U, \bm X \in \mathbb{R}^{p \times r}$, where $p\geq r$ and $\bm X$ is full-rank, define 
$$\mathrm{dist}(\bm U,\bm X)^2=\min_{\bm O\in\mathbb{O}^{r\times r}}\fnorm{\bm U-\bm X\bm O}^2,$$ 
we have 
    $$\operatorname{dist}(\bm U, \bm X)^2 \leq \frac{1}{2(\sqrt{2}-1) \sigma_r^2(\bm X)}\left\|\bm U \bm U^{\top}-\bm X \bm X^{\top}\right\|_\textup{F}^2.$$
\end{lemma}

\newpage
\section{Statistical Convergence Analysis}\label{sec:statistical_convergence}

In this appendix, we present the stochastic analysis underlying our statistical theory, including verification of the RSC/RSS condition and bounding of the deviation term $\xi$ and initialization errors. The conclusions are derived under assumptions presented in Section \ref{sec: stat theory} of the main article. Notations are inherited from Appendix \ref{sec: computation converge}.

\subsection{Proof of Theorem \ref{theorem: statistical error}}

Under Assumptions \ref{assumption: spectral redius}, \ref{assumption: sub-Gaussian noise}, and \ref{assumption: gmin}, by Proposition \ref{prop: RSCRSS}, we have that with probability at least $1-2\exp\left\{-C(p_1r_1+p_2r_2+4r_1^2r_2^2)\right\}$, the RSC/RSS condition holds with $\alpha=\alpha_\mathrm{RSC}$ and $\beta=\beta_\mathrm{RSS}$, whose proof is relegated to Appendix \ref{sub:RSCRSS}. Meanwhile, by Proposition \ref{proposition: small kappa} and Proposition \ref{prop: sample_size_initialization}, the conditions in Theorem \ref{theorem: computational convergence} related to $\xi$ and $\mathrm{dist}^2_{(0)}$ are satisfied with high probability, whose proofs are relegated to Appendices \ref{sub:UpperBoundXi} and \ref{sub:Initialization}, respectively.

Therefore, by the conditions of Theorem \ref{theorem: statistical error}, the conditions of Theorem \ref{theorem: computational convergence} are satisfied, implying that $\forall j\geq 1$,
$$
    \begin{aligned}
    &\fnorm{\bm A_1^{(j)}-\bm A_1^*}^2+\fnorm{\bm A_2^{(j)}-\bm A_2^*}^2\\
    \lesssim &\kappa^2(1-C_0\eta_0\alpha_\mathrm{RSC}\beta_\mathrm{RSS}^{-1}\kappa^{-2})^j\left(\fnorm{\bm A_1^{(0)}-\bm A_1^*}^2+\fnorm{\bm A_2^{(0)}-\bm A_2^*}^2\right)+\kappa^4\phi^{-2}\alpha_\mathrm{RSC}^{-2}\xi^2(r_1,r_2,d_1,d_2).
\end{aligned}$$
For the initial error, we apply Proposition \ref{prop: initialization} to see that with probability at least $1-C\exp\left(-C(p_1r_1+p_2r_2)\right)$,
$$
\fnorm{\bm A_1^{(0)}-\bm A_1^*}^2+\fnorm{\bm A_2^{(0)}-\bm A_2^*}^2\lesssim \phi^{4/3}\mathrm{dist}^2_{(0)}\lesssim \phi^2 \underline{\sigma}^{-4} g_{\min }^{-4}\alpha_\mathrm{RSC}^{-2}\tau^2 M_1^2\frac{\mathrm{df}_\mathrm{RRMAR}}{T}.
$$
In addition, by Proposition \ref{prop: Upper bound of xi}, the statistical error is bounded as
$$\xi(r_1,r_2,d_1,d_2)\lesssim \tau M_1\sqrt{\frac{\mathrm{df}_\mathrm{MARCF}}{T}},$$
with probability at least $1-C\exp\left(-C(p_1r_1+p_2r_2)\right)$.

To ensure the upper bound of the statistical error term dominates the computational error term after $j$ iterations, we need 
\begin{equation}\label{eq:iteration_numbers}
    \kappa^2(1-C_0\eta_0\alpha_\mathrm{RSC}\beta_\mathrm{RSS}^{-1}\kappa^{-2})^j\phi^2 \underline{\sigma}^{-4} g_{\min }^{-4}\alpha_\mathrm{RSC}^{-2}\tau^2 M_1^2\frac{\mathrm{df}_\mathrm{RRMAR}}{T}\lesssim \kappa^4\alpha_\mathrm{RSC}^{-2}\phi^{-2}\tau^2 M_1^2\frac{\mathrm{df}_\mathrm{MARCF}}{T}.
\end{equation}
With $\mathrm{df}_\mathrm{RRMAR}\leq 2\mathrm{df}_\mathrm{MARCF}$, \eqref{eq:iteration_numbers} gives that when
$$
    J\gtrsim\frac{\log\left(\kappa^{-2}g_{\min}^{4}\right)}{\log(1-C_0\eta_0\alpha_\mathrm{RSC}\beta_\mathrm{RSS}^{-1}\kappa^{-2})},
$$
the computational error is absorbed, and then
$$
    \fnorm{\bm A_1^{(J)}-\bm A_1^*}^2+\fnorm{\bm A_2^{(J)}-\bm A_2^*}^2\lesssim \kappa^4\phi^{-2}\alpha_\mathrm{RSC}^{-2}\tau^2M_1^2\frac{\mathrm{df}_\mathrm{MARCF}}{T}.
$$
By similar analysis on $\fnorm{\bm A_2^{(J)}\otimes \bm A_1^{(J)}-\bm A_2^*\otimes \bm A_1^*}^2$, we also have
$$
    \fnorm{\bm A_2^{(J)}\otimes \bm A_1^{(J)}-\bm A_2^*\otimes \bm A_1^*}^2\lesssim \kappa^4\alpha_\mathrm{RSC}^{-2}\tau^2M_1^2\frac{\mathrm{df}_\mathrm{MARCF}}{T}.
$$
\qed

\subsection{Verification of RSC/RSS Condition}\label{sub:RSCRSS}

The RSC and RSS properties are defined with respect to the identifiable parameter matrix $\bm A=\bm A_2\otimes \bm A_1$. The vectorized MARCF model is $\bm y_t=(\bm A_2\otimes \bm A_1)\bm y_{t-1}+\bm e_t$, where $\bm y_t:=\V{\bm Y_t}$ and $\bm e_t:=\V{\bm E_t}$. Then, the least-squares loss function with respect to $\bm A$ can be written as 
\begin{equation}
    \widetilde{\mathcal{L}}(\bm{A})=\frac{1}{2T}\sum_{t=1}^T\left\|\bm y_t-\bm{A}\bm y_{t-1}\right\|_2^2.
\end{equation}
We verify in the following proposition that the RSC/RSS condition holds for any matrices satisfying the restrictions of the MARCF model structure.
\begin{proposition}\label{prop: RSCRSS}
    Suppose Assumptions \ref{assumption: spectral redius} and \ref{assumption: sub-Gaussian noise} hold. Let $\bm A=\bm A_2\otimes \bm A_1$ and $\bm A'=\bm A_2'\otimes \bm A_1'$ be any two matrices satisfying the MARCF structural decomposition, where $\bm A_i=[\bm C_i~\bm R_i]\bm D_i[\bm C_i~\bm P_i]^\top$ and $\bm A_i'=[\bm C_i'~\bm R_i']\bm D_i'[\bm C_i'~\bm P_i']^\top$ with dimensions consistent with the model definition. If the sample size $T\gtrsim (p_1r_1+p_2r_2+4r_1^2r_2^2)M_2^{-2}\max\left(\tau,\tau^2\right)$, then, with probability at least $1-2\exp\left\{-C(p_1r_1+p_2r_2+4r_1^2r_2^2)\right\}$, for such $\bm A$ and $\bm A'$, the loss function $\widetilde{\mathcal{L}}(\cdot)$ satisfies the RSC/RSS condition
    \begin{equation}
        \frac{\alpha_\mathrm{RSC}}{2}\left\|\mathbf{A}-\mathbf{A}'\right\|_{\mathrm{F}}^2 \leqslant \widetilde{\mathcal{L}}(\bm{A})-\widetilde{\mathcal{L}}\left(\mathbf{A}'\right)-\left\langle\nabla \widetilde{\mathcal{L}}\left(\mathbf{A}'\right), \mathbf{A}-\mathbf{A}'\right\rangle \leqslant \frac{\beta_\mathrm{RSS}}{2}\left\|\mathbf{A}-\mathbf{A}'\right\|_{\mathrm{F}}^2
    \end{equation}
    where
    \begin{equation}
        \alpha_\mathrm{RSC}=\lambda_{\min }\left(\bbm{\Sigma}_{\bm e}\right)/(2\mu_{\max }(\mathcal{A})),\quad \beta_\mathrm{RSS}=3\lambda_{\max }\left(\bbm{\Sigma}_{\bm e}\right)/(2\mu_{\min }(\mathcal{A})),
    \end{equation}
    and
    \begin{equation}
        M_2=\lambda_{\min }\left(\bbm{\Sigma}_{\bm e}\right) \mu_{\min }\left(\mathcal{A}\right)/\left(\lambda_{\max}\left(\bbm{\Sigma}_{\bm e}\right) \mu_{\max }\left(\mathcal{A}\right)\right).
    \end{equation}
\end{proposition}

\begin{proof}[Proof of Proposition \ref{prop: RSCRSS}]
Straightforward algebraic manipulation yields
\begin{equation}
    \widetilde{\mathcal{L}}(\bm{A})-\widetilde{\mathcal{L}}\left(\mathbf{A}'\right)-\left\langle\nabla \widetilde{\mathcal{L}}\left(\mathbf{A}'\right), \mathbf{A}-\mathbf{A}'\right\rangle =\frac{1}{2T}\sum_{t=0}^{T-1}\fnorm{(\bm A-\bm A')\bm y_t}^2.
\end{equation}
Hence, it suffices to prove that
$$
    \frac{\alpha_\mathrm{RSC}}{2}\fnorm{\bm A-\bm A'}^2\leq \frac{1}{2T}\sum_{t=0}^{T-1}\fnorm{(\bm A-\bm A')\bm y_t}^2\leq\frac{\beta_\mathrm{RSS}}{2}\fnorm{\bm A-\bm A'}^2
$$
with high probability.

Based on the moving average representation of $\operatorname{VAR}(1)$, we can rewrite $\bm{y}_t$ as a $\operatorname{VMA}(\infty)$ process,
$$
\bm{y}_t=\bm e_t+\bm{A} \bm e_{t-1}+\bm{A}^2 \bm e_{t-2}+\bm{A}^3 \bm e_{t-3}+\cdots
$$
Let $\bm{z}=\left(\bm{y}_{T-1}^{\top}, \bm{y}_{T-2}^{\top}, \ldots, \bm{y}_0^{\top}\right)^{\top}, \bm e=\left(\bm e_{T-1}^{\top}, \bm e_{T-2}^{\top}, \ldots, \bm e_0^\top, \ldots\right)^{\top}$. Now $\bm{z}=\widetilde{\bm{A}} \bm e$, where $\widetilde{\bm{A}}$ is a matrix with $Tp_1p_2$ rows and $\infty$ columns, defined as
$$
\widetilde{\bm{A}}=\left[\begin{array}{ccccc}
\bm{I}_{p_1p_2} & \bm{A} & \bm{A}^2 & \bm{A}^3 & \ldots \\
\bm{O} & \bm{I}_{p_1p_2} & \bm{A} & \bm{A}^2 & \ldots \\
\vdots & \vdots & \ddots & \vdots & \vdots\\
\bm{O} & \bm{O} & \bm{O} & \bm{I}_{p_1p_2} & \ldots
\end{array}\right]
=\left[\begin{array}{ccccc}
    \bm{I}_{p_1p_2} & \bm A_2\otimes \bm A_1 & (\bm A_2\otimes \bm A_1)^2 & (\bm A_2\otimes \bm A_1)^3 & \ldots \\
    \bm{O} & \bm{I}_{p_1p_2} & \bm A_2\otimes \bm A_1 & (\bm A_2\otimes \bm A_1)^2 & \ldots \\
    \vdots & \vdots & \ddots & \vdots & \vdots\\
    \bm{O} & \bm{O} & \bm{O} & \bm{I}_{p_1p_2} & \ldots
    \end{array}\right].
$$
Let $\boldsymbol{\zeta}=\left(\boldsymbol{\zeta}_{T-1}^{\top}, \boldsymbol{\zeta}_{T-2}^{\top}, \ldots, \boldsymbol{\zeta}_0^\top, \ldots\right)$. By assumption on noise we know that $\bm e=\widetilde{\bbm \Sigma}^{1/2}_{\bm e}\boldsymbol{\zeta}$ where
$$
\widetilde{\bbm \Sigma}_{\bm e}^{1/2}=
    \begin{pmatrix}
        \bbm \Sigma_{\bm e}^{1/2} & \bm O & \bm O & \cdots\\
        \bm O & \bbm \Sigma_{\bm e}^{1/2} & \bm O & \cdots\\
        \bm O & \bm O & \bbm \Sigma_{\bm e}^{1/2} & \cdots\\
        \vdots & \vdots & \vdots & \ddots
    \end{pmatrix}.
$$

Let $\bm \Delta=\bm A-\bm A'$. It suffices to show that for $\fnorm{\bbm\Delta}=1$, there exist $0<\alpha\leq \beta$, such that 
$$\frac{\alpha}{2}\leq \frac{1}{2T}\sum_{t=0}^{T-1}\fnorm{\bbm\Delta\bm y_t}^2\leq\frac{\beta}{2}.$$
We define the set of unit-norm differences of low-rank Kronecker products as as:
\begin{equation}\label{eq:KP_difference_set}
    \begin{aligned}
        \mathcal{S}(r_1,r_2,p_1,p_2) := \Big\{\bm M:~ &\bm M=\bm A_2\otimes \bm A_1-\bm A_2'\otimes \bm A_1',\\
        &\bm A_i, \bm A_i'\in \mathbb{R}^{p_i\times p_i},~\text{rank}(\bm A_i),~\text{rank}(\bm A_i')\leq r_i, ~ \fnorm{\bm M}=1, ~ i=1,2 \Big\}.
    \end{aligned}
\end{equation}
For brevity, we denote this set as $\mathcal{S}$ when the dimensions and ranks are clear from the context. By construction, $\bm\Delta \in \mathcal{S}$. Denoting $R_T(\bm \Delta):=\sum_{t=0}^{T-1}\norm{\bm \Delta\bm y_t}^2_2$, we have
\begin{equation}\label{R_T}
\begin{aligned}
R_T(\boldsymbol{\Delta})&=\sum_{t=0}^{T-1}\bm y_t^\top\bm \Delta^\top\bm \Delta\bm y_t\\
&=\left(\bm{y}_{T-1}^{\top}, \bm{y}_{T-2}^{\top}, \ldots, \bm{y}_0^{\top}\right)\begin{pmatrix}
    \bm \Delta^\top\bm \Delta& & \\
    & \ddots &\\
    & & \bm \Delta^\top\bm \Delta
\end{pmatrix}
\left(\begin{array}{c}
    \bm y_{T-1}\\
    \bm y_{T-2}\\
    \vdots\\
    \bm y_0
\end{array}\right)\\
&=\bm{z}^{\top}\left(\bm{I}_T \otimes \boldsymbol{\Delta}^{\top} \boldsymbol{\Delta}\right) \bm{z}\\
&=\bm e^{\top} \widetilde{\bm{A}}^{\top}\left(\bm{I}_T \otimes \boldsymbol{\Delta}^{\top} \boldsymbol{\Delta}\right) \widetilde{\bm{A}} \bm e \\
&= \boldsymbol{\zeta}^{\top} \widetilde{\boldsymbol{\Sigma}}_{\bm e}^{1/2} \widetilde{\bm{A}}^{\top}\left(\bm{I}_T \otimes \boldsymbol{\Delta}^{\top} \boldsymbol{\Delta}\right) \widetilde{\bm{A}} \widetilde{\boldsymbol{\Sigma}}_{\bm e}^{1/2} \boldsymbol{\zeta}\\
&:=\boldsymbol{\zeta}^{\top} \boldsymbol{\Sigma}_{\boldsymbol{\Delta}} \boldsymbol{\zeta}.
\end{aligned}
\end{equation}

Note that $R_T(\bbm \Delta)\geq \mathbb{E}R_T(\bbm \Delta)-\sup\limits_{\bbm \Delta\in\mathcal{S}}|R_T(\bbm \Delta)-\mathbb{E}R_T(\bbm \Delta)|$, we will derive a lower bound for $\mathbb{E}R_T(\bbm \Delta)$ and an upper bound for $\sup\limits_{\bbm \Delta\in\mathcal{S}}|R_T(\bbm \Delta)-\mathbb{E}R_T(\bbm \Delta)|$ to complete the proof of RSC part.

For $\mathbb{E}R_T(\bbm \Delta)$, by \eqref{R_T} and properties of Frobenius norm,
$$\begin{aligned}
    \mathbb{E}R_T(\bbm \Delta)&=\mathbb{E}\left[\boldsymbol{\zeta}^{\top} \widetilde{\boldsymbol{\Sigma}}_{\bm e}^{1/2} \widetilde{\bm{A}}^{\top}\left(\bm{I}_T \otimes \boldsymbol{\Delta}^{\top} \boldsymbol{\Delta}\right) \widetilde{\bm{A}} \widetilde{\boldsymbol{\Sigma}}_{\bm e}^{1/2} \boldsymbol{\zeta}\right]\\
    &=\mathrm{tr}\left(\widetilde{\boldsymbol{\Sigma}}_{\bm e}^{1/2} \widetilde{\bm{A}}^{\top}\left(\bm{I}_T \otimes \boldsymbol{\Delta}^{\top} \boldsymbol{\Delta}\right) \widetilde{\bm{A}} \widetilde{\boldsymbol{\Sigma}}_{\bm e}^{1/2}\right)\\
    &=\fnorm{(\bm I_T\otimes \bbm\Delta)\widetilde{\bm{A}} \widetilde{\bbm\Sigma}_{\bm e}^{1/2}}^2\\
    &\geq T\lambda_{\min}(\widetilde{\bm A}\widetilde{\bm A}^\top)\lambda_{\min}(\bbm \Sigma_{\bm e}).
\end{aligned}$$

For $|R_T(\bbm \Delta)-\mathbb{E}R_T(\bbm \Delta)|$, noting that
\begin{equation}
    \left\|\bbm{\Sigma}_{\bbm{\Delta}}\right\|_{\mathrm{F}}^2 \leq \fnorm{\bm I_T\otimes \bbm\Delta}^2\opnorm{\bm I_T\otimes \bbm\Delta}^2\|\widetilde{\bbm\Sigma}_{\bm e}^{1/2}\|_\mathrm{op}^4\|\widetilde{\bm A}\|_\mathrm{op}^4\leq  T\lambda_{\max}\left(\bbm\Sigma_{\bm e}\right)\lambda_{\max}\left(\widetilde{\bm A}\widetilde{\bm A}^\top\right)
\end{equation}
and
$$
\left\|\bbm{\Sigma}_{\bbm{\Delta}}\right\|_{\text {op }} \leq \opnorm{\bm I_T\otimes \bbm\Delta}^2\|\widetilde{\bbm\Sigma}_{\bm e}^{1/2}\|_\mathrm{op}^2\|\widetilde{\bm A}\|_\mathrm{op}^2\leq \lambda_{\max}^2\left(\bbm\Sigma_{\bm e}\right)\lambda_{\max}^2\left(\widetilde{\bm A}\widetilde{\bm A}^\top\right),
$$
by Hanson-Wright inequality (Lemma \ref{Hanson-Wright}), we have that for any $x\geq 0$,
\begin{equation}\label{RSCineq_Hanson-Wright}\begin{aligned}
&\mathbb{P}\left(|R_T(\bbm \Delta)-\mathbb{E}R_T(\bbm \Delta)|\geq \tau x\right)\\
\leq& 2\exp\left\{-C\min\left(\frac{x}{\lambda_{\max }\left(\bbm{\Sigma}_{\bm e}\right) \lambda_{\max }\left(\widetilde{\mathbf{A}} \widetilde{\mathbf{A}}^{\top}\right)},\frac{x^2}{T\lambda_{\max }^2\left(\bbm{\Sigma}_{\bm e}\right) \lambda_{\max }^2\left(\widetilde{\mathbf{A}} \widetilde{\mathbf{A}}^{\top}\right)}\right)\right\},
\end{aligned}\end{equation}
where $C$ is a constant.

For any $\varepsilon\in(0,1)$, consider an $\varepsilon$-net of $\mathcal{S}$ with respect to the Frobenius norm, denoted by $\overline{\mathcal{S}}$. Then for any $\bbm \Delta\in\mathcal{S}$, there exists $\overline{\bbm \Delta}\in\overline{\mathcal{S}}$, such that $\fnorm{\bbm \Delta-\overline{\bbm \Delta}}\leq \varepsilon$. By Lemma \ref{lemma:covering_KP_difference}, we have 
\begin{equation}
    |\overline{\mathcal{S}}|\leq \left(\frac{15}{\varepsilon}\right)^{4p_1r_1+4p_2r_2+16r_1^2r_2^2}.
\end{equation}
Therefore, 
$$\begin{aligned}
    &|R_T(\bbm \Delta)-\mathbb{E}R_T(\bbm \Delta)|\\
    =&\left|\sum_{t=0}^{T}\left(\fnorm{\bm y_t^\top\otimes\bm I_{p_1p_2}\V{\bbm \Delta}}^2-\mathbb{E}\fnorm{\bm y_t^\top\otimes\bm I_{p_1p_2}\V{\bbm \Delta}}\right)\right|\\
    =&\left|\V{\bbm \Delta}^\top\sum_{t=0}^{T}\left(\bm y_t\bm y_t^\top\otimes\bm I_{p_1p_2}-\mathbb{E}\bm y_t\bm y_t^\top\otimes\bm I_{p_1p_2}\right)\V{\bbm \Delta}\right|\\
    :=&\left|\V{\bbm \Delta}^\top\bm M(\bm y_t)\V{\bbm \Delta}\right|\\
    \leq &\left|\V{\overline{\bbm \Delta}}^\top\bm M(\bm y_t)\V{\overline{\bbm \Delta}}\right|+2\left|\V{\overline{\bbm \Delta}}^\top\bm M(\bm y_t)(\V{\bbm \Delta}-\V{\overline{\bbm \Delta}})\right|\\
    &+\left|(\V{\bbm \Delta}^\top-\V{\overline{\bbm \Delta}}^\top)\bm M(\bm y_t)(\V{\bbm \Delta}-\V{\overline{\bbm \Delta}})\right|\\
    \leq &\max_{\overline{\bbm \Delta}\in\overline{\mathcal{S}}}|R_T(\overline{\bbm \Delta})-\mathbb{E}R_T(\overline{\bbm \Delta})|+(2\varepsilon+\varepsilon^2)\sup_{\bbm \Delta\in \mathcal{S}}|R_T(\bbm \Delta)-\mathbb{E}R_T(\bbm \Delta)|.
\end{aligned}$$ 
When $\varepsilon< 1$, we have
$$ \sup_{\bbm \Delta\in \mathcal{S}}|R_T(\bbm \Delta)-\mathbb{E}R_T(\bbm \Delta)|\leq (1-3\varepsilon)^{-1}\max_{\overline{\bbm \Delta}\in\overline{\mathcal{S}}}|R_T(\overline{\bbm \Delta})-\mathbb{E}R_T(\overline{\bbm \Delta})|.$$
Letting $\varepsilon=0.1$, by union bounds and \eqref{RSCineq_Hanson-Wright},
$$\begin{aligned}
    &\mathbb{P}\left(\sup_{\bbm \Delta\in \mathcal{S}}|R_T(\bbm \Delta)-\mathbb{E}R_T(\bbm \Delta)|\geq \tau x\right)\\
    \leq& \mathbb{P}\left(\max_{\overline{\bbm \Delta}\in\overline{\mathcal{S}}}|R_T(\overline{\bbm \Delta})-\mathbb{E}R_T(\overline{\bbm \Delta})|\geq (1-3\varepsilon)\tau x\right)\\
    \leq &\sum_{\overline{\bbm \Delta}\in\overline{\mathcal{S}}}\mathbb{P}\left(|R_T(\overline{\bbm \Delta})-\mathbb{E}R_T(\overline{\bbm \Delta})|\geq \frac{1}{2}\tau x\right)\\
    \leq& 2\times 150^{4p_1r_1+4p_2r_2+16r_1^2r_2^2}\exp\left\{-C\min\left(\frac{x}{\lambda_{\max }\left(\bbm{\Sigma}_{\bm e}\right) \lambda_{\max }\left(\widetilde{\mathbf{A}} \widetilde{\mathbf{A}}^{\top}\right)},\frac{x^2}{T\lambda_{\max }^2\left(\bbm{\Sigma}_{\bm e}\right) \lambda_{\max }^2\left(\widetilde{\mathbf{A}} \widetilde{\mathbf{A}}^{\top}\right)}\right)\right\}\\
    \leq &2\exp\left\{21(p_1r_1+p_2r_2+4r_1^2r_2^2)-C\min\left(\frac{x}{\lambda_{\max }\left(\bbm{\Sigma}_{\bm e}\right) \lambda_{\max }\left(\widetilde{\mathbf{A}} \widetilde{\mathbf{A}}^{\top}\right)},\frac{x^2}{T\lambda_{\max }^2\left(\bbm{\Sigma}_{\bm e}\right) \lambda_{\max }^2\left(\widetilde{\mathbf{A}} \widetilde{\mathbf{A}}^{\top}\right)}\right)\right\}.
\end{aligned}$$
Here we take $x=T\lambda_{\min }\left(\bbm{\Sigma}_{\bm e}\right) \lambda_{\min }\left(\widetilde{\mathbf{A}} \widetilde{\mathbf{A}}^{\top}\right)/(2\tau)$ and define
$$M_2:=\lambda_{\min }\left(\bbm{\Sigma}_{\bm e}\right) \lambda_{\min }\left(\widetilde{\mathbf{A}} \widetilde{\mathbf{A}}^{\top}\right)/\lambda_{\max}\left(\bbm{\Sigma}_{\bm e}\right) \lambda_{\max }\left(\widetilde{\mathbf{A}} \widetilde{\mathbf{A}}^{\top}\right)\leq 1,$$
then
\begin{equation}\label{eq: RSC prob}
\begin{aligned}
&\mathbb{P}\left(\sup_{\bbm \Delta\in \mathbb{S}}|R_T(\bbm \Delta)-\mathbb{E}R_T(\bbm \Delta)|\geq \frac{T}{2}\lambda_{\min }\left(\bbm{\Sigma}_{\bm e}\right) \lambda_{\min }\left(\widetilde{\mathbf{A}} \widetilde{\mathbf{A}}^{\top}\right)\right)\\
\leq& 2\exp\left\{21(p_1r_1+p_2r_2+4r_1^2r_2^2)-C\min\left(\frac{M_2}{2\tau},\frac{M^2_2}{4\tau^2}\right)T\right\}\\
\leq& 2\exp\left\{21(p_1r_1+p_2r_2+4r_1^2r_2^2)-CM_2^2\min\left(\tau^{-1},\tau^{-2}\right)T\right\}.
\end{aligned}
\end{equation}

Hence, when $T\gtrsim (p_1r_1+p_2r_2+4r_1^2r_2^2)M_2^{-2}\max\left(\tau,\tau^2\right)$, we have that with probability at least $1-2\exp\left\{-C(p_1r_1+p_2r_2+4r_1^2r_2^2)\right\}$,
$$R_T(\bbm \Delta)\geq \frac{T}{2}\lambda_{\min }\left(\bbm{\Sigma}_{\bm e}\right) \lambda_{\min }\left(\widetilde{\mathbf{A}} \widetilde{\mathbf{A}}^{\top}\right).$$
Then the RSC coefficient $\alpha_\mathrm{RSC}=\lambda_{\min }\left(\bbm{\Sigma}_{\bm e}\right) \lambda_{\min }\left(\widetilde{\mathbf{A}} \widetilde{\mathbf{A}}^{\top}\right)/2.$

As for RSS part, similarly, we have $R_T(\bbm \Delta)\leq \mathbb{E}R_T(\bbm \Delta)+\sup\limits_{\bbm \Delta\in\mathcal{S}}|R_T(\bbm \Delta)-\mathbb{E}R_T(\bbm \Delta)|$ and $\mathbb{E}R_T(\bbm \Delta)\leq T\lambda_{\max}(\widetilde{\bm A}\widetilde{\bm A}^\top)\lambda_{\max}(\bbm \Sigma_{\bm e}).$ Since the event 
$$\left\{\sup\limits_{\bbm \Delta\in\mathcal{S}}|R_T(\bbm \Delta)-\mathbb{E}R_T(\bbm \Delta)|\leq T\lambda_{\min }\left(\bbm{\Sigma}_{\bm e}\right) \lambda_{\min }\left(\widetilde{\mathbf{A}} \widetilde{\mathbf{A}}^{\top}\right)/2\right\}$$
implies 
$$\left\{\sup\limits_{\bbm \Delta\in\mathcal{S}}|R_T(\bbm \Delta)-\mathbb{E}R_T(\bbm \Delta)|\leq T\lambda_{\max }\left(\bbm{\Sigma}_{\bm e}\right) \lambda_{\max }\left(\widetilde{\mathbf{A}} \widetilde{\mathbf{A}}^{\top}\right)/2\right\},$$
we have when the event above occurs,
\begin{equation}\label{upper bound of RT}
R_T(\bbm \Delta)\leq \frac{3T}{2}\lambda_{\max }\left(\bbm{\Sigma}_{\bm e}\right) \lambda_{\max }\left(\widetilde{\mathbf{A}} \widetilde{\mathbf{A}}^{\top}\right).
\end{equation}
Thus $\beta_\mathrm{RSS}:=3\lambda_{\max }\left(\bbm{\Sigma}_{\bm e}\right) \lambda_{\max }\left(\widetilde{\mathbf{A}} \widetilde{\mathbf{A}}^{\top}\right)/2.$

Finally, since $\widetilde{\mathbf{A}}$ is related to $\operatorname{VMA}(\infty)$ process, by the spectral measure of ARMA process discussed in \citet{BasuandGeorgeMichailidis2015} we replace $\lambda_{\max }\left(\widetilde{\mathbf{A}} \widetilde{\mathbf{A}}^{\top}\right)$ and $\lambda_{\min }\left(\widetilde{\mathbf{A}} \widetilde{\mathbf{A}}^{\top}\right)$ with $1 / \mu_{\min }(\mathcal{A})$ and $1 / \mu_{\max }(\mathcal{A})$, respectively.
\end{proof}

\subsection{Property of Deviation Bound}\label{sub:UpperBoundXi}

In this appendix, we derive an upper bound for $\xi$ defined in \ref{def: statistical error xi}, which demonstrates the benefit of dimension reduction conferred by the MARCF model.
\begin{proposition}\label{prop: Upper bound of xi}
    Under Assumptions \ref{assumption: spectral redius} and \ref{assumption: sub-Gaussian noise}, if $T\gtrsim (p_1r_1+p_2r_2+4r_1^2r_2^2)M_2^{-2}\max\left(\tau,\tau^2\right)$, then with probability as least $1-2\exp\left(-C(p_1r_1+p_2r_2)\right)$, 
    $$\xi(r_1,r_2,d_1,d_2)\lesssim \tau M_1\sqrt{\frac{\mathrm{df}_\mathrm{MARCF}}{T}},$$
    where $M_1=\lambda_{\max }\left(\bbm{\Sigma}_{\bm e}\right)/ \mu_{\min }^{1/2}\left(\mathcal{A}\right)=\lambda_{\max }\left(\bbm{\Sigma}_{\bm e}\right) \lambda_{\max }^{1/2}\left(\widetilde{\mathbf{A}} \widetilde{\mathbf{A}}^{\top}\right)$, $\mathrm{df}_\mathrm{MARCF}=p_1(2r_1-d_1)+p_2(2r_2-d_2)+r_1^2+r_2^2$, and $M_2$ is defined in Proposition \ref{prop: RSCRSS}. 
\end{proposition}

\begin{proof}[Proof of Proposition \ref{prop: Upper bound of xi}]

We first define the parameter spaces for $\bm A_1$ and $\bm A_2$ in the MARCF model. Specifically, we consider the following set of matrices with unit Frobenius norm and common column and row subspaces, denoted by $\mathcal{W}(r,d;p)$:
\begin{equation}\label{eq: W space}
\begin{aligned}
    \mathcal{W}(r, d ; p):=\left\{\mathbf{W} \in \mathbb{R}^{p \times p}: \mathbf{W}=[\mathbf{C}~\mathbf{R}] \mathbf{D}[\mathbf{C}~\mathbf{P}]^{\top}, \mathbf{C} \in\mathbb{O}^{p \times d}, \mathbf{R}, \mathbf{P} \in \mathbb{O}^{p \times(r-d)},\right.\\
\inner{\bm C}{\bm R}=\inner{\bm C}{\bm P}=0, \left. \textrm{ and } \fnorm{\bm W}=1\right\}.
\end{aligned}
\end{equation}
Let $\overline{\mathcal{W}}(r,d;p)$ be the $\varepsilon/2$-net of $\mathcal{W}(r,d,p)$ with respect to the Frobenius norm. By Lemma \ref{lemma: common covering}, $\overline{\mathcal{W}}(r,d;p)\subset \mathcal{W}(r,d;p)$, and
$$|\overline{\mathcal{W}}(r,d;p)|\leq \left(\frac{48}{\varepsilon}\right)^{p(2r-d)+r^2}.$$
Define
$$\mathcal{V}(r_1,r_2,d_1,d_2):=\left\{\bm W_2\otimes \bm W_1: \bm W_2\in \mathcal{W}(r_2,d_2;p_2), \bm W_1 \in \mathcal{W}(r_1,d_1;p_1)\right\},$$
then 
$\overline{\mathcal{V}}(r_1,r_2,d_1,d_2):=\left\{\overline{\bm W}_2\otimes \overline{\bm W}_1: \overline{\bm W}_2\in \overline{\mathcal{W}}(r_2,d_2;p_2), \overline{\bm W}_1 \in \overline{\mathcal{W}}(r_1,d_1;p_1)\right\}$ is an $\varepsilon$-net of $\mathcal{V}(r_1,r_2,d_1,d_2)$, which can be directly verified. Meanwhile,
\begin{equation}\label{ineq: V covering}
|\overline{\mathcal{V}}(r_1,r_2,d_1,d_2)|\leq \left(\frac{48}{\varepsilon}\right)^{2(p_1r_1+p_2r_2)-p_1d_1-p_2d_2+r_1^2+r_2^2}.
\end{equation}
In addition, note that $\nabla\mathcal{L}(\bm A^*)= (\sum_{t=1}^T \bm e_t \bm y_{t-1}^{\top})/T$, we have
$$\xi(r_1,r_2,d_1,d_2)=\sup_{\bm A\in \mathcal{V}(r_1,r_2,d_1,d_2)}\inner{\frac{1}{T}\sum_{t=1}^T\bm e_t \bm y_{t-1}^{\top}}{\bm A}.$$

In the following, we start from establishing an upper bound for $\xi$ when there is no common space, i.e., $d_1=d_2=0$. In this case, elements of $\bm W_1$ and $\bm W_2$ are merely low rank matrices.

For every $\bm A=\bm W_2\otimes \bm W_1\in\mathcal{V}(r_1,r_2,0,0)$, let $\overline{\bm A}=\overline{\bm W}_2\otimes \overline{\bm W}_1\in\overline{\mathcal{V}}(r_1,r_2,0,0)$ be its covering matrix. By splitting SVD with common space (Lemma \ref{lemma:splitting SVD}), we know that $\bm W_2-\overline{\bm W}_2$ can be decomposed as $\bm W_2-\overline{\bm W}_2=\bbm \Delta_{2,1}+\bbm \Delta_{2,2}$, where $\bbm \Delta_{2,1}, \bbm \Delta_{2,2}$ are both rank-$r_2$ and $\inner{\bbm \Delta_{2,1}}{\bbm \Delta_{2,2}}=0$. For $\bm W_1-\overline{\bm W}_1$, there exist $\bbm \Delta_{1,1}$ and $\bbm \Delta_{1,2}$ following the same property. Then with Cauchy's inequality and $\fnorm{\bbm\Delta_{i,1}+\bbm\Delta_{i,2}}^2=\fnorm{\bbm\Delta_{i,1}}^2+\fnorm{\bbm\Delta_{i,1}}^2$ we know that $\fnorm{\bbm\Delta_{i,1}}+\fnorm{\bbm\Delta_{i,2}}\leq \sqrt{2}\fnorm{\bbm\Delta_{i,1}+\bbm\Delta_{i,2}}\leq\sqrt{2}\varepsilon$, for $i=1,2$. Since $\bbm\Delta_{i,j}/\fnorm{\bbm\Delta_{i,j}}\in \mathcal{W}(r_i,0;p_i)$ we have
$$\begin{aligned}
    &\inner{\frac{1}{T}\sum_{t=1}^T\bm e_t \bm y_{t-1}^{\top}}{\bm A}\\
    =&\inner{\frac{1}{T}\sum_{t=1}^T\bm e_t \bm y_{t-1}^{\top}}{\overline{\bm A}}\\
    &+\inner{\frac{1}{T}\sum_{t=1}^T\bm e_t \bm y_{t-1}^{\top}}{(\bm W_2-\overline{\bm W}_2)\otimes\bm W_1}+\inner{\frac{1}{T}\sum_{t=1}^T\bm e_t \bm y_{t-1}^{\top}}{\overline{\bm W}_2\otimes(\bm W_1-\overline{\bm W}_1)}\\
    =&\inner{\frac{1}{T}\sum_{t=1}^T\bm e_t \bm y_{t-1}^{\top}}{\overline{\bm A}}\\
    &+\inner{\frac{1}{T}\sum_{t=1}^T\bm e_t \bm y_{t-1}^{\top}}{\frac{\bbm\Delta_{2,1}}{\fnorm{\bbm\Delta_{2,1}}}\otimes\bm W_1}\fnorm{\bbm\Delta_{2,1}}+\inner{\frac{1}{T}\sum_{t=1}^T\bm e_t \bm y_{t-1}^{\top}}{\frac{\bbm\Delta_{2,1}}{\fnorm{\bbm\Delta_{2,1}}}\otimes\bm W_1}\fnorm{\bbm\Delta_{2,2}}\\
    &+\inner{\frac{1}{T}\sum_{t=1}^T\bm e_t \bm y_{t-1}^{\top}}{\overline{\bm W}_2\otimes\frac{\bbm\Delta_{1,1}}{\fnorm{\bbm\Delta_{1,1}}}}\fnorm{\bbm\Delta_{1,1}}+\inner{\frac{1}{T}\sum_{t=1}^T\bm e_t \bm y_{t-1}^{\top}}{\overline{\bm W}_2\otimes\frac{\bbm\Delta_{1,2}}{\fnorm{\bbm\Delta_{1,2}}}}\fnorm{\bbm\Delta_{1,2}}\\
    \leq&\max_{\overline{A}\in\overline{\mathcal{V}}(r_1,r_2,0,0)}\inner{\frac{1}{T}\sum_{t=1}^T\bm e_t \bm y_{t-1}^{\top}}{\overline{\bm A}}\\
    &+\xi(r_1,r_2,0,0)\left(\fnorm{\bbm\Delta_{2,1}}+\fnorm{\bbm\Delta_{2,2}}+\fnorm{\bbm\Delta_{1,1}}+\fnorm{\bbm\Delta_{1,2}}\right)\\
    \leq&\max_{\overline{A}\in\overline{\mathcal{V}}(r_1,r_2,0,0)}\inner{\frac{1}{T}\sum_{t=1}^T\bm e_t \bm y_{t-1}^{\top}}{\overline{\bm A}}+2\sqrt{2}\varepsilon\xi(r_1,r_2,0,0).
\end{aligned}$$
Hence,
$$\xi(r_1,r_2,0,0)\leq(1-2\sqrt{2}\varepsilon)^{-1}\max_{\overline{A}\in\overline{\mathcal{V}}(r_1,r_2,0,0)}\inner{\frac{1}{T}\sum_{t=1}^T\bm e_t \bm y_{t-1}^{\top}}{\overline{\bm A}}.$$

Define $R_T(\bm M):=\sum_{t=0}^{T-1}\fnorm{\bm M\bm y_t}^2$ and $S_T(\bm M):=\sum_{t=1}^{T}\inner{\bm e_t}{\bm M\bm y_{t-1}}$ for any $p_1p_2\times p_1p_2$ real matrix $\bm M$ with unit Frobenius norm. By the proof of LemmaS5 in \cite{wang2024LRtsAR}, for any $z_1$ and $z_2 \geq 0$,
$$\mathbb{P}\left[\left\{S_T(\mathbf{M}) \geq z_1\right\} \cap\left\{R_T(\mathbf{M}) \leq z_2\right\}\right] \leq \exp \left(-\frac{z_1^2}{2 \tau^2 \lambda_{\max }\left(\boldsymbol{\Sigma}_{\bm e}\right)z_2}\right).$$

With the derivation of \eqref{eq: RSC prob} and  \eqref{upper bound of RT}, when $T\gtrsim (p_1r_1+p_2r_2+4r_1^2r_2^2)M_2^{-2}\max\left(\tau,\tau^2\right)$, using proper constants $C$ here, we have that 
$$
    \mathbb{P}\left(R_T(\bbm \Delta)\geq \frac{3T}{2}\lambda_{\max }\left(\bbm{\Sigma}_{\bm e}\right) \lambda_{\max }\left(\widetilde{\mathbf{A}} \widetilde{\mathbf{A}}^{\top}\right)\right) \leq 2\exp \{-15(p_1r_1+p_2r_2+4r_1^2r_2^2)\}.
$$
Then, using the pieces above, we have for any $x\geq 0$, 
\begin{equation}
\begin{aligned}
    &\mathbb{P}\left(\xi(r_1,r_2,0,0)\geq x\right)\\
    \leq&\mathbb{P}\left(\max_{\overline{A}\in\overline{\mathcal{V}}(r_1,r_2,0,0)}\inner{\frac{1}{T}\sum_{t=1}^T\bm e_t \bm y_{t-1}^{\top}}{\overline{\bm A}}\geq (1-2\sqrt{2}\varepsilon)x\right)\\
    \leq &\sum_{\overline{A}\in\overline{\mathcal{V}}(r_1,r_2,0,0)}\mathbb{P}\left(S_T(\overline{\bm A})\geq (1-2\sqrt{2}\varepsilon)Tx\right)\\
    \leq &\sum_{\overline{A}\in\overline{\mathcal{V}}(r_1,r_2,0,0)}\mathbb{P}\left(\left\{S_T(\overline{\bm A})\geq (1-2\sqrt{2}\varepsilon)Tx\right\}\bigcap\left\{R_T(\overline{\bm A})\leq \frac{3T}{2}\lambda_{\max }\left(\bbm{\Sigma}_{\bm e}\right) \lambda_{\max }\left(\widetilde{\mathbf{A}} \widetilde{\mathbf{A}}^{\top}\right)\right\}\right)\\
    &+\sum_{\overline{A}\in\overline{\mathcal{V}}(r_1,r_2,0,0)}\mathbb{P}\left(R_T(\overline{\bm A})\geq \frac{3T}{2}\lambda_{\max }\left(\bbm{\Sigma}_{\bm e}\right) \lambda_{\max }\left(\widetilde{\mathbf{A}} \widetilde{\mathbf{A}}^{\top}\right)\right)\\
    \leq &|\overline{\mathcal{V}}(r_1,r_2,0,0)|\left(\exp\left\{-\frac{(1-2\sqrt{2}\varepsilon)^2Tx^2}{3\tau^2\lambda_{\max }^2\left(\bbm{\Sigma}_{\bm e}\right) \lambda_{\max }\left(\widetilde{\mathbf{A}} \widetilde{\mathbf{A}}^{\top}\right)}\right\}+\exp\left\{-15(p_1r_1+p_2r_2+4r_1^2r_2^2)\right\}\right).
\end{aligned}
\end{equation}
Here we take $\varepsilon=0.1$ and $x=4\sqrt{3}\tau \lambda_{\max }\left(\bbm{\Sigma}_{\bm e}\right) \lambda_{\max }^{1/2}\left(\widetilde{\mathbf{A}} \widetilde{\mathbf{A}}^{\top}\right)\sqrt{(2p_1r_1+2p_2r_2+r_1^2+r_2^2)/T}$, then
\begin{equation}
    \begin{aligned}
        &\mathbb{P}\left(\xi(r_1,r_2,0,0)\geq x\right)\\
        \leq &\left(\frac{48}{0.1}\right)^{2(p_1r_1+p_2r_2)+r_1^2+r_2^2}\\
        &\quad \times \left(\exp\left\{-\frac{Tx^2}{6\tau^2\lambda_{\max }^2\left(\bbm{\Sigma}_{\bm e}\right) \lambda_{\max }\left(\widetilde{\mathbf{A}} \widetilde{\mathbf{A}}^{\top}\right)}\right\}+\exp\left\{-15(p_1r_1+p_2r_2+4r_1^2r_2^2)\right\}\right)\\
        \leq& \exp\left\{(7-8)(2p_1r_1+2p_2r_2+r_1^2+r_2^2)\right\}\\
        & +2\exp\left\{7(2p_1r_1+2p_2r_2+r_1^2+r_2^2)-15(p_1r_1+p_2r_2+4r_1^2r_2^2)\right\}\\
        \leq& C\exp \{-C(p_1r_1+p_2r_2)\}.
\end{aligned}
\end{equation}
Define 
$\mathrm{df}_\mathrm{RRMAR}=2p_1r_1+2p_2r_2+r_1^2+r_2^2$ and $M_1=\lambda_{\max }\left(\bbm{\Sigma}_{\bm e}\right)/ \mu_{\min }^{1/2}\left(\mathcal{A}\right)=\lambda_{\max }\left(\bbm{\Sigma}_{\bm e}\right) \lambda_{\max }^{1/2}\left(\widetilde{\mathbf{A}} \widetilde{\mathbf{A}}^{\top}\right)$.
Therefore, when $T\gtrsim (p_1r_1+p_2r_2+4r_1^2r_2^2)M_2^{-2}\max\left(\tau,\tau^2\right)$, we have that with probability at least $1-C\exp \{-C(p_1r_1+p_2r_2)\}$,
$$\xi(r_1,r_2,0,0)\lesssim \tau M_1\sqrt{\frac{\mathrm{df}_\mathrm{RRMAR}}{T}}.$$

Next, we construct the upper bound of $\xi$ when common spaces exist. By Lemma \ref{lemma:splitting SVD}, now $\bm W_2-\overline{\bm W}_2$ can be decomposed as $\bm W_2-\overline{\bm W}_2=\bbm \Delta_{2,1}+\bbm \Delta_{2,2}+\bbm \Delta_{2,3}+\bbm \Delta_{2,4}$, where $\bbm \Delta_{2,1}, \bbm \Delta_{2,2}$ are both rank-$r_2$ with common dimension $d_2$. $\bbm \Delta_{2,3}, \bbm \Delta_{2,4}$ are rank-$r_2$. Moreover, $\inner{\bbm\Delta_{2,j}}{\bbm\Delta_{2,k}}=0$ for any $j,k=1,2,3,4,j\neq k$. Then with Cauchy's inequality and $\fnorm{\sum_{s=1}^{4}\bbm\Delta_{i,s}}^2=\sum_{s=1}^{4}\fnorm{\bbm\Delta_{i,s}}^2$ we know that $\sum_{s=1}^{4}\fnorm{\bbm\Delta_{i,s}}\leq 2\fnorm{\sum_{s=1}^{4}\bbm\Delta_{i,s}}\leq2\varepsilon, i=1,2.$ Similarly, we can decompose each side of the Kronecker product into four parts. The first two parts contain common space while the last two do not. Thus, $\xi(r_1,r_2,d_1,d_2)$ can be upper bounded by the covering of $\mathcal{V}(r_1,r_2,d_1,d_2)$ and $\xi(r_1,r_2,0,0)$ as following:
$$\begin{aligned}
    &\inner{\frac{1}{T}\sum_{t=1}^T\bm e_t \bm y_{t-1}^{\top}}{\bm A}\\
    =&\inner{\frac{1}{T}\sum_{t=1}^T\bm e_t \bm y_{t-1}^{\top}}{\overline{\bm A}}\\
    &+\sum_{s=1}^{4}\inner{\frac{1}{T}\sum_{t=1}^T\bm e_t \bm y_{t-1}^{\top}}{\frac{\bbm\Delta_{2,s}}{\fnorm{\bbm\Delta_{2,s}}}\otimes\bm W_1}\fnorm{\bbm\Delta_{2,s}}\\
    &+\sum_{s=1}^{4}\inner{\frac{1}{T}\sum_{t=1}^T\bm e_t \bm y_{t-1}^{\top}}{\overline{\bm W}_2\otimes\frac{\bbm\Delta_{1,s}}{\fnorm{\bbm\Delta_{1,s}}}}\fnorm{\bbm\Delta_{1,s}}\\
    \leq&\max_{\overline{A}\in\overline{\mathcal{V}}(r_1,r_2,d_1,d_2)}\inner{\frac{1}{T}\sum_{t=1}^T\bm e_t \bm y_{t-1}^{\top}}{\overline{\bm A}}+4\varepsilon\xi(r_1,r_2,0,0)+4\varepsilon\xi(r_1,r_2,d_1,d_2)\\
    \leq&\max_{\overline{A}\in\overline{\mathcal{V}}(r_1,r_2,d_1,d_2)}\inner{\frac{1}{T}\sum_{t=1}^T\bm e_t \bm y_{t-1}^{\top}}{\overline{\bm A}}+4\varepsilon\xi(r_1,r_2,0,0)+4\varepsilon\xi(r_1,r_2,d_1,d_2).
\end{aligned}$$
Therefore,
$$\xi(r_1,r_2,d_1,d_2)\leq (1-4\varepsilon)^{-1}\left(\max_{\overline{A}\in\overline{\mathcal{V}}(r_1,r_2,d_1,d_2)}\inner{\frac{1}{T}\sum_{t=1}^T\bm e_t \bm y_{t-1}^{\top}}{\overline{\bm A}}+4\varepsilon\xi(r_1,r_2,0,0)\right).$$
Similar to the derivation of $\xi(r_1,r_2,0,0)$, we have
$$\begin{aligned}
    & \mathbb{P}\left(\max_{\overline{A}\in\overline{\mathcal{V}}(r_1,r_2,d_1,d_2)}\inner{\frac{1}{T}\sum_{t=1}^T\bm e_t \bm y_{t-1}^{\top}}{\overline{\bm A}}\geq x\right)\\
    \leq &|\overline{\mathcal{V}}(r_1,r_2,d_1,d_2)|\left(\exp\left\{-\frac{Tx^2}{3\tau^2\lambda_{\max }^2\left(\bbm{\Sigma}_{\bm e}\right) \lambda_{\max }\left(\widetilde{\mathbf{A}} \widetilde{\mathbf{A}}^{\top}\right)}\right\}+\exp\left\{-C(p_1r_1+p_2r_2+4r_1^2r_2^2)\right\}\right).
\end{aligned}$$
Define 
$$\mathrm{df}_\mathrm{MARCF}:=p_1(2r_1-d_1)+p_2(2r_2-d_2)+r_1^2+r_2^2.$$ Taking $\varepsilon=0.1$ and $x=C\tau \lambda_{\max }\left(\bbm{\Sigma}_{\bm e}\right) \lambda_{\max }^{1/2}\left(\widetilde{\mathbf{A}} \widetilde{\mathbf{A}}^{\top}\right)\sqrt{\mathrm{df}_\mathrm{MARCF}/T}$ and with \eqref{ineq: V covering}, similarly we have
$$\begin{aligned}
    & \mathbb{P}\left(\max_{\overline{A}\in\overline{\mathcal{V}}(r_1,r_2,d_1,d_2)}\inner{\frac{1}{T}\sum_{t=1}^T\bm e_t \bm y_{t-1}^{\top}}{\overline{\bm A}}\gtrsim \tau M_1\sqrt{\mathrm{df}_\mathrm{MARCF}/T}\right)\leq \exp\left\{-C(p_1r_1+p_2r_2)\right\}
\end{aligned}$$
Combining the upper bound of $\xi(r_1,r_2,0,0)$, we have
$$\xi(r_1,r_2,d_1,d_2)\lesssim \tau M_1\left(\sqrt{\frac{\mathrm{df}_\mathrm{MARCF}}{T}}+\sqrt{\frac{\mathrm{df}_\mathrm{RRMAR}}{T}}\right),$$
with probability at least $1-C\exp\left(-C(p_1r_1+p_2r_2)\right)$. Moreover, since $\mathrm{df}_\mathrm{RRMAR}\leq 2\mathrm{df}_\mathrm{MARCF}$,
$$\xi(r_1,r_2,d_1,d_2)\lesssim \tau M_1\sqrt{\frac{\mathrm{df}_\mathrm{MARCF}}{T}}.$$
\end{proof}

With this upper bound, we verify the condition on the upper bound of $\xi$ in Theorem \ref{theorem: computational convergence}.

\begin{proposition}\label{proposition: small kappa}
Under Assumptions \ref{assumption: spectral redius} and \ref{assumption: sub-Gaussian noise}, if $T\gtrsim (p_1r_1+p_2r_2+4r_1^2r_2^2)M_2^{-2}\max\left(\tau,\tau^2\right)$ as well as $T\gtrsim \kappa^2\underline{\sigma}^{-4}\alpha_\mathrm{RSC}^{-3}\beta_\mathrm{RSS} \tau^2M_1^2(p_1r_1+p_2r_2)$, then we have that, with probability as least $1-C\exp\left(-C(p_1r_1+p_2r_2)\right)$,
$$
\xi^2\lesssim \frac{\phi^4\alpha_\mathrm{RSC}^3}{\kappa^{6}\beta_\mathrm{RSS}}.
$$
\end{proposition}

\begin{proof}[Proof of Proposition \ref{proposition: small kappa}]
By Proposition \ref{prop: Upper bound of xi}, we know that with probability at least $1-C\exp\left(-C(p_1r_1+p_2r_2)\right)$, 
$$\xi(r_1,r_2,d_1,d_2)\lesssim \tau M_1\sqrt{\frac{\mathrm{df}_\mathrm{MARCF}}{T}}.$$
Then, when $T\gtrsim \kappa^2\underline{\sigma}^{-4}\alpha_\mathrm{RSC}^{-3}\beta_\mathrm{RSS} \tau^2M_1^2(p_1r_1+p_2r_2)\gtrsim \kappa^2\underline{\sigma}^{-4}\alpha_\mathrm{RSC}^{-3}\beta_\mathrm{RSS} \tau^2M_1^2\mathrm{df}_\mathrm{MARCF}$, we have
$$\xi^2\lesssim \frac{\tau^2 M_1^2\mathrm{df}_\mathrm{MARCF}}{\kappa^2\underline{\sigma}^{-4}\alpha_\mathrm{RSC}^{-3}\beta_\mathrm{RSS} \tau^2M_1^2\mathrm{df}_\mathrm{MARCF}}=\frac{\phi^4\alpha_\mathrm{RSC}^3}{\kappa^6\beta_\mathrm{RSS}}.$$
\end{proof}

\subsection{Properties of Initialization}\label{sub:Initialization}

Our initialization begins with finding the solution of the reduced-rank least squares problem of RRMAR model introduced in \cite{XiaoRRMAR2024}:
\begin{equation}\label{eq: RRMAR initialization}
   \widetilde{\bm A}_1^\mathrm{RR}, \widetilde{\bm A}_2^\mathrm{RR}:=\argmin_{\mathrm{rank}(\bm A_i)\leq r_i, i=1,2}\frac{1}{2T}\sum_{i=1}^{T}\fnorm{\bm Y_t-\bm A_1\bm Y_{t-1}\bm A_2^\top}^2. 
\end{equation}
Then we let $\widehat{\bm A}_1^\mathrm{RR}$ and $\widehat{\bm A}_2^\mathrm{RR}$ to be the rescaled estimation with equal Frobenius norm. Multiple numerical approaches can be applied to find the optimal solution. For example, the alternating least squares method (RR.LS) and alternating canonical correlation analysis method (RR.CC) proposed in \cite{XiaoRRMAR2024}. In this article, we use our gradient decent algorithm with $d_1=d_2=0$ to obtain the first-stage estimation. Then the solutions $\widehat{\bm A}_1^{\text{RR}}$ and $\widehat{\bm A}_2^{\text{RR}}$ are decomposed to obtain the initial value of $\bm A_1^{(0)}$ and $\bm A_2^{(0)}$. Once we obtain minimizers in \eqref{eq: RRMAR initialization}, we have an upper bound of initialization error.

\begin{proposition}\label{prop: initialization}
Under Assumptions \ref{assumption: spectral redius} and \ref{assumption: sub-Gaussian noise}, when $T\gtrsim (p_1r_1+p_2r_2+4r_1^2r_2^2)M_2^{-2}\max\left(\tau,\tau^2\right)$, with probability at least $1-C\exp\left(-C(p_1r_1+p_2r_2)\right)$, we have
$$
\fnorm{\widehat{\bm A}_2^\mathrm{RR}\otimes \widehat{\bm A}_1^\mathrm{RR}-\bm A_2^*\otimes \bm A_1^*} \lesssim \alpha_\mathrm{RSC}^{-1}\tau M_1\sqrt{\frac{\mathrm{df}_\mathrm{RRMAR}}{T}},
$$
and
$$
\fnorm{\widehat{\bm A}_1^\mathrm{RR}-\bm A_1^*}+\fnorm{\widehat{\bm A}_2^\mathrm{RR}-\bm A_2^*}\lesssim \phi^{-1}\alpha_\mathrm{RSC}^{-1}\tau M_1\sqrt{\frac{\mathrm{df}_\mathrm{RRMAR}}{T}}.
$$
Moreover, together with Assumption \ref{assumption: gmin}, the initialization error satisfies
$$
    \mathrm{dist}^2_{(0)}\lesssim \phi^{2/3} \underline{\sigma}^{-4} g_{\min }^{-4}\alpha_\mathrm{RSC}^{-2}\tau^2 M_1^2\frac{\mathrm{df}_\mathrm{RRMAR}}{T}.
$$
\end{proposition}

\begin{proof}[Proof of Proposition \ref{prop: initialization}]
Let $\widehat{\bm A}^\mathrm{RR}=\widehat{\bm A}_2^\mathrm{RR}\otimes \widehat{\bm A}_1^\mathrm{RR}$. First, we give the error bound of $\widehat{\bm A}^\mathrm{RR}$. Let $\bbm\Delta:=\widehat{\bm A}^\mathrm{RR}-\bm A^*=\widehat{\bm A}_2^\mathrm{RR}\otimes \widehat{\bm A}_1^\mathrm{RR}-\bm A_2^*\otimes \bm A_1^*$. By the optimality of $\widehat{\bm A}^\mathrm{RR}$,
$$
    \frac{1}{2 T} \sum_{t=1}^T\left\|\mathbf{y}_t-\widehat{\bm A}^\mathrm{RR} \mathbf{y}_{t-1}\right\|_2^2 \leq \frac{1}{2 T} \sum_{t=1}^T\left\|\mathbf{y}_t-\mathbf{A}^* \mathbf{y}_{t-1}\right\|_2^2
,$$
we have
$$\frac{1}{2 T} \sum_{t=1}^T\left\|\boldsymbol{\Delta} \mathbf{y}_{t-1}\right\|_2^2 \leq \left\langle\frac{1}{T} \sum_{t=1}^T\bm e_t \mathbf{y}_{t-1}^{\top}, \boldsymbol{\Delta}\right\rangle.$$
Define $a_2:=\V{\widehat{\bm A}_2^\mathrm{RR}}/\fnorm{\widehat{\bm A}_2^\mathrm{RR}}, a_1:=\V{\widehat{\bm A}_1^\mathrm{RR}}\fnorm{\widehat{\bm A}_2^\mathrm{RR}}, a_2^*:=\V{\bm A_2^{*}}, a_1^*:=\V{\bm A_2^{*}}$. By permutation operator, $\bbm \Delta$ can be decomposed as:
$$\begin{aligned}
    \bbm \Delta=&\mathcal{P}^{-1}\mathcal{P}\left(\widehat{\bm A}_2^\mathrm{RR}\otimes \widehat{\bm A}_1^\mathrm{RR}-\bm A_2^*\otimes \bm A_1^*\right)\\
    =&\mathcal{P}^{-1}\left(a_2a_1^\top-a_2^*a_1^{*\top}\right)\\
    =&\mathcal{P}^{-1}\left( a_2(a_1^\top-a_2^\top a_2^*a_1^{*\top})+(a_2a_2^\top-\bm I_{p_2^2} )a_2^*a_1^{*\top}\right)\\
    =&\mathcal{P}^{-1}\left( a_2(a_1^\top-a_2^\top a_2^*a_1^{*\top})\right)+\mathcal{P}^{-1}\left((a_2a_2^\top-\bm I_{p_2^2} )a_2^*a_1^{*\top}\right)\\
    =&\widehat{\bm A}_2^\mathrm{RR}\otimes \left(\widehat{\bm A}_1^\mathrm{RR}-\frac{\inner{\widehat{\bm A}_2^\mathrm{RR}}{\bm A_2^*}}{\fnorm{\widehat{\bm A}_2^\mathrm{RR}}^2}\bm A_1^*\right)+\left(\frac{\inner{\widehat{\bm A}_2^\mathrm{RR}}{\bm A_2^*}}{\fnorm{\widehat{\bm A}_2^\mathrm{RR}}^2}\widehat{\bm A}_2^\mathrm{RR}-\bm A_2^*\right)\otimes \bm A_1^*\\
    :=&\widehat{\bm A}_2^\mathrm{RR}\otimes \bm W_1+\bm W_2\otimes \bm A_1^*.
\end{aligned}$$
From definition we see that $\bm W_1$ and $\bm W_2$ are rank-$2r_1$ and rank-$2r_2$ matrices, respectively. Meanwhile, $\inner{\widehat{\bm A}_2^\mathrm{RR}\otimes \bm W_1}{\bm W_2\otimes \bm A_1^*}=0$. Hence, with Cauchy-Schwarz inequality,
$$\begin{aligned}
    &\left\langle\frac{1}{T} \sum_{t=1}^T\bm e_t \mathbf{y}_{t-1}^{\top}, \boldsymbol{\Delta}\right\rangle\\
    =&\left(\fnorm{\widehat{\bm A}_2^\mathrm{RR}\otimes \bm W_1}+\fnorm{\bm W_2\otimes \bm A_1^*}\right)\\
    &\times \left(\inner{\frac{1}{T}\sum_{t=1}^T\bm e_t \mathbf{y}_{t-1}^{\top}}{\frac{\widehat{\bm A}_2^\mathrm{RR}}{\fnorm{\widehat{\bm A}_2^\mathrm{RR}}}\otimes \frac{\bm W_1}{\fnorm{\bm W_1}}}+\inner{\frac{1}{T}\sum_{t=1}^T\bm e_t \mathbf{y}_{t-1}^{\top}}{\frac{\bm W_2}{\fnorm{\bm W_2}}\otimes \frac{\bm A_1^*}{\fnorm{\bm A_1^*}}}\right)\\
    \leq& \left(\fnorm{\widehat{\bm A}_2^\mathrm{RR}\otimes \bm W_1}+\fnorm{\bm W_2\otimes \bm A_1^*}\right)(\xi(2r_1,r_2,0,0)+\xi(r_1,2r_2,0,0))\\
    \leq&\sqrt{2}\fnorm{\bbm \Delta} (\xi(2r_1,r_2,0,0)+\xi(r_1,2r_2,0,0)).
\end{aligned}$$
By Proposition \ref{prop: Upper bound of xi}, when $T\gtrsim (p_1r_1+p_2r_2+4r_1^2r_2^2)M_2^{-2}\max\left(\tau,\tau^2\right)$, we have that with probability as least $1-C\exp\left(-C(p_1r_1+p_2r_2)\right)$, 
    $$\xi(2r_1,r_2,0,0)+\xi(r_1,2r_2,0,0)\lesssim \tau M_1\sqrt{\frac{\mathrm{df}_\mathrm{RRMAR}}{T}}.$$
By Proposition \ref{prop: RSCRSS}, with probability at least $1-C\exp\left(-C(p_1r_1+p_2r_2)\right)$, 
$$
\frac{\alpha_\mathrm{RSC}}{2}\fnorm{\bbm \Delta}^2\leq \frac{1}{2 T} \sum_{t=1}^T\left\|\boldsymbol{\Delta} \mathbf{y}_{t-1}\right\|_2^2.
$$
Combining the pieces together, we have when $T\gtrsim (p_1r_1+p_2r_2+4r_1^2r_2^2)M_2^{-2}\max\left(\tau,\tau^2\right)$, with probability as least $1-C\exp\left(-C(p_1r_1+p_2r_2)\right)$, 
$$
\fnorm{\widehat{\bm A}_2^\mathrm{RR}\otimes \widehat{\bm A}_1^\mathrm{RR}-\bm A_2^*\otimes \bm A_1^*} \lesssim \alpha_\mathrm{RSC}^{-1}\tau M_1\sqrt{\frac{\mathrm{df}_\mathrm{RRMAR}}{T}}.
$$
Since $\|\widehat{\bm A}_1^\mathrm{RR}\|_\mathrm{F}=\|\widehat{\bm A}_2^\mathrm{RR}\|_\mathrm{F}$, by \eqref{A1A2 and kron and regularizer1},
$$
\fnorm{\widehat{\bm A}_1^\mathrm{RR}-\bm A_1^*}+\fnorm{\widehat{\bm A}_2^\mathrm{RR}-\bm A_2^*}\lesssim \phi^{-1}\alpha_\mathrm{RSC}^{-1}\tau M_1\sqrt{\frac{\mathrm{df}_\mathrm{RRMAR}}{T}}.
$$

Finally, with Assumption \ref{assumption: gmin}, employing an argument analogous to Lemma B.4 in \cite{wang2023commonFactor}, we bound the initialization error as
$$
    \mathrm{dist}^2_{(0)}\lesssim \phi^{8/3} \underline{\sigma}^{-4} g_{\min }^{-4}\left(\fnorm{\widehat{\bm A}_1^\mathrm{RR}-\bm A_1^*}^2+\fnorm{\widehat{\bm A}_2^\mathrm{RR}-\bm A_2^*}^2\right)\lesssim \phi^{2/3} \underline{\sigma}^{-4} g_{\min }^{-4}\alpha_\mathrm{RSC}^{-2}\tau^2 M_1^2\frac{\mathrm{df}_\mathrm{RRMAR}}{T}.
$$
\end{proof}

Therefore, when $T$ is large, the initialization condition for local convergence in Theorem \ref{theorem: computational convergence} holds with high probability.

\begin{proposition}\label{prop: sample_size_initialization}
    Under Assumptions \ref{assumption: spectral redius}, \ref{assumption: sub-Gaussian noise} and \ref{assumption: gmin}, if $T\gtrsim g_{\min}^{-4}\kappa^2\underline{\sigma}^{-4}\alpha_\mathrm{RSC}^{-3}\beta_\mathrm{RSS} \tau^2M_1^2(p_1r_1+p_2r_2)$ and $T\gtrsim (p_1r_1+p_2r_2+4r_1^2r_2^2)M_2^{-2}\max\left(\tau,\tau^2\right)$, with probability at least $1-C\exp\left(-C(p_1r_1+p_2r_2)\right)$, we have
    $$
        \mathrm{dist}^2_{(0)}\lesssim \frac{\phi^{2/3}\alpha_\mathrm{RSC}}{\kappa^2\beta_\mathrm{RSS}}.
    $$
\end{proposition}

\begin{proof}[Proof of Proposition \ref{prop: sample_size_initialization}]
By Proposition \ref{prop: initialization}, we know that with probability at least $1-C\exp\left(-C(p_1r_1+p_2r_2)\right)$,
$$
    \mathrm{dist}^2_{(0)}\lesssim \phi^{2/3} \underline{\sigma}^{-4} g_{\min }^{-4}\alpha_\mathrm{RSC}^{-2}\tau^2 M_1^2\frac{\mathrm{df}_\mathrm{RRMAR}}{T}.
$$
Thus when $T\gtrsim g_{\min}^{-4}\kappa^2\underline{\sigma}^{-4}\alpha_\mathrm{RSC}^{-3}\beta_\mathrm{RSS} \tau^2M_1^2(p_1r_1+p_2r_2)\gtrsim g_{\min}^{-4}\kappa^2\underline{\sigma}^{-4}\alpha_\mathrm{RSC}^{-3}\beta_\mathrm{RSS} \tau^2M_1^2\mathrm{df}_\mathrm{RRMAR}$, we have $\mathrm{dist}^2_{(0)}\leq C_D\kappa^{-2}\phi^{2/3}\alpha_\mathrm{RSC}\beta_\mathrm{RSS}^{-1}$.
\end{proof}

\subsection{Auxiliary Lemmas}

The first lemma is Hanson--Wright inequality and can be found in high-dimentional statistics monograph \citep{Wainwright2019}.
\begin{lemma}\label{Hanson-Wright}\citep{Wainwright2019}
    Given random variables $\left\{X_i\right\}_{i=1}^n$ and a positive semidefinite matrix $\mathbf{Q} \in \mathcal{S}_{+}^{n \times n}$, consider the random quadratic form
$$
Z=\sum_{i=1}^n \sum_{j=1}^n \mathbf{Q}_{i j} X_i X_j .
$$
If the random variables $\left\{X_i\right\}_{i=1}^n$ are i.i.d. with mean zero, unit variance, and $\sigma$-sub-Gaussian, then there are universal constants $\left(c_1, c_2\right)$ such that
$$
\mathbb{P}[|Z -\mathbb{E}Z|\geq\sigma t] \leq 2 \exp \left\{-\min \left(\frac{c_1 t}{\|\mathbf{Q}\|_2}, \frac{c_2 t^2}{\|\mathbf{Q}\|_{\mathrm{F}}^2}\right)\right\},
$$
where $\|\mathbf{Q}\|_2$ and $\|\mathbf{Q}\|_{\mathrm{F}}$ denote the operator nd Frobenius norms, respectively.
\end{lemma}

The second lemma establishs the covering number of $\mathcal{S}(r_1,r_2,p_1,p_2)$ introduced in the proof of Proposition \ref{prop: RSCRSS}.
\begin{lemma}\label{lemma:covering_KP_difference}
    Define the set of unit-norm differences of low-rank Kronecker products as:
    \begin{equation}
    \begin{aligned}
        \mathcal{S}(r_1,r_2,p_1,p_2) := \Big\{\bm M:~ &\bm M=\bm A_2\otimes \bm A_1-\bm A_2'\otimes \bm A_1',\\
        &\bm A_i, \bm A_i'\in \mathbb{R}^{p_i\times p_i},~\mathrm{rank}(\bm A_i),~\mathrm{rank}(\bm A_i')\leq r_i, ~ \fnorm{\bm M}=1, ~ i=1,2 \Big\}.
    \end{aligned}
\end{equation}
Let $\overline{\mathcal{S}}(r_1,r_2,p_1,p_2)$ be an $\varepsilon$-net of $\mathcal{S}(r_1,r_2,p_1,p_2)$ with respect to the Frobenius norm, where $\varepsilon\in (0,1)$. Then,
\begin{equation}
    \left|\overline{\mathcal{S}}(r_1,r_2,p_1,p_2)\right|\leq \left(\frac{15}{\varepsilon}\right)^{4p_1r_1+4p_2r_2+16r_1^2r_2^2}.
\end{equation}
\end{lemma}

\begin{proof}[Proof of Lemma \ref{lemma:covering_KP_difference}]
    For brevity, given $(r_1,r_2,p_1,p_2)$, we use $\mathcal{S}$ and $\overline{\mathcal{S}}$ to represent $\mathcal{S}(r_1,r_2,p_1,p_2)$ and $\overline{\mathcal{S}}(r_1,r_2,p_1,p_2)$, respectively. Since the proof utilizes tensor operations as techniques, we give some basic notations and tensor algebra relevant to the proof for clearity. We use 4-dimensional tensors to illustrate the concepts.

    Our definitions align with \cite{kolda2009tensor}. For any tensor $\bbm{\mathcal{X}}\in\mathbb{R}^{d_1\times d_2\times d_3\times d_4}$, denote its $(i_1,i_2,i_3,i_4)$-th entry as $\mathcal{X}_{i_1,i_2,i_3,i_4}$. The Frobenius norm is defined by
    \begin{equation}
        \|\bbm{\mathcal{X}}\|_\mathrm{F} = \sqrt{\sum_{i_1=1}^{d_1}\sum_{i_2=1}^{d_2}\sum_{i_3=1}^{d_3}\sum_{i_4=1}^{d_4} \mathcal{X}_{i_1,i_2,i_3,i_4}^2}.
    \end{equation}
    The mode-$k$ matricization of a tensor $\bbm{\mathcal{X}}$, denoted by $\bbm{\mathcal{X}}_{(k)}$, is a $d_k\times (d_1d_2d_3d_4/d_k)$ matrix. Specifically, $\mathcal{X}_{i_1,i_2,i_3,i_4}$ maps to matrix element $\left(i_k, j\right)$, where
    $$
        j=1+\sum_{\substack{n=1 \\ n \neq k}}^4\left(i_n-1\right) J_n \quad \text { with } \quad J_n=\prod_{\substack{m=1 \\ m \neq k}}^{n-1} d_m .
    $$ 
    For a matrix $\bm U\in\mathbb{R}^{p_1\times d_1}$, the 1-mode product, denoted by $\bbm{\mathcal{X}}\times_1 \bm U$, is of size $p_1\times d_2\times d_3\times d_4$. It is defined by
    $$
        (\bbm{\mathcal{X}}\times_1 \bm U)_{i',j,k,l}=\sum_{i=1}^{d_1} \bm U_{i',i} \mathcal{X}_{i,j,k,l},
    $$
    The 2,3 and 4-mode products are defined analogously. Let $\bbm{\mathcal{Y}}=\bbm{\mathcal{X}}\times_1 \bm U_1\times_2 \bm U_2\times_3 \bm U_3\times_4 \bm U_4$ with $\bm U_k\in\mathbb{R}^{p_k\times d_k}$, its $(i',j',k',l')$-th entry can be written as
    \begin{equation}\label{eq:elements_mode_product}
        \mathcal{Y}_{i',j',k',l'}=\sum_{i=1}^{d_1}\sum_{j=2}^{d_2}\sum_{k=3}^{d_3}\sum_{l=4}^{d_4} \mathcal{X}_{i,j,k,l}(\bm U_1)_{i',i}(\bm U_2)_{j',j}(\bm U_3)_{k',k}(\bm U_4)_{l',l}.
    \end{equation}
    An important property is that
    \begin{equation}\label{eq:matricization_mode_product}
        \bbm{\mathcal{Y}}_{(k)}=\bm U_k\bbm{\mathcal{X}}_{(k)}\left(\bm U_4\otimes \dots\otimes \bm U_{(k+1)}\otimes \bm U_{(k-1)}\otimes \dots \otimes \bm U_1\right)^\top.       
    \end{equation}

    Next, we introduce the concept of Tucker rank. Defining $r_k=\text{rank}(\bbm{\mathcal{X}}_{(k)})$ for $k=1,2,3,4$, the Tucker rank of $\bbm{\mathcal{X}}$ is given by $(r_1,r_2,r_3,r_4)$. If $\bbm{\mathcal{X}}$ is low-Tucker-rank, i.e., $r_k\leq \bar{r}_k$ for some $\bar{r}_k$ and $k=1,2,3,4$, it can be decomposed as
    \begin{equation}
        \bbm{\mathcal{X}} = \bbm{\mathcal{G}} \times_1 \bm U_1\times_2 \bm U_2\times_3 \bm U_3\times_4 \bm U_4,
    \end{equation}
    where $\bm U_k\in \mathbb{R}^{d_k\times \bar{r}_k}$ having orthonormal columns and $\bbm{\mathcal{G}}\in\mathbb{R}^{\bar{r}_1\times \bar{r}_2\times \bar{r}_3\times \bar{r}_4}$. The decomposition can be obtained by HOSVD, see \cite{de2000HOSVD} for more details.

    Now we start our proof. We first validate that Kronecker products can be viewed as 4-dimensional tensors. For any matrix $\bm M\in\mathbb{R}^{p_2p_1\times p_2p_1}$, define the rearrangement operator $\mathcal{T}:\mathbb{R}^{p_2p_1\times p_2p_1}\to \mathbb{R}^{p_2\times p_1\times p_2\times p_1}$. The operation is defined element-wise as follows: if $\bm M=\bm M_2\otimes \bm M_1$ where $\bm M_i\in\mathbb{R}^{p_i\times p_i}$ for $i=1,2$, then
    \begin{equation}
        \left[\mathcal{T}(\bm M_2\otimes \bm M_1)\right]_{i_1,j_1,i_2,j_2}=\left(\bm M_2\right)_{i_1,i_2}\left(\bm M_1\right)_{j_1,j_2},~ \text{for all}~ i_1,i_2=1\ldots,p_2, ~\text{and}~ j_1,j_2=1,\ldots,p_1.
    \end{equation}
    Since $\mathcal{T}$ is a bijection, the inverse operator $\mathcal{T}^{-1}$ is well-defined. Furthermore, $\mathcal{T}$ and $\mathcal{T}^{-1}$ are linear isometry with respect to the Frobenius norm. Therefore, we find the $\varepsilon$-net of $\mathcal{S}$ by applying $\mathcal{T}^{-1}$ on the $\varepsilon$-net of $\mathcal{T}(\mathcal{S}):=\{ \mathcal{T}(\bm M): \bm M\in \mathcal{S}\}$. Denote $\overline{\mathcal{T}(\mathcal{S})}$ to be the $\varepsilon$-net $\mathcal{T}(\mathcal{S})$ with respect to the tensor Frobenius norm, we have $|\overline{\mathcal{S}}|=|\overline{\mathcal{T}(\mathcal{S})}|$.

    It remains to characterize the structure of $\mathcal{T}(\mathcal{S})$ and bound the covering number of it. For any $\bm M \in \mathcal{S}$, it is decomposed as $\bm M=\bm A_2\otimes\bm A_1-\bm A_2'\otimes \bm A_1'$, where $\bm A_i, \bm A_i'\in\mathbb{R}^{p_i\times p_i}$ with $\text{rank}(\bm A_i)$ and $\text{rank}(\bm A_i')\leq r_i$. Define their SVD components as $\bm A_i=\bm U_i\bm S_i\bm V_i^\top$ and $\bm A_i'=\bm U_i'\bm S_i'\bm V_i'^\top$, where $\bm U_i$, $\bm U_i'$, $\bm V_i$, and $\bm V_i'\in\mathbb{O}^{p_i\times r_i}$, for $i=1,2$. Denote $\sigma_{k}$ and $\gamma_{k}$ to be the $k$-th largest singular value of $\bm A_2$ and $\bm A_1$, respectively. Applying $\mathcal{T}$ on $\bm A_2\otimes \bm A_1$, the $(i_1,j_1,i_2,j_2)$-th entry is
    \begin{equation}\label{eq:tensor_KP}
        \begin{aligned}
            &\left[\mathcal{T}(\bm A_2\otimes \bm A_1)\right]_{i_1,j_1,i_2,j_2}\\
            =&\left(\sum_{k=1}^{r_2}\sigma_{k}(\bm U_2)_{i_1,k}(\bm V_2)_{i_2,k}\right)\left(\sum_{l=1}^{r_1}\gamma_{l}(\bm U_1)_{j_1,l}(\bm V_1)_{j_2,l}\right)\\
            =&\sum_{k=1}^{r_2}\sum_{l=1}^{r_1}\sigma_{k}\gamma_{l}(\bm U_2)_{i_1,k}(\bm V_2)_{i_2,k}(\bm U_1)_{j_1,l}(\bm V_1)_{j_2,l}\\
            =&\sum_{k=1}^{r_2}\sum_{l=1}^{r_1}\sum_{m=1}^{r_2}\sum_{n=1}^{r_1}\sigma_{k}\gamma_{l}1_{\{m=k\}}1_{\{n=l\}}(\bm U_2)_{i_1,k}(\bm V_2)_{i_2,m}(\bm U_1)_{j_1,l}(\bm V_1)_{j_2,n},
        \end{aligned}
    \end{equation}
    where $1_{\{\cdot\}}$ is the indicator function. It corresponds to the form of mode product given in \eqref{eq:elements_mode_product}, which implies that
    \begin{equation}
        \mathcal{T}(\bm A_2\otimes \bm A_1)=\bbm{\mathcal{C}}\times_1 \bm U_2\times_2 \bm U_1\times_3 \bm V_2\times_4 \bm V_1,
    \end{equation}
    where $\bbm{\mathcal{C}}\in\mathbb{R}^{r_2\times r_1\times r_2\times r_1}$ is a low-dimensional core tensor being defined by the last line of \eqref{eq:tensor_KP}. Then, by matricization with \eqref{eq:matricization_mode_product} and checking the ranks of four modes, we know that $\mathcal{T}(\bm A_2\otimes \bm A_1)$ is low-Tucker-rank with ranks of four modes being at most $(r_2,r_1,r_2,r_1)$. Similar argument holds analogously for $\mathcal{T}(\bm A_2'\otimes \bm A_1')$. Therefore, by the property of matrix ranks, the Tucker rank of $\mathcal{T}(\bm M)=\mathcal{T}(\bm A_2\otimes \bm A_1)-\mathcal{T}(\bm A_2'\otimes \bm A_1')$ are at most $(2r_2,2r_1,2r_2,2r_1)$. Applying Lemma \ref{lemma:covering_tensors}, we have that
    \begin{equation}
        |\overline{\mathcal{S}}|=|\overline{\mathcal{T}(\mathcal{S})}|\leq \left(\frac{15}{\varepsilon}\right)^{4p_1r_1+4p_2r_2+16r_1^2r_2^2}.
    \end{equation}

\end{proof}

The next lemma is some conclusions on splitting martices with common dimensions. Its goal is to show that we can decompose the sum of two martices equipped with common dimensions as the sum of four matrices. Each of them are prependicular to the others, two of them are low-rank without common dimensions and the other two are equipped with the same common dimensions.
\begin{lemma}\label{lemma:splitting SVD}
    Suppose that there are two $p\times p$ rank-$r$ matrices $\bm W_1$ and $\bm W_2$. Then,
    \begin{enumerate}
        \item There exists two $p\times p$ rank-$r$ matrices $\widetilde{\bm W}_1, \widetilde{\bm W}_2$, such that $\bm W_1+\bm W_2=\widetilde{\bm W}_1+\widetilde{\bm W}_2$ and $\inner{\widetilde{\bm W}_1}{\widetilde{\bm W}_2}=0$.
        \item Suppose that now $\bm W_1$ and $\bm W_2$ both have common dimension $d$. That is,
        Define 
            $$\begin{aligned}\mathcal{W}'(r, d ; p)=\left\{\mathbf{W} \in \mathbb{R}^{p \times p}: \mathbf{W}=[\mathbf{C}~\mathbf{R}] \mathbf{D}[\mathbf{C}~\mathbf{P}]^{\top}, \mathbf{C} \in\mathbb{O}^{p \times d}, \mathbf{R}, \mathbf{P} \in \mathbb{O}^{p \times(r-d)},\right.\\
            \left. \inner{\bm C}{\bm R}=\inner{\bm C}{\bm P}=0\right\}\end{aligned}$$
        to be the model space of MARCF model with arbitrary scale, 
        $\bm W_1, \bm W_2\in\mathcal{W}'(r,d;p)$. Then there exist $\widetilde{\bm W}_1, \widetilde{\bm W}_2\in \mathcal{W}'(r,d;p)$ and $\widetilde{\bm W}_3, \widetilde{\bm W}_4\in \mathcal{W}'(r,0;p)$, such that $\bm W_1+\bm W_2=\widetilde{\bm W}_1+\widetilde{\bm W}_2+\widetilde{\bm W}_3+\widetilde{\bm W}_4$ and $\inner{\widetilde{\bm W}_j}{\widetilde{\bm W}_k}=0$ for every $j,k=1,2,3,4,j\neq k$.
    \end{enumerate}
\end{lemma}
\begin{proof}[Proof of Lemma \ref{lemma:splitting SVD}]
    For the first part of the lemma, by SVD decomposition we know that there exists $p\times r$ matrices $\bm U_1, \bm U_2, \bm V_1, \bm V_2$ with mutually orthogonal columns, such that $\bm U_1^\top\bm U_1=\bm I_{r_1}, \bm W_1=\bm U_1\bm V_1^\top$ and  $\bm W_2=\bm U_2\bm V_2^\top$. Let
    $\widetilde{\bm W}_1=\bm U_1(\bm V_1^\top+\bm U_1^\top\bm U_2\bm V_2^\top)$ and $\widetilde{\bm W}_2=(\bm I-\bm U_1\bm U_1^\top)\bm U_2\bm V_2^\top$, then $\widetilde{\bm W}_1$ and $\widetilde{\bm W}_2$ satisfy the conditions.

    For the second part, decompose and write together $\bm W_1+\bm W_2$ gives
    $$\begin{aligned}
        \bm W_1+\bm W_2=[\bm C_1~\bm R_1~\bm C_2~\bm R_2]\begin{bmatrix}
            \bm D_1 & \bm O\\
            \bm O & \bm D_2
        \end{bmatrix}
        [\bm C_1~\bm P_1~\bm C_2~\bm P_2]^\top
    \end{aligned}.$$
    By applying Gram-Schmidt orthogonalization procedure (or equivalently, QR decomposition), we can obtain two sets of orthogonal basis,
    $$[\bm C_1~\widetilde{\bm R}_1~\widetilde{\bm C}_2~\widetilde{\bm R}_2]\in\mathbb{O}^{p\times(r_1+r_2)},\quad [\bm C_1~\widetilde{\bm P}_1~\widetilde{\bm C}_2~\widetilde{\bm P}_2]\in\mathbb{O}^{p\times(r_1+r_2)}.$$
    Then 
    $$\begin{aligned}
        \bm W_1+\bm W_2=&[\bm C_1~\widetilde{\bm R}_1~\widetilde{\bm C}_2~\widetilde{\bm R}_2]\begin{bmatrix}
            \widetilde{\bm D}_1 & \widetilde{\bm D}_3\\
            \widetilde{\bm D}_4 & \widetilde{\bm D}_2
        \end{bmatrix}
        [\bm C_1~\widetilde{\bm P}_1~\widetilde{\bm C}_2~\widetilde{\bm P}_2]^\top\\
        =&[\bm C_1~\widetilde{\bm R}_1]\widetilde{\bm D}_1[\bm C_1~\widetilde{\bm P}_1]^\top+[\widetilde{\bm C}_2~\widetilde{\bm R}_2]\widetilde{\bm D}_2[\widetilde{\bm C}_2~\widetilde{\bm P}_2]^\top\\
        &+[\bm C_1~\widetilde{\bm R}_1]\widetilde{\bm D}_3[\widetilde{\bm C}_2~\widetilde{\bm P}_2]^\top+[\widetilde{\bm C}_2~\widetilde{\bm R}_2]\widetilde{\bm D}_4[\bm C_1~\widetilde{\bm P}_1]^\top\\
        :=&\widetilde{\bm W}_1+\widetilde{\bm W}_2+\widetilde{\bm W}_3+\widetilde{\bm W}_4.
    \end{aligned}$$
    It can be directly verified that $\inner{\widetilde{\bm W}_j}{\widetilde{\bm W}_k}=0$ for every $j,k=1,2,3,4,j\neq k$.
\end{proof}

The next lemma gives the covering number of low rank martices with common column and row spaces.
\begin{lemma}\label{lemma: common covering}
    (Lemma B.6, \cite{wang2023commonFactor})
    Define 
    $$\begin{aligned}\mathcal{W}(r, d ; p)=\left\{\mathbf{W} \in \mathbb{R}^{p \times p}: \mathbf{W}=[\mathbf{C}~\mathbf{R}] \mathbf{D}[\mathbf{C}~\mathbf{P}]^{\top}, \mathbf{C} \in\mathbb{O}^{p \times d}, \mathbf{R}, \mathbf{P} \in \mathbb{O}^{p \times(r-d)},\right.\\
    \left. \inner{\bm C}{\bm R}=\inner{\bm C}{\bm P}=0, \textrm{ and }\|\mathbf{W}\|_{\mathrm{F}}=1\right\}.
    \end{aligned}$$
    Let $\overline{\mathcal{W}}(r, d ; p)$ be an $\epsilon$-net of $\mathcal{W}(r, d ; p)$, where $\epsilon \in(0,1]$. Then
    $$
    |\overline{\mathcal{W}}(r, d ; p)| \leq\left(\frac{24}{\epsilon}\right)^{p(2 r-d)+r^2}.
    $$
    
\end{lemma}

Finally, the following lemma establishes the covering number bound for the set of low-Tucker-rank tensors.
\begin{lemma}\label{lemma:covering_tensors}
    (Lemma 2, \cite{rauhut2017low})
    Define
    $$
    \mathcal{H}=\left\{\bbm{\mathcal{X}} \in \mathbb{R}^{n_1 \times n_2 \times \cdots \times n_N}: \operatorname{rank}(\bbm{\mathcal{X}}_{(k)}) \leq r_k, k=1,2,...,N, \|\bbm{\mathcal{X}}\|_\mathrm{F}=1\right\}.
    $$
    For any $\varepsilon\in(0,1)$, the covering number of $\mathcal{H}$ with respect to the Frobenius norm satisfies:
    $$
    \mathcal{N}\left(\mathcal{H},\|\cdot\|_\mathrm{F}, \varepsilon\right) \leq(3(N+1) / \varepsilon)^{r_1 r_2 \cdots r_N+\sum_{k=1}^N n_k r_k}.
    $$
\end{lemma}


\newpage

\section{Consistency of Rank Selection}\label{sec:rank_selection}

\noindent\textit{Proof of Theorem \ref{theorem: selecting r}:} We focus on the consistency for selecting $r_1$, as the result for $r_2$ can be developed analogously. With $T\gtrsim (p_1r_1+p_2r_2+4r_1^2r_2^2)M_2^{-2}\max\left(\tau,\tau^2\right)$, by Proposition \ref{prop: initialization},
$$\fnorm{\widehat{\bm A}_1^\mathrm{RR}(\bar{r}_1)-\bm A_1^*}+\fnorm{\widehat{\bm A}_2^\mathrm{RR}(\bar{r}_2)-\bm A_2^*}\lesssim \phi^{-1}\alpha_\mathrm{RSC}^{-1}\tau M_1\sqrt{\frac{p_1+p_2}{T}}$$
with probability approaching $1$ as $T\to\infty$ and $p_1,p_2\to\infty$. Since $\widehat{\bm A}_1^\mathrm{RR}-\bm A_1^*$ is rank-$(\bar{r}_1+r_1)$, by the fact that $L_\infty$ norm is smaller than $L_2$ norm and Mirsky's singular value inequality \cite{Mirsky1960},
$$\begin{aligned}
    \max_{1\leq j\leq \bar{r}_1+r_1}|\sigma_j(\widehat{\bm A}_1^\mathrm{RR}(\bar{r}_1))-\sigma_j(\bm A_1^*)|^2 \leq& \sum_{j=1}^{\bar{r}_1+r_1}\left(\sigma_j(\widehat{\bm A}_1^\mathrm{RR}(\bar{r}_1))-\sigma_j(\bm A_1^*)\right)^2\\
    \leq& \sum_{j=1}^{\bar{r}_1+r_1}\sigma^2_j(\widehat{\bm A}_1^\mathrm{RR}(\bar{r}_1)-\bm A_1^*)\\
    =&\fnorm{\widehat{\bm A}_1^\mathrm{RR}(\bar{r}_1)-\bm A_1^*}^2\\
    \lesssim& \phi^{-2}\alpha_\mathrm{RSC}^{-2}\tau^2 M_1^2\frac{p_1+p_2}{T}.
\end{aligned}$$ 
Then, $\forall j=1,2,...,\bar{r}_1$, 
\begin{equation}
    |\sigma_j(\widehat{\bm A}_1^\mathrm{RR}(\bar{r}_1))-\sigma_j(\bm A_1^*)|= O(\phi^{-1}\alpha_\mathrm{RSC}^{-1}\tau M_1\sqrt{(p_1+p_2)/T}).
\end{equation}
Next, we show that as $T,p_1,p_2\to\infty$, the ratio $(\sigma_{j+1}(\widehat{\bm A}_1^\mathrm{RR}(\bar{r}_1))+s(p_1,p_2,T))/(\sigma_{j}(\widehat{\bm A}_1^\mathrm{RR}(\bar{r}_1))+s(p_1,p_2,T))$ achieves its minimum at $j=r_1$. For $j>r_1$, $\sigma_j(\bm A_1^*)=0$ and 
\begin{equation}
    \sigma_j(\widehat{\bm A}_1^\mathrm{RR}(\bar{r}_1))=O(\phi^{-1}\alpha_\mathrm{RSC}^{-1}\tau M_1\sqrt{\mathrm{df}_\mathrm{RRMAR}/T})=o(s(p_1,p_2,T)).
\end{equation}

Therefore, we have
$$
    \frac{\sigma_{j+1}(\widehat{\bm A}_1^\mathrm{RR}(\bar{r}_1))+s(p_1,p_2,T)}{\sigma_{j}(\widehat{\bm A}_1^\mathrm{RR}(\bar{r}_1))+s(p_1,p_2,T)}\to 1.
$$
For $j<r_1$,
$$\begin{aligned}
    &\lim_{\substack{T\to\infty\\p_1,p_2\to\infty}}\frac{\sigma_{j+1}(\widehat{\bm A}_1^\mathrm{RR}(\bar{r}_1))+s(p_1,p_2,T)}{\sigma_{j}(\widehat{\bm A}_1^\mathrm{RR}(\bar{r}_1))+s(p_1,p_2,T)}\\
    =&\lim_{\substack{T\to\infty\\p_1,p_2\to\infty}}\frac{\sigma_{j+1}(\bm A_1^*)+o(s(p_1,p_2,T))+s(p_1,p_2,T)}{\sigma_{j}(\bm A_1^*)+o(s(p_1,p_2,T))+s(p_1,p_2,T)}\\
    =&\lim_{p_1,p_2\to\infty}\frac{\sigma_{j+1}(\bm A_1^*)}{\sigma_{j}(\bm A_1^*)}\leq 1.
\end{aligned}$$
For $j=r_1$,
$$\begin{aligned}
    \frac{\sigma_{j+1}(\widehat{\bm A}_1^\mathrm{RR}(\bar{r}_1))+s(p_1,p_2,T)}{\sigma_{j}(\widehat{\bm A}_1^\mathrm{RR}(\bar{r}_1))+s(p_1,p_2,T)}=&\frac{o(s(p_1,p_2,T))+s(p_1,p_2,T)}{\sigma_{r_1}(\bm A_1^*)+o(s(p_1,p_2,T))+s(p_1,p_2,T)}\\
    \to &\frac{s(p_1,p_2,T)}{\sigma_{r_1}(\bm A_1^*)}\\
    =& o\left(\min_{1\leq j\leq r_1-1}\frac{\sigma_{j+1}(\bm A_1^*)}{\sigma_j(\bm A_1^*)}\right).
\end{aligned}$$
Therefore, when $T,p_1,p_2\to\infty$, the ratio will finally achieve its minimum at $j=r_1$, with a probability approaching one.

\end{document}